\RequirePackage{ifthen}

\newcommand{\typeof}{1} %

\newcommand{\cmnts}{0} %

\newcommand{\longv}[1]{\ifthenelse{\equal{\typeof}{0}}{}{#1}}
\newcommand{\shortv}[1]{\ifthenelse{\equal{\typeof}{0}}{#1}{}}
\newcommand{\longshortv}[2]{\ifthenelse{\equal{\typeof}{0}}{#2}{#1}}
\newcommand{\nocommentsv}[1]{\ifthenelse{\equal{\cmnts}{0}}{#1}{}}
\newcommand{\commentsv}[1]{\ifthenelse{\equal{\cmnts}{0}}{}{#1}}

\documentclass[a4paper,10pt]{article}
\usepackage{a4wide}
\usepackage{microtype}
\usepackage{enumerate}
\usepackage{textcomp}
\usepackage{latexsym} 
\usepackage{amsfonts} 
\usepackage{amssymb}
\usepackage{amsmath}
\usepackage{hyperref}
\usepackage[round]{natbib}
\usepackage{subcaption}
\usepackage{cleveref}
\usepackage{general}
\usepackage{sizetypes}
\usepackage{code}

\title{Automating Sized Type Inference for Complexity Analysis (Technical Report)}
\date{}
\author{Martin Avanzini \and Ugo Dal Lago}

\newtheorem{theorem}{Theorem}
\newtheorem{proposition}{Proposition}
\newtheorem{lemma}{Lemma}
\newtheorem{definition}{Definition}

\newtheorem{corollary}{Corollary}

\newenvironment{proof}{\begin{trivlist}
  \item[\hskip \labelsep {\bfseries Proof.}]}{\hfill $\Box$ \end{trivlist}}
\newcommand{\qed}{}

\newcounter{commentcounter}
\newcommand{\cmt}[1]{\nocommentsv{}\commentsv{\refstepcounter{commentcounter}{\color{red}{\textbf{(!)}}\marginpar{\textbf{\color{red}{\scriptsize{Cmt. \thecommentcounter.}}}\\ \color{red}{\texttt{\scriptsize{#1}}}}}}}

\begin{document}
\begin{abstract}
This paper introduces a new methodology for the complexity analysis of
higher-order functional programs, which is based on three ingredients:
a powerful type system for size analysis and a sound type inference
procedure for it, a ticking monadic transformation, and constraint
solving. Noticeably, the presented methodology can be fully automated,
and is able to analyse a series of examples which cannot be handled by
most competitor methodologies. This is possible due to the choice of
adopting an abstract index language and index polymorphism at higher
ranks. A prototype implementation is available.
\end{abstract}

\maketitle

\section{Introduction}
Programs can be incorrect for very different reasons. Modern compilers
are able to detect many syntactic errors, including type errors. When
the errors are semantic, namely when the program is well-formed but
does not compute what it should, traditional static analysis
methodologies like abstract interpretation or model checking could be
of help. When a program is functionally correct but performs quite
poorly in terms of space and runtime behaviour, 
even \emph{defining} the property 
of interest
is very hard. If the units of measurement in which program
performances are measured are close to the physical ones, the problem
can only be solved if the underlying architecture is known, due to the
many transformation and optimisation layers which are applied to programs.
One then obtains WCET
techniques~\citep{WEEA:ATECS:08}, which indeed need to deal with how
much machine instructions cost when executed by modern architectures
(including caches, pipelining, etc.), a task which is becoming even
harder with the current trend towards multicore architectures.

As an alternative, one can analyse the \emph{abstract} complexity of
programs. As an example, one can take the number of evaluation steps
to normal form, as a measure of the underlying program's execution
time. This can be accurate if the actual time complexity \emph{of each
  instruction} is kept low, and has the advantage of being independent
from the specific hardware platform executing the program at hand,
which only needs to be analysed \emph{once}. A variety of verification
techniques have indeed been defined along these lines, from type
systems to program logics, to abstract interpretation,
see~\citep{ABHLM:TCS:07,HDW:POPL:17,AGM:TOCL:13,SZV:2014}. 

If we restrict our attention to higher-order functional programs,
however, the literature becomes sparser. There seems to be a trade-off
between allowing the user full access to the expressive power of
modern, higher-order programming languages, and the fact that
higher-order parameter passing is a mechanism which intrinsically
poses problems to complexity analysis: how big is a certain (closure
representation of a) higher-order parameter?  If we focus our
attention on automatic techniques for the complexity analysis of
higher-order programs, the literature only provides very few
proposals~\citep{HDW:POPL:17,ADM:ICFP:15,Vasconcelos:Diss:08}, which
we will discuss in Section~\ref{sect:ERW} below.

One successful approach to automatic verification of
\emph{termination} properties of higher-order functional programs is
based on \emph{sized types}~\citep{HPS:POPL:96}, and has been shown to
be quite robust~\citep{GBR:CSL:08}. In sized types, a type carries not
only some information about the \emph{kind} of each object, but also
about its \emph{size}, hence the name.  This information is then
exploited when requiring that recursive calls are done on arguments of
\emph{strictly smaller} size, thus enforcing termination.  Estimating
the size of intermediate results is also a crucial aspect of
complexity analysis, but up to now the only attempt of using sized
types for complexity analysis is due to
\citet{Vasconcelos:Diss:08}, and confined to space
complexity. If one wants to be sound for time analysis, size types
need to be further refined, e.g., by turning them into linear
dependently types~\citep{LG:LMCS:11}.

In this paper, we take a fresh look at sized types by introducing a
new type system which is substantially more expressive than the
traditional one. This is possible due to the presence of
\emph{arbitrary rank index polymorphism}: functions that take
functions as their argument can be polymorphic in their size
annotation. The introduced system is then proved to be a sound
methodology for \emph{size} analysis, and a type inference algorithm
is given and proved sound and relatively complete. Finally, the type
system is shown to be amenable to time complexity analysis by a
ticking monadic transformation. A prototype implementation is
available, see below for more details. More specifically, this paper's
contributions can be summarized as follows:
\begin{varitemize}
\item
  We show that size types can be generalised so as to encompass a
  notion of index polymorphism, in which (higher-order subtypes of)
  the underlying type can be universally quantified. This allows for a
  more flexible treatment of higher-order functions. Noticeably, this
  is shown to preserve soundness (i.e., subject reduction), the minimal
  property one expects from such a type system. On the one hand, this
  is enough to be sure that types reflect the size of the underlying
  program. On the other hand, termination is not enforced anymore by
  the type system, contrarily to, e.g., the system of \citet{HPS:POPL:96}. In
  particular, we do not require that recursive calls are made on
  arguments of smaller size. All this is formulated on a language of
  applicative programs, introduced in Section~\ref{sect:APST}, and
  will be developed in Section~\ref{sect:STS}. Nameless functions
  (i.e., $\lambda$-abstractions) are not considered for brevity, 
  as these can be easily lifted to the top-level.
\item
  The type inference problem is shown to be (relatively) decidable by
  giving in Section~\ref{sect:STI} an algorithm which, given a
  program, produces in output candidate types for the program,
  together with a set of integer index inequalities which need to be
  checked for satisfiability. This style of results is quite common in
  similar kinds of type systems.
  What is uncommon though, at least in the context of sized types, 
  is that we do not restrict ourselves to a particular
  algebra in which sizes are expressed. 
  Indeed, many of the more advanced sized type systems are restricted to 
  the successor algebra \cite{Blanqui:CSL:05,AP:JFP:16}. 
  This is often sufficient in the context of termination analysis, 
  where one is interested in determining which recursion parameters decrease.
  Here, the programs runtime will be expressed in this algebra,
  and thus a more expressive algebra is required.
\item
  The polymorphic sized type system, by itself, does not guarantee
  any complexity-theoretic property on the typed program, except for
  the \emph{size} of the output being bounded by a function on the
  size of the input, itself readable from the type. Complexity
  analysis of a program $\progone$ can however be seen as a size analysis of
  another program $\tprogone$ which computes not only $\progone$, but its
  complexity. This transformation, called the \emph{ticking transformation}, has 
  already been studied in similar settings~\citep{DLR:ICFP:15}, but this study has never been 
  automated.
  The ticking transformation is formally introduced in Section~\ref{sect:TTTCA}.
\item
  Contrarily to many papers from the literature, we spent considerable efforts on 
  devising a system that is susceptible to automation with current technology. 
  Moreover, we have taken care not only of constraint \emph{inference}, 
  but also of constraint \emph{solving}.  
  To demonstrate the feasibility of our approach, 
  we have built a prototype which implements type inference, 
  resulting in a set of constraints. To deal with the resulting 
  constraints, we have also built a constraint solver on 
  top of state-of-the-art SMT solvers.
  All this, together with some experimental results, are described in detail in
  Section~\ref{sect:PER}.
\end{varitemize}
An extended version with more details and proofs is available online \cite{Techreport}.


\section{A Bird Eye's View on Index-Polymorphic Sized Types}\label{sect:ERW}
In this section, we will motivate the design choices we made when
defining our type system through some examples. This can also be taken
as a gentle introduction to the system for those readers which are
familiar with functional programming and type theory. Our type system
shares quite some similarities with the seminal type system introduced
by \citet{HPS:POPL:96} and similar ones \citep{GBR:CSL:08,Vasconcelos:Diss:08}, 
but we try to keep presentation as self-contained as possible.
\paragraph{Basics.}
We work with functional programs over a fixed set of inductive
datatypes, e.g. $\tynat$ for natural numbers and $\tylist{\alpha}$ for
lists over elements of type $\tyvarone$.  Each such datatype is
associated with a set of typed \emph{constructors}, below we will use
the constructors $0 \decl \tynat{}$, $\con{Succ} \decl \tynat \starr
\tynat$ for naturals, and the constructor $\cnil \decl \forall
\tyvarone.\ \tylist{\tyvarone}$ and the infix constructor $(\ccons)
\decl \forall \tyvarone.\ \tyvarone \starr \tylist{\tyvarone} \starr
\tylist{\tyvarone}$ for lists.  Sized types refine each such datatype
into a \emph{family} of datatypes indexed by natural
numbers, their \emph{size}. E.g., to $\tynat$ and
$\tylist{\tyvarone}$ we associate the families
$\tynat[0],\tynat[1],\tynat[2],\dots$ and
$\tylist[0]{\tyvarone},\tylist[1]{\tyvarone},\tylist[2]{\tyvarone},\dots$,
respectively. 
An indexed datatype such as $\tylist[n]{\tynat[m]}$
then represents lists of length $n$, over naturals of size $m$. 

A function $\fun{f}$ will then be given a polymorphic type $\forall
\vec{\tyvarone}.\ \forall\vec{\ivarone}.\ \mtpone \to \mtptwo$. 
Whereas the variables $\vec{\tyvarone}$ range over types, the variables
$\vec{\ivarone}$ range over sizes.  Datatypes
occurring in the types $\mtpone$ and $\mtptwo$ will be indexed by
expressions over the variables $\vec{\ivarone}$. E.g., the append
function can be attributed the sized type $\forall \tyvarone.\ \forall
\ivarone\ivartwo.\ \tylist[\ivarone]{\tyvarone} \starr
\tylist[\ivartwo]{\tyvarone} \starr \tylist[\ivarone +
  \ivartwo]{\tyvarone}$.

Soundness of our type-system will guarantee that when append is
applied to lists of length $n$ and $m$ respectively, it will yield a list
of size $n + m$, or possibly diverge. In particular, our type system is
\emph{not} meant to guarantee termination, and complexity analysis will be done
via the aforementioned ticking transformation, to be described later. As
customary in sized types, we will also integrate a
subtyping relation $\mtpone \subtypeof[] \mtptwo$ into our system,
allowing us to relax size annotations to less precise ones.
This flexibility is necessary to treat conditionals where the
branches are attributed different sizes, or, to treat higher-order
combinators which are used in multiple contexts.

Our type system, compared to those from the literature, has its main
novelty in polymorphism, but is also different in some key aspects,
addressing intensionality but also practical considerations towards
type inference. In the following, we shortly discuss the main differences.
\paragraph{Canonical Polymorphic Types.}
We allow polymorphism over size expressions, but put some syntactic
restrictions on function declarations: In essence, we disallow
non-variable size annotations directly to the left of an arrow, and
furthermore, all these variables must be pairwise distinct.  We call
such types canonical. The first restriction dictates that
e.g. $\fun{half} \decl \forall \ivarone. \tynat[2 \cdot i] \starr
\tynat[i]$ has to be written as $\fun{half} \decl \forall
\ivarone. \tynat[i] \starr \tynat[i/2]$.  The second restriction
prohibits e.g.\ the type declaration $\fun{f} \decl \forall
\ivarone. \tynat[\ivarone] \starr \tynat[\ivarone] \starr \mtpone$,
rather, we have to declare $\fun{f}$ with a more general type 
$\forall \ivarone \ivartwo. \tynat[\ivarone] \starr \tynat[\ivartwo] \starr \mtpone'$.  
The two restrictions considerably simplify the inference
machinery when dealing with pattern matching, and pave the way towards
automation.  Instead of a complicated unification based mechanism, a
matching mechanism suffices. Unlike in~\cite{HPS:POPL:96}, where
indices are formed over naturals and addition, we keep the index
language abstract. This allows for more flexibility, and ultimately we
can capture more programs. Indeed, having the freedom of not adopting
a fixed index language is known to lead towards
completeness~\cite{LG:LMCS:11}.
\paragraph{Polymorphic Recursion over Sizes.}
%
\begin{figure}[t]
  \begin{framed}
    \begin{minipage}{1.0\linewidth}
\begin{lstlisting}[emph={i,j,k,l,ijk,ij,x,xs,f,ys},emph={[2] reverse, rev},style=haskell, style=numbered]
rev $\decl \forall \tyvarone.\ \forall \ivarone\ivartwo.\ \tylist[\ivarone]{\tyvarone} \starr \tylist[\ivartwo]{\tyvarone} \starr \tylist[\ivarone + \ivartwo]{\tyvarone}$
rev []       ys = ys
rev (x : xs) ys = rev xs (x : ys)%\label{fig:reverse:ind}%

reverse $\decl \forall \tyvarone.\ \forall \ivarone.\ \tylist[\ivarone]{\tyvarone} \starr \tylist[\ivarone]{\tyvarone}$
reverse xs = rev xs [] 
\end{lstlisting}
    \end{minipage}
  \end{framed}
  \caption{Sized type annotated tail-recursive list reversal function.}\label{fig:reverse}
\end{figure}
Type inference in functional programming languages, such as \haskell\ or \ocaml, 
is restricted to parametric polymorphism in the form of \emph{let-polymorphism}.
Recursive definitions are checked under a monotype, thus, 
types cannot change between recursive calls. 
Recursive functions that require full parametric polymorphism~\cite{Mycroft:ISP:84} have to be 
annotated in general, as type inference is undecidable in this setting.

Let-polymorphism poses a significant restriction in our context, because sized types considerably refine upon simple types.
Consider for instance the usual tail-recursive definition of list reversal depicted in Figure~\ref{fig:reverse}.
With respect to the annotated sized types, 
in the body of the auxiliary function $\fun{rev}$ defined on line \ref{fig:reverse:ind}, 
the type of the second argument to $\fun{rev}$ will change from $\tylist[\ivartwo]{\tyvarone}$ (the assumed type of $ys$) to
$\tylist[\ivartwo+1]{\tyvarone}$ (the inferred type of $x \ccons ys$). Consequently, $\fun{rev}$ is not typeable under a monomorphic sized type.
Thus, to handle even such very simple functions, we will have to overcome let-polymorphism, on the layer of size annotations.
To this end, conceptually we allow also recursive calls to be given a type polymorphic over size variables. 
This is more general than the typing rule for recursive definitions found in more traditional systems~\cite{HPS:POPL:96,GBR:CSL:08}.

\paragraph{Higher-ranked Polymorphism over Sizes.}
In order to remain decidable, classical type inference systems 
work on polymorphic types in \emph{prenex form} $\forall{\vec{\tyvarone}}.\mtpone$, 
where $\mtpone$ is quantifier free.  In our context, it is often not enough to give a
combinator a type in prenex form, in particular when the combinator
uses a functional argument more than once. All uses of the functional
argument have to be given then \emph{the same} type. In the context of sized
types, this means that functional arguments can be applied only to
expressions whose attributed size equals.  
\shortv{%
This happens for instance in recursive combinators, but also non-recursive
ones such as the function $\fun{twice}  \api f \api x = f \api (f \api x)$.
}\longv{%
This happens for instance
in recursive combinators, but also non-recursive
ones such as the following function $\fun{twice}$.
\begin{align*}
  & \fun{twice} \decl \forall \tyvarone.\ (\tyvarone \starr \tyvarone) \starr \tyvarone \starr \tyvarone\\
  & \fun{twice} \api f \api x = f \api (f \api x)
    \tpkt
\end{align*}}
A strong type-system would allow us to type the expression
$\fun{twice} \api \con{Succ}$ with a sized type $\tynat[c] \starr
\tynat[c + 2]$.  A (specialised) type in prenex form for $\fun{twice}$, such as
\[
  \fun{twice} \decl
  \forall \ivarone.\ (\tynat[\ivarone] \starr \tynat[\ivarone + 1]) \starr \tynat[\ivarone] \starr \tynat[\ivarone + 2]
  \tkom
\]
would immediately yield the mentioned sized type for $\fun{twice} \api
\con{Succ}$.  However, we will not be able to type $\fun{twice}$
itself, because the outer occurrence of $f$ would need to be typed as
$\tynat[\ivarone+1] \starr \tynat[\ivarone + 2]$, whereas the type of
$\fun{twice}$ dictates that $f$ has type $\tynat[\ivarone] \starr
\tynat[\ivarone + 1]$.

The way out is to allow polymorphic types of rank \emph{higher than}
one when it comes to size variables, i.e.\ to allow quantification of
size variables to the left of an arrow at arbitrary depth.  Thus, we can declare
\[
  \fun{twice} \decl
  \forall \ivarone.\ (\forall \ivartwo.\tynat[\ivartwo] \starr \tynat[\ivartwo + 1]) \starr \tynat[\ivarone] \starr \tynat[\ivarone + 2]
  \tpkt
\]
This allows us to type the expression $\fun{twice} \api
\con{Succ}$ as desired.  Moreover, the inner quantifier permits the
two occurrences of the variable $f$ in the body of $\fun{twice}$ to
take types $\tynat[\ivarone] \starr \tynat[\ivarone + 1]$ and
$\tynat[\ivarone+1] \starr \tynat[\ivarone + 2]$ respectively, and
thus $\fun{twice}$ is well-typed.
\paragraph{A Worked Out Example.}
\begin{figure}[t]
  \centering
  \begin{framed}
    \begin{minipage}{1.0\linewidth}
\begin{lstlisting}[emph={f,x,xs,b,ms,ns,ps,m,n},emph={[2] foldr,product},style=haskell, style=numbered]
foldr $\decl \forall \tyvarone\tyvartwo.\ \forall \ivartwo\ivarthree\ivarfour.\ (\forall \ivarone.\ \tyvarone \starr \tylist[\ivarone]{\tyvartwo} \starr \tylist[\ivarone + \ivartwo]{\tyvartwo})
        \starr \tylist[\ivarthree]{\tyvartwo} \starr \tylist[\ivarfour]{\tyvarone} \starr \tylist[\ivarfour \cdot \ivartwo + \ivarthree]{\tyvartwo}$
foldr f b []       = b
foldr f b (x : xs) = f x (foldr f b xs)%\label{fig:doublefilter:foldrec}%

product $\decl \forall \tyvarone \tyvartwo.\ \forall \ivarone\ivartwo.\ \tylist[\ivarone]{\tyvarone} \starr \tylist[\ivartwo]{\tyvartwo} \starr \tylist[\ivarone \cdot \ivartwo]{(\tpair{\tyvarone}{\tyvartwo})}$
product ms ns = foldr (\ m ps. foldr (\ n. (:) (m,n)) ps ns) [] ms
\end{lstlisting}
    \end{minipage}
  \end{framed}
  \caption{Sized type annotated program computing the cross-product of two lists.}\label{fig:doublefilter}
\end{figure}

We conclude this section by giving a nontrivial example.
The sized type
annotated program is given in Figure~\ref{fig:doublefilter}.  The
function $\fun{product}$ computes the cross-product $[\,(m,n) \mid m
  \in ms, n \in ns\,]$ for two given lists $ms$ and $ns$. It is
defined in terms of two folds. The inner fold appends, for a fixed
element $m$, the list $[\,(m,n) \mid n \in ns \,]$ to an accumulator
$ps$, the outer fold traverses this function over all elements $m$
from $ms$.

In a nutshell, checking that a function $\fun{f}$ is typed correctly
amounts to checking that all its defining equations are well-typed,
i.e.\ under the assumption that the variables are typed according to
the type declaration of $\fun{f}$, the right-hand side of the equation
has to be given the corresponding return-type. Of course, all of this
has to take pattern matching into account.
Let us illustrate this on the recursive equation of $\fun{foldr}$ given in Line~\ref{fig:doublefilter:foldrec} in Figure~\ref{fig:doublefilter}.
Throughout the following, we denote by $\termone \oftype \mtpone$ that
the term $\termone$ has type $\mtpone$.  
To show that the equation is well-typed, let us assume the following 
types for arguments: 
$f \oftype \forall \ivarone.\ \tyvarone \starr \tylist[\ivarone]{\tyvartwo} \starr \tylist[\ivarone + \ivartwo]{\tyvartwo}$, 
$b \oftype \tylist[\ivarthree]{\tyvartwo}$, 
$x \oftype \tyvarone$ and 
$xs \oftype \tylist[m]{\tyvarone}$ for arbitrary size-indices
$\ivartwo,\ivarthree,m$.
Under these assumptions, the left-hand side has type 
$\tylist[(m + 1) \cdot \ivartwo + \ivarthree]{\tyvartwo}$, taking 
into account that the recursion parameter $x \ccons xs$ has size $m+1$.
To show that the equation is well-typed, we verify that the right-hand 
side can be attributed the same sized type.
\begin{varenumerate}
\item 
  We instantiate the polymorphic type of $\fun{foldr}$ and derive
  \[
    \fun{foldr} \oftype (\forall \ivarone.\ \tyvarone \starr \tylist[\ivarone]{\tyvartwo} \starr \tylist[\ivarone + \ivartwo]{\tyvartwo})
    \starr \tylist[\ivarthree]{\tyvartwo} \starr \tylist[m]{\tyvarone} \starr \tylist[m \cdot \ivartwo + \ivarthree]{\tyvartwo}\tspkt
  \]
\item 
  from this and the above assumptions we get $\fun{foldr} \api f \api b \api xs \oftype
  \tylist[m \cdot \ivartwo + \ivarthree]{\tyvartwo}$;
\item 
  by instantiating the quantified size variable $\ivarone$ in the assumed type of $f$ with 
  the index term $m \cdot \ivartwo + \ivarthree$ we get
  $f \oftype \tyvarone \starr \tylist[m \cdot \ivartwo + \ivarthree]{\tyvartwo} \starr \tylist[(m \cdot \ivartwo + \ivarthree) + \ivartwo]{\tyvartwo}$;
\item 
  from the last two steps we finally get $f \api x \api (\fun{foldr} \api f \api b \api xs) \oftype \tylist[(m + 1) \cdot \ivartwo + \ivarthree]{\tyvartwo}$.
\end{varenumerate}
We will not explain the type checking of the remaining
equations, but revisit this example in Section~\ref{sect:PER}.


\section{On Related Work}\label{sect:RW}
Since the first inception in the seminal paper of
\citet{HPS:POPL:96}, the literature on sized types has
grown to a considerable extent.  Indeed, various significantly more
expressive systems have been introduced, with the main aim to enlarge
the class of typable (and thus proved terminating) programs. For
instance, \citet{Blanqui:CSL:05} introduced a novel sized type
system on top of the \emph{calculus of algebraic constructions}. 

Notably, it has been shown that for size indices over
the successor algebra, type checking is
decidable~\cite{Blanqui:CSL:05}.  The system is thus capable of
expressing additive relations between sizes. In the context of
termination analysis, where one would like to statically detect that a
recursion parameter decreases in size, this is sufficient.  In this
line of research falls also more recent work of
\citet{AP:JFP:16}, where a novel sized type system for
termination analysis on top of $\mathsf{F}_\omega$ is proposed.
Noteworthy, this system has been integrated in the dependently typed
language~\tool{Agda}.

Type systems related to sized types have been introduced and studied
not only in the context of termination analysis, but also for size and
complexity analysis of programs. One noticeable example is the series
of work by \citet{Shkaravska:LMCS:09}, which aims at
size analysis but which is limited to first-order programs. 
Also \citet{CW:POPL:00} use types, like here, to express 
the runtime of functions. However, the system is inherently \emph{semi-automatic}.
Related to this is also the work by \citet{Danielsson:POPL:08},
whose aim is again complexity analysis, but which is not fully automatable
and limited to linear bounds. 
If one's aim is complexity analysis of higher-order functional programs,
achieving a form of completeness is indeed possible by linear dependent
types~\cite{LG:LMCS:11,LP:SCP:14}. 
While the front-end of this verification machinery is
fully-automatable~\cite{LP:POPL:13}, the back-end is definitely
not, and this is the reason why this paper should be considered an
advance over this body of work.
Our work is also related to that of \citet{GM:POPL:16}, which uses 
a combination of runtime and size analysis to reason about the complexity 
of functional programs expressed as interaction nets. 

Our work draws inspiration from \citet{DLR:ICFP:15}. 
In this work, the complexity analysis 
of higher-order functional programs, defined in a system akin to G\"odel's $\mathsf{T}$
enriched with inductive types, is studied. 
A ticking transformation is used to instrument the program with a clock, 
recurrence relations are then extracted from the ticked version that express 
the complexity of the input program.
Conceptually, our ticking transformation is identical
to the one defined by \citeauthor{DLR:ICFP:15}, and differs only in details 
to account for the peculiarities of the language that we are considering. 
In particular, our simulation theorem, Theorem~\ref{t:simul}, 
has an analogue in \cite{DLR:ICFP:15}. The proof in the present work is however
more delicate, as our language admits arbitrary recursion and 
programs may thus very well diverge. 
To our best knowledge, no attempts have 
been made so far to automate solving of the resulting recurrences.

In contrast, \citeauthor{HAH:POPL:11} refine in a series of works 
the methodology of \citet{JHLH:POPL:10} based on 
\emph{Tarjan's amortised resource analysis}. 
This lead to the development of \raml~\cite{HAH:CAV:12}, a fully fledged automated 
resource analysis tool. 
Similar to the present work, the analysis is expressed as a type system.  
Data types are annotated by \emph{potentials}, inference generates a set of linear constraints 
which are then solved by an external tool.
This form of analysis can not only deal with non-linear bounds~\cite{HAH:POPL:11}, 
but also demonstrates that type based systems are relatively stable under language features
such as parallelism~\cite{HZ:ESOP:15} or imperative features~\cite{ZH:JFP:15}. 
In more recent work \cite{HDW:POPL:17},
the methodology has been lifted to the higher-order case and \raml\ can now interface
with Inria's \tool{OCaml} compiler. Noteworthy, some of the peculiarities of 
this compiler are taken into account. 
The overall approach is in general incomparable to our methodology. 
Whilst it seems feasible, our method neither takes amortisation into account
nor does our prototype interface with a industrial strength compiler. 
On the other hand, our system can properly account for
closures, whereas inherent to the methodology underlying \raml,
closures can only be dealt with in a very restricted form.  We return
to this point in Section~\ref{sect:PER} within our experimental
assessment.
%

There are also connections to the work of the authors and Moser
\cite{ADM:ICFP:15}, where a complexity preserving
transformation from higher-order to first-order programs is
proposed. This transformation works by a form of control-flow guided
defunctionalisation.  Furthermore, a variety of simplification
techniques, such as inlining and narrowing, are employed to make the
resulting first-order program susceptible to an automated analysis.
The complete procedure has been implemented in the tool \hoca, which
relies on the complexity analyser \tct~\cite{AMS:TACAS:16} to analyse
the resulting first-order program.  Unlike for our system, it is
unclear whether the overall method can be used to derive precise
bounds.


\section{Applicative Programs and Simple Types}\label{sect:APST}
We restrict our attention to a small prototypical, strongly typed
functional programming language. For the sake of simplifying
presentation, we impose a simple, monomorphic, type system on
programs, which does not guarantee anything except a form of type
soundness.  We will only later in this paper introduce sized types
proper.  Our theory can be extended straightforwardly to an ML-style
polymorphic type setting. Here, such an extension would only distract
from the essentials. Indeed, our implementation (described in Section
\ref{sect:PER}) allows polymorphic function definitions.
\paragraph{Statics.}
Let $\STBASE$ denote a finite set of base types $\stbaseone,\stbasetwo, \dots$\,. 
\emph{Simple types} are inductively generated from $\stbaseone\in\STBASE$:
\begin{align*}
  & \text{\textbf{(simple types)}} & \stone,\sttwo,\stthree \bnfdef {} 
  & \stbase \mid\ \tpair{\stone}{\sttwo} \mid \stone \starr \sttwo \tpkt
\end{align*}

We follow the usual convention that $\starr$ associates to the right.
Let $\VAR$ denote a countably infinite set of \emph{variables}, ranged
over by metavariables like $\varone$, $\vartwo$.  Furthermore, let
$\FUN$ and $\CON$ denote two disjoint sets of symbols, the set of
\emph{functions} and \emph{constructors}, respectively, all pairwise
distinct with elements from $\VAR$. Functions and constructors are
denoted in \texttt{teletype font}.  We keep the convention that
functions start with a lower-case letter, whereas constructors start
with an upper-case letter.  Each symbol $s \in \VAR \cup \FUN \cup
\CON$ has a simple type $\stone$, and when we want to insist on that,
we write $s^\stone$ instead of just $s$.  Furthermore, each symbol
$s^{\stone_1 \starr \cdots \starr \stone_n \starr \sttwo} \in \FUN
\cup \CON$ is associated with a natural number $\arity{s} \leq n$, its
\emph{arity}.  The set of \emph{terms}, \emph{patterns}, \emph{values}
and \emph{data values} over functions $\funone \in \FUN$, constructors
$\conone \in \CON$ and variables $\varone \in \VAR$ is inductively
generated as follows.  Here, each term receives implicitly a type, in
Church style. Below, we employ the usual convention that application
associates to the left.
\begin{align*}
  & \text{\textbf{(terms)}} & \termone,\termtwo \bnfdef {} 
  & \varone^\stone  && \textit{variable} \\[-1mm]
  &&\mid\ & \funone^\stone && \textit{function}\\[-1mm]
  &&\mid\ & \conone^\stone && \textit{constructor}\\[-1mm]
  &&\mid\ & (\termone^{\stone \starr \sttwo} \api \termtwo^{\stone})^{\sttwo} && \textit{application}\\[-1mm]
  &&\mid\ & (\termone^{\stone}, \termtwo^{\sttwo})^{\tpair{\stone}{\sttwo}} && \textit{pair constructors}\\[-1mm]
  &&\mid\ & (\letexp{\termone^{\tpair{\stone}{\sttwo}}}{\varone^{\stone}}{\vartwo^{\sttwo}}{\termtwo^\stthree})^\stthree && \textit{pair destructor;}\\[1mm]
  & \text{\textbf{(patterns)}} & \patone,\pattwo \bnfdef {}
  & \varone^\stone \mid \conone^{\stone_1 \starr \cdots \stone_n \starr \stbase} \api \pat_1^{\stone_1} \cdots \pat_n^{\stone_n}\tspkt\\[1mm]
  & \text{\textbf{(values)}} & \valone,\valtwo \bnfdef {}
  & 
    \conone^{\stone_1 \starr \cdots \starr \stone_n \starr \stone} \api \valone_1^{\stone_1} \cdots \valone_n^{\stone_n} \\[-1mm]
  && \mid\ & \funone^{\stone_1 \starr \cdots \starr \stone_m \starr \stone_{m+1} \starr \stone} \api \valone_1^{\stone_1} \cdots \valone_m^{\stone_m}\\[-1mm]
  && \mid\ & (\valone^\stone,\valtwo^\sttwo)^{\tpair{\stone}{\sttwo}}\tspkt\\[1mm]
  & \text{\textbf{(data values)}} & \dataone \bnfdef {} 
  & \conone^{\stbase_1 \starr \cdots \starr \stbase_{n+1}} \api \dataone_1 \cdots \dataone_n\tpkt
\end{align*}
The presented operators are all standard, except the pair destructor
$\letexp{\termone}{\varone}{\vartwo}{\termtwo}$ which binds the variables
$\varone$ and $\vartwo$ to the two components of the result of
$\termone$ in $\termtwo$. The set of \emph{free variables}
$\FV(\termone)$ of a term $\termone$ is defined in the usual way.  If
$\FV(\termone) = \varnothing$, we call $\termone$ \emph{ground}.  A
term $\termone$ is called \emph{linear}, if each variable occurs at
most once in $\termone$.  A \emph{substitution} $\substone$ is a
finite mapping from variables $\varone^\stone$ to terms
$\termone^\stone$.  The substitution mapping
$\vec{\varone}=\varone_1,\ldots,\varone_n$ to
$\vec{\termone}=\termone_1,\ldots,\termone_n$, respectively, is
indicated with $\substseq{\varone}{\termone}$ or
$\substvec{\varone}{\termone}$ for short.  The variables
$\vec{\varone}$ are called the \emph{domain} of $\substone$.  We
denote by $\termone\substone$ the application of $\substone$ to
$\termone$. Let-bound variables are renamed to avoid variable capture.
\longv{The composition $\substone_2 \compose \substone_1$ of two
  substitutions is given by the substitution that maps elements
  $\varone$ from the domain of $\substone_1$ to
  $(\varone\substone_1)\substone_2$.}

A \emph{program} $\progone$ over functions $\FUN$ and constructors
$\CON$ defines each function $\funone \in \FUN$ through a finite set
of \emph{equations} $l^\stone = r^\stone$, where $l$ is of the form
$\funone \api \pat_1 \api \cdots \pat_{\arity{f}}$.  We put the usual
restriction on equations that each variable occurs at most once in
$l$, i.e. that $l$ is linear, and that the variables of the
\emph{right-hand side} $r$ are all included in $l$.  To keep the
semantics simple, we do not impose any order on the equations.
Instead, we require that left-hand sides defining $\funone$ are all
pairwise non-overlapping.  This ensures that our programming model is
deterministic.

Some remarks are in order before proceeding.  As standard in
functional programming, only values of base type can be destructed by
pattern matching. In a pattern, a constructor always needs to be fully
applied.  We deliberately disallow the destruction of pairs through
pattern matching. This would unnecessarily complicate some key
definitions in later sections.  Instead, a dedicated destructor
$\letexp{\termone}{\varone}{\vartwo}{\termtwo}$ is provided.  We also
excluded $\lambda$-abstractions from our language, for brevity, 
as these can always be lifted to the top-level. Similarly,
conditionals and case-expressions would not improve upon
expressivity.

\paragraph{Dynamics.}
We impose a \emph{call-by-value} semantics on programs
$\progone$.  \emph{Evaluation contexts} are defined according to the
following grammar:
\begin{align*}
  E \bnfdef {} &
  \hole^\stone \mid (E^{\stone\starr\sttwo} \api \termone^\stone)^\sttwo 
     \mid (\termone^{\stone\starr\sttwo} \api E^\stone)^\sttwo
     \mid (E^\stone, \termone^\sttwo)^{\tpair{\stone}{\sttwo}} 
     \mid (\termone^\stone, E^\sttwo)^{\tpair{\stone}{\sttwo}}
     \mid (\letexp{E^{\tpair{\stone}{\sttwo}}}{\varone^\stone}{\vartwo^\sttwo}{\termone^\stthree})^\stthree
     \tpkt
\end{align*}
As with terms, type annotations will be omitted from evaluation
contexts whenever this does not cause ambiguity. With
$E[\termone^\stone]$ we denote the term obtained by replacing the \emph{hole}
$\hole^\stone$ in $E$ by $\termone^\stone$.  The one-step
\emph{call-by-value} reduction relation $\rew[\progone]$, defined over
ground terms, is then given as the closure over all evaluation
contexts, of the following two rules:
\[
\Infer
    {(\funone\api \patone_1 \api \cdots \api \patone_n)\substvec{\varone}{\valone} \rew[\progone] r\substvec{\varone}{\valone}}
    {\funone\api \patone_1 \api \cdots \api \patone_n = r \in \progone}
\quad
\Infer
    {\letexp{\pair{\valone}{\valtwo}}{\varone}{\vartwo}{\termtwo}
     \rew[\progone] 
     \termtwo\subst{\varone,\vartwo}{\valone,\valtwo}}
    {}
\]
We denote by $\rss[\progone]$ the transitive and reflexive closure, and
likewise, $\rsl[\progone]{\ell}$ denotes the $\ell$-fold composition
of $\rew[\progone]$.

Notice that reduction simply gets stuck if pattern
matching in the definition of $\funone$ is not exhaustive.
We did not specify a particular reduction order, e.g., left-to-right
or right-to-left.  Reduction itself is thus non-deterministic, but
this poses no problem since programs are \emph{non-ambiguous}: not
only are the results of a computation independent from the reduction
order, but also reduction lengths coincide.
\begin{proposition}\label{p:nonambiguous}
    All \emph{normalising reductions} of $\termone$ have the same length and yield the same result, i.e.\ if 
    $\termone \rsl[\progone]{m} \valone$ and $\termone \rsl[\TRSone]{n} \valtwo$
    then $m = n$ and $\valone = \valtwo$.
\end{proposition}
To define the \emph{runtime-complexity} of $\progone$, we assume a
single entry point to the program via a \emph{first-order} function
$\fun{main}^{\stbase_1 \starr \cdots \starr \stbase_k \starr \stbase_n}$, 
which takes as input data values and also produces a data value as output.
The (\emph{worst-case}) \emph{runtime-complexity} of
$\progone$ then measures the reduction length of $\fun{main}$
in the sizes of the inputs. Here, the size $\size{\dataone}$ 
of a data value is defined as the number of constructors in $\dataone$.
Formally, the runtime-complexity function of $\progone$ is defined
as the function $\rc[\progone]:\N \times \cdots \times \N \rightarrow\N^\infty$:
\begin{align*}
  \rc[\progone](\seq[k]{n}) \defsym \sup \{ \ell \mid {} & \exists \seq[k]{\dataone}.\ \fun{main} \api \aseq[k]{\dataone} \rsl[\progone]{\ell} \termone \text{ and $\size{\dataone_i} \leqslant n_i$}\}\tpkt
  \tpkt
\end{align*}
We emphasise that the runtime-complexity function defines a
cost model that is invariant to traditional models of computation,
e.g., Turing machines~\citep{LM:ICALP:09,AM:RTA:10}.


\section{Sized Types and Their Soundness}\label{sect:STS}
\longv{
\newcommand{\BN}[1]{\sizeannotate{\tycon{Int}}{#1}}
\newcommand{\BL}[1]{\sizeannotate{\tycon{IntList}}{#1}}
}
This section is devoted to introducing the main object of study of
this paper, namely a sized type system for the applicative programs that we
introduced in Section~\ref{sect:APST}. We have tried to keep the
presentation of the relatively involved underlying concepts
as simple as possible.
\subsection{Indices}
As a first step, we make the notion of \emph{size index}, 
with which we will later annotate data types, precise. 
Let $\IS$ denote a set of first-order function symbols, the
\emph{index symbols}.  Any symbol $\ifunone \in \IS$ is associated
with a natural number $\arity{\ifunone}$, its \emph{arity}.  The set
of \emph{index terms} is generated over a countable infinite set of
\emph{index variables} $\ivarone \in \IVARS$ and index symbols
$\ifunone \in \IS$.
\begin{align*}
  & \text{\textbf{(index terms)}} & \itermone,\itermtwo \bnfdef {} \ivarone \mid \ifunone(\seq[\arity{\ifunone}]{\itermone}) \tpkt
\end{align*}
We denote by $\Var{\itermone} \subset \IVARS$ the set of variables
occurring in $\itermone$.  Substitutions mapping index variables to
index terms are called \emph{index substitutions}.  With $\isubstone$
we always denote an index substitution.  We adopt the notions
concerning term substitutions to index substitutions from the previous
section.

Throughout this section, $\IS$ is kept fixed. Meaning is given to
index terms through an \emph{interpretation} $\iinter$, that maps
every $k$-ary $\ifunone \in \IS$ to a (total) and \emph{weakly
  monotonic} function $\interpretation[\iinter]{\ifunone} \ofdom
\N^{\arity{\ifunone}} \to \N$.
We suppose that $\IS$ always contains a constant
$\izero$, a unary symbol $\isucc$, and a binary symbol $+$ which we
write in infix notation below. These are always interpreted as zero,
the successor function and addition, respectively.  Our index language
encompasses the one of \citet{HPS:POPL:96}, where
linear expressions over natural numbers are considered.  
The interpretation of an index term $\itermone$, under an
\emph{assignment} $\assignone \ofdom \IVARS \to \N$ and an
interpretation $\iinter$, is defined recursively in the usual way:
\longshortv{%
\[
   \interpret[\iinter][\assignone]{\itermone} \defsym 
   \begin{cases}
     \assignone(\itermone) & \text{if $\itermone \in \IVARS$,}\\
     \interpretation[\iinter]{\ifunone}(\interpret[\iinter][\assignone]{\itermone_1},\dots,\interpret[\iinter][\assignone]{\itermone_k})
     & \text{if $\itermone = \ifunone(\seq[k]{\itermone})$.}
   \end{cases}
\]}{%
$\interpret[\iinter][\assignone]{\ivarone} \defsym \assignone(\ivarone)$
and 
$\interpret[\iinter][\assignone]{\ifunone(\seq[k]{\itermone})} 
\defsym \interpretation[\iinter]{\ifunone}(\interpret[\iinter][\assignone]{\itermone_1},\dots,\interpret[\iinter][\assignone]{\itermone_k})$.
}
We define $\itermone \leqs[\iinter] \itermtwo$ if
$\interpret[\iinter][\assignone]{\itermone} \leq
\interpret[\iinter][\assignone]{\itermtwo}$ holds \emph{for all}
assignments $\assignone$.
The following lemma collects useful properties of the relation $\leqs[\iinter]$.
\begin{lemma}\label{l:leqs}
  \envskipline
  \begin{varenumerate}
    \item\label{l:leqs:refltrans}\label{l:leqs:subst}
      The relation $\leqs$ is reflexive, transitive and closed under substitutions,
      i.e. $\itermone \leqs \itermtwo$ implies $\itermone\isubstone \leqs \itermtwo\isubstone$.
    \item\label{l:leqs:mon}  
      If $\itermone \leqs \itermtwo$ then 
      $\itermthree\subst{\ivarone}{\itermone} \leqs \itermthree\subst{\ivarone}{\itermtwo}$
      for each index term $\itermthree$. 
    \item\label{l:leqs:zeroleft}
      If $\itermone \leqs \itermtwo$ then $\itermone\subst{\ivarone}{0} \leqs \itermtwo$.
    \item\label{l:leqs:zeroright} 
      If $\itermone \leqs \itermtwo$ and $\ivarone \not\in \Var{\itermone}$ then $\itermone \leqs \itermtwo\subst{\ivarone}{\itermthree}$ for every index term $\itermthree$.
  \end{varenumerate}
\end{lemma}
\subsection{Sized Types Subtyping and Type Checking}
The set of \emph{sized types} is given by annotating occurrences of base types 
in simple types with index terms $\itermone$, possibly introducing quantification over 
index variables.
More precise, the sets of \emph{(sized) monotypes}, \emph{(sized) polytypes} and \emph{(sized) types}
are generated from base types $\stbase$, index variables $\vec{\ivarone}$ 
and index terms $\itermone$ as follows:
\[
  \text{\textbf{(monotypes)}}\quad \mtpone,\mtptwo \bnfdef \base[\itermone] \mid \tpair{\mtpone}{\mtptwo} \mid \tpone \sarr \mtpone \!\tkom\quad
  \text{\textbf{(polytypes)}}\quad \ptpone \bnfdef \forall{\vec{\ivarone}}.\ \tpone \starr \mtpone \!\tkom\quad
  \text{\textbf{(types)}}\quad     \tpone \bnfdef \mtpone \mid \ptpone \tpkt
\]
Types $\base[\itermone]$ are called \emph{indexed base types}.
We keep the convention that the arrow binds stronger than quantification. 
Thus in a polytype $\forall{\vec{\ivarone}}.\ \tpone \starr \mtpone$ the variables $\vec{\ivarone}$ are bound in $\tpone$ and $\mtpone$. 
We will sometimes write a monotype $\mtpone$ as $\forall \seqempty.\ \mtpone$. This way, 
every type $\tpone$ can given in the form $\forall \vec{\ivarone}.\ \mtpone$. 
The \emph{skeleton} of a type $\tpone$ is the simple type obtained by dropping quantifiers and
indices.  The sets $\FPV(\cdot)$ and $\FNV(\cdot)$, of free variables
occurring in \emph{positive} and \emph{negative} positions,
respectively, are defined in the natural way\shortv{.}\longv{:
\begin{align*}
  \FPV(\base[\itermone]) & = \Var{\itermone} &
  \FNV(\base[\itermone]) & = \varnothing\\
  \FPV(\tpair{\mtpone}{\mtptwo}) & = \FPV(\mtpone) \cup \FPV(\mtptwo) & 
  \FNV(\tpair{\mtpone}{\mtptwo}) & = \FNV(\mtpone) \cup \FNV(\mtptwo) \\
  \FPV(\forall \vec{\ivarone} . \mtpone) & = \FPV(\mtpone) \setminus \{\vec{\ivarone}\} &
  \FNV(\forall \vec{\ivarone} . \mtpone) & = \FNV(\mtpone) \setminus \{\vec{\ivarone}\} \tpkt
\end{align*}}
\longshortv{The set of free variables $\FPV(\tpone) \cup \FNV(\tpone)$ in $\tpone$ is denoted by $\FV(\tpone)$.}
{The set of free variables in $\tpone$ is denoted by $\FV(\tpone)$.}
We consider types equal up to $\alpha$-equivalence. Index substitutions are extended to sized types
in the obvious way, using $\alpha$-conversion to avoid variable capture.

We denote by $\tpone \instantiates \mtpone$ that the monotype $\mtpone$ is obtained by \emph{instantiating} 
the variables quantified in $\tpone$ with arbitrary index terms, i.e.
if $\tpone = \forall{\vec{\ivarone}}.\mtptwo$ then
$\mtpone = \mtptwo\substvec{\ivarone}{\itermone}$ for some 
index terms $\vec{\itermone}$.
Notice that by our convention $\mtpone=\forall\seqempty.\ \mtpone$, 
we have $\mtpone \instantiates \mtpone$ for every monotype $\mtpone$.

\newcommand{\checkstbase}{\rl{\ensuremath{\subtypeof[{\stbase}]}}}
\newcommand{\checkstpair}{\rl{\ensuremath{\subtypeof[\times]}}}
\newcommand{\checkstarr}{\rl{\ensuremath{\subtypeof[\starr]}}}
\newcommand{\checkstforall}{\rl{\ensuremath{\subtypeof[\forall]}}}
\newcommand{\checkvarsd}{\rl{Var}}
\newcommand{\checkfunsd}{\rl{Fun}}
\newcommand{\checkappsd}{\rl{App}}
\newcommand{\checkletsd}{\rl{Let}}
\newcommand{\checkpairsd}{\rl{Pair}}

\begin{figure}[t]
  \begin{subfigure}[b]{1.0\linewidth}
    \centering
    \begin{framed}
      \[
        \begin{array}[b]{c@{\quad}c}
          \Infer[\checkstbase]
          {\base[\itermone] \subtypeof \base[\itermtwo]}
          {\itermone \leqs \itermtwo} 
          &
            \Infer[\checkstpair]
            {\tpair{\mtpone_1}{\mtpone_2} \subtypeof \tpair{\mtpone_3}{\mtpone_4}}
            {\mtpone_1 \subtypeof \mtpone_3 & \mtpone_2 \subtypeof \mtpone_4}
          \\[1mm]
          \Infer[\checkstarr]
          {\tpone_1 \sarr \mtpone_1 \subtypeof \tpone_2 \sarr \mtpone_2}
          {\tpone_2 \subtypeof \tpone_1 & \mtpone_1 \subtypeof \mtpone_2}
          & 
             \Infer[\checkstforall]
             {\forall\vec{\ivarone}.\mtpone_1 \subtypeof \tpone_2} 
             {\tpone_2 \instantiates \mtpone_2 & \mtpone_1 \subtypeof \mtpone_2 & \vec{\ivarone} \not\in\FV(\tpone_2) } 
      \end{array}
    \]
    \end{framed}
    \caption{Subtyping rules.}\label{fig:typecheck:subtype}
  \end{subfigure}
  \\[5mm]
  \begin{subfigure}[b]{1.0\linewidth}
    \centering
    \begin{framed}
      \[
        \begin{array}[b]{c@{\quad}c}
         \Infer[\checkvarsd]
               {\typedsd{\ctxone,\varone \oftype \tpone}{\varone}{\mtpone}}
               {\tpone \instantiates \mtpone} 
         &
         \Infer[\checkfunsd]
               {\typedsd{\ctxone}{s}{\mtpone}} 
               {s \in \FUN \cup \CON & s \decl \tpone & \tpone \instantiates \mtpone} 
         \\[1mm]                                               
         \Infer[\checkletsd]
            {\typedsd{\ctxone}{\letexp{\termone}{\varone_1}{\varone_2}{\termtwo}}{\mtpone}}
            {
              \typedsd{\ctxone}{\termone}{\tpair{\mtpone_1}{\mtpone_2}}
              & \typedsd{\ctxone, \varone_1 \oftype \mtpone_1, \varone_2 \oftype \mtpone_2}{\termtwo}{\mtpone}
            }
          & 
           \Infer[\checkpairsd]
              {\typedsd{\ctxone}{\pair{\termone_1}{\termone_2}}{\tpair{\mtpone_1}{\mtpone_2}}} 
              {\typedsd{\ctxone}{\termone_1}{\mtpone_1} & \typedsd{\ctxone}{\termone_2}{\mtpone_2}}
        \end{array}
      \]
      \[
        \Infer[\checkappsd]
        {\typedsd{\ctxone}{\termone \api \termtwo}{\mtpone}}
        {
          \typedsd{\ctxone}{\termone}{(\forall \vec{\ivarone}. \mtptwo_1) \sarr \mtpone} 
          & \typedsd{\ctxone}{\termtwo}{\mtptwo_2}
          & \mtptwo_2 \subtypeof \mtptwo_1
          & \vec{\ivarone} \not\in\FV(\restrictctx{\ctxone}{\FV(\termtwo)})}
      \]
  \end{framed}
    \caption{Typing rules}\label{fig:typecheck:type}
  \end{subfigure}
  \\[5mm]
  \caption{Typing and subtyping rules, depending on the semantic interpretation $\iinter$.}\label{fig:typecheck}
\end{figure}

The subtyping relation $\subtypeof$ is given in \Cref{fig:typecheck:subtype}. 
It depends on the interpretation of size indices, but otherwise is defined in the expected way.
Subtyping inherits the following properties from the relation $\leqs$, see Lemma~\ref{l:leqs}.
\begin{lemma}\label{l:subtyping}
  \envskipline
  \begin{varenumerate}
  \item\label{l:subtyping:refltrans}\label{l:subtyping:substclosed} 
    The subtyping relation is reflexive, transitive and closed under index substitutions.
  \item\label{l:subtyping:mon} 
    If $\itermone \leqs \itermtwo$
    then $\tpone\subst{\ivarone}{\itermone} \subtypeof \tpone\subst{\ivarone}{\itermtwo}$ 
    for all index variables $\ivarone \not\in \FNV(\tpone)$.
  \end{varenumerate}
\end{lemma}
\longv{%
\begin{proof}
  By a standard induction.
\qed
\end{proof}}

We are interested in certain linear types, namely those in which any
index term occurring in negative position is in fact a fresh index
variable.
\begin{definition}[Canonical Sized Type, Sized Type Declaration]\label{d:canonical}\envskipline
  \begin{varenumerate}
    \item 
      A monotype $\mtpone$ is \emph{canonical} if one of the following alternatives hold:
      \begin{varitemize}
        \item $\mtpone = \base[\itermone]$ is an indexed base type;
        \item $\mtpone = \tpair{\mtptwo_1}{\mtptwo_2}$ for two canonical monotypes $\mtptwo_1,\mtptwo_2$; 
        \item $\mtpone = \base[\ivarone] \sarr \mtptwo$ with $\ivarone \not\in \FNV(\mtptwo)$;
        \item $\mtpone = \ptpone \sarr \mtptwo$ for a canonical polytype $\ptpone$ and canonical type
          $\mtptwo$ with $\FV(\ptpone) \cap \FNV(\mtptwo) = \varnothing$.  
        \end{varitemize}
      \item 
      A polytype $\ptpone = \forall\vec{\ivarone}.\mtpone$ is \emph{canonical} if
      $\mtpone$ is canonical and $\FNV(\mtpone) \subseteq \{\vec{\ivarone}\}$. 
    \item       
      To each function symbol $s \in \FUN \cup \CON$, we associate a \emph{closed} and \emph{canonical} 
      type 
      $\tpone$ whose skeleton coincides with the simple type of $s$.
      We write $s \decl \tpone$ and call $s \decl \tpone$ the \emph{sized type declaration} of $s$.
    \end{varenumerate}
\end{definition}
Canonicity ensures that pattern matching can be resolved with a simple
substitution mechanism, rather than a sophisticated unification based
mechanism that takes the semantic interpretation $\iinter$ into
account. 
\longv{%
Observe that the above definition dictates that a function is given a sized type declaration of the form
\[
\forall{\vec{\ivarone}}.\ \ptpone_1 \starr \cdots \starr \cdots \starr \ptpone_k \starr \base[\itermone] \tkom
\]
all the variables occurring free in $\ptpone_i$ ($1 \leq l \leq k$) are pairwise disjoint. 
For instance, consider base types $\BN{}$ and $\BL{}$ represent integers and integer lists, respectively.
Then e.g.\ the type $\BN{\ivarone} \starr \BL{\ivartwo} \starr \BL{\itermone}$ for some 
index term $\itermone$ is canonical, provided that $\ivarone$ and $\ivartwo$ are distinct.
This type can then be turned in a canonical polytype, by quantifying (at least) over the two index variables $\ivarone$ and $\ivartwo$, 
resulting in a polytype
$\forall\vec{\ivarone}.\ \BN{\ivarone} \starr \BL{\ivartwo} \starr \BL{\itermone}$. 
Similar, the type 
\[
(\forall\vec{\ivarone}.\ \BN{\ivarone} \starr \BL{\ivartwo} \starr \BL{\itermone}) \starr \BL{\ivarthree} \starr \BL{\ivarfour} \starr \BL{\itermtwo}
\tkom
\]
is canonical, provided the two different index variables $\ivarthree$ and $\ivarfour$ are distinct from the variables occurring free in 
$(\forall\vec{\ivarone}.\ \BN{\ivarone} \starr \BL{\ivartwo} \starr \BL{\itermone})$. 
For instance, the polytype
\[
\forall \ivarfive \ivarthree \ivarfour.\ (\forall\ivarone \ivartwo.\ \BN{\ivarone} \starr \BL{\ivartwo} \starr \BL{\ivartwo+\ivarfive}) \starr \BL{\ivarthree} \starr \BL{\ivarfour} \starr \BL{(\ivarfour - 1)\cdot\ivarfive + \ivarthree}
\tkom
\]
is canonical. Note that this type corresponds to (a monomorphic copy) of the sized type given to $\fun{foldr}$ in Figure~\ref{fig:doublefilter} on page~\pageref{fig:doublefilter}.}
Canonical types enjoy the following substitution property.
\begin{lemma}\label{l:declaration:canonical}
  Let $\tpone$ be a canonical type and suppose that 
  $\ivarone \not\in \FNV(\tpone)$. Then $\tpone\subst{\ivarone}{\itermone}$ is again canonical.
\end{lemma}
\longv{
\begin{proof}
  The proof is by structural induction. The base case, where we consider a type $\base[\itermone]$, is trivial. 
  In the first inductive step, where we consider a type $\tpair{\mtpone_1}{\mtpone_2}$,
  we conclude directly from the  induction hypothesis.  In the second
  inductive step, we consider a type $\tpone \sarr \mtpone$ with both
  $\tpone$ and $\mtpone$ canonical.  When $\tpone$ is of base
  type, i.e. $\tpone = \base[\ivartwo]$, then canonicity implies
  $\ivarone \not\in \FNV(\mtpone)$ and hence $\ivarone \not= \ivartwo$.
  So $(\tpone \starr \mtpone)\subst{\ivarone}{\itermone} 
  = \base[\ivartwo] \sarr \mtpone\subst{\ivarone}{\itermone}$ 
  is canonical by induction hypothesis.  
  Otherwise, $\tpone$ is not a base type
  and by assumption $\ivarone \not\in \FNV(\tpone \starr \mtpone) =
  \FPV(\tpone) \cup \FNV(\mtpone) = \FV(\tpone) \cup
  \FNV(\mtpone)$.  Here, the former equality follows by definition,
  and the latter follows since $\tpone$ is canonical.  Thus
  $(\tpone \starr \mtpone)\subst{\ivarone}{\itermone} = \tpone \sarr \mtpone\subst{\ivarone}{\itermone}$.  
  Note that if $\ivarone$ occurs in $\mtpone_1$, then it does so positively by assumption.
  Hence $\FNV(\mtpone) = \FNV(\mtpone\subst{\ivarone}{\itermone})$ and thus
  $\FV(\tpone) \cap \FNV(\mtpone\subst{\ivarone}{\itermone}) = \FV(\tpone) \cap \FNV(\mtpone)
  = \emptyset$ by assumption that $\tpone \sarr \mtpone$ is
  canonical. The result follows then from induction hypothesis.
\end{proof}}

In \Cref{fig:typecheck:type} we depict the typing rules of our sized type system. 
A \emph{(typing) context} $\ctxone$ is a mapping 
from variables $\varone$ to types $\tpone$ so that the skeleton of $\tpone$ coincides with the simple type of $\varone$.
We denote the context $\ctxone$ that maps variables $\varone_i$ to $\tpone_i$ ($1 \leq i \leq n$)
by $\ctxseq{\varone}{\tpone}$. The empty context is denoted by $\emptyctx$.
We lift set operations as well as the notion of (positive, negative) free 
variables and application of index substitutions to contexts in the obvious way.
We denote by $\restrictctx{\ctxone}{X}$ the \emph{restriction} 
of context $\ctxone$ to a set of variables $X \subseteq \VAR$. 
The typing statement $\typedsd[\iinter]{\ctxone}{\termone}{\mtpone}$ states that under
the typing contexts $\ctxone$, the term $\termone$ has the \emph{monotype} $\mtpone$, 
when indices are interpreted with respect to $\iinter$.
The typing rules from \Cref{fig:typecheck:type} are fairly standard.
Symbols $s \in \FUN \cup \CON \cup \VAR$ are given instance types of their associated types. 
This way we achieve the desired degree of polymorphism outlined in Section~\ref{sect:ERW}.
Subtyping and generalisation are confined to function application, see rule~\checkappsd. 
Here, the monotype $\mtptwo_2$ of the argument term $\termtwo$ is weakened to $\mtptwo_1$, 
the side-conditions put on index variables $\vec{\ivarone}$ allow then a generalisation of $\mtptwo_1$ to $\forall\vec{\ivarone}. \mtptwo_1$,
the type expected by the function $\termone$.
This way, the complete system becomes syntax directed. 
We remark that subtyping is prohibited in the typing of the left spine of applicative terms.
\newcommand{\fpfun}{\rl{FpFun}}
\newcommand{\fpappvar}{\rl{FpAppVar}}
\newcommand{\fpappnvar}{\rl{FpAppNVar}}
\begin{figure}[t]
  \centering
  \begin{framed}
    \[
        \Infer[\fpfun]
          {\fpInfer{\emptyctx}{\funone}{\mtpone}}
          {\funone \decl \forall \vec{\ivarone}.\mtpone}
    \qquad\qquad\qquad\qquad 
        \Infer[\fpappvar]
          {\fpInfer{\ctxone \uplus \{\varone \oftype \tpone\}}{\termtwo \api \varone}{\mtpone}}
          {\fpInfer{\ctxone}{\termtwo}{\tpone \sarr \mtpone}}
      \]
      \[
      \Infer[\fpappnvar]
          {\fpInfer{\ctxone_1 \uplus \ctxone_2}{\termone \api \termtwo}{\mtpone\subst{\ivarone}{\itermone}}}
          {
            \begin{array}{c}
              (\FV(\ctxone_1) \cup \FV(\mtpone)) \cap (\FV(\ctxone_2) \cup \FV(\base[\itermone])) = \emptyset\\
              \fpInfer{\ctxone_1}{\termone}{\base[\ivarone] \sarr \mtpone} \qquad \fpInfer{\ctxone_2}{\termtwo}{\base[\itermone]}  \qquad \termone \not\in\VAR
            \end{array}
          }
       \]
  \end{framed}
  \caption{Rules for computing the footprint of a term.}
  \label{fig:footprint}
\end{figure}

Since our programs are equationally-defined, we need to define when
equations are well-typed. In essence, we will say that a
program $\progone$ is \emph{well-typed}, if, for all equations $l = r$, 
the right-hand side $r$ can be given a subtype of $l$.  Due to
polymorphic typing of recursion, and since our typing relation
integrates subtyping, we have to be careful.  Instead of giving $l$ an
arbitrary derivable type, we will have to give it a
\emph{most general type} that has not been weakened through subtyping. 
Put otherwise, the type for the equation, which is determined by $l$, should 
precisely relate to the declared type of the considered function.

To this end, we introduce the restricted
typing relation, the \emph{footprint relation}, depicted in
Figure~\ref{fig:footprint}.  The footprint relation makes essential
use of canonicity of sized type declaration and the shape of patterns.
\shortv{%
In particular, $\fpInfer{\ctxseq{\varone}{\tpone}}{\termone}{\mtpone}$ implies that
all $\tpone_i$ and $\mtpone$ are canonical.}
\longv{%
The following tells us that footprints guarantee canonicity of
the employed types:
\begin{lemma}\label{l:footprint:canonical}
  If $\fpInfer{\ctxseq{\varone}{\tpone}}{\termone}{\mtpone}$ 
  then all $\tpone_i$ and $\mtpone$ are canonical.
\end{lemma}
\begin{proof}
  Suppose $\fpInfer{\ctxone}{\termone}{\mtpone}$ for some context $\ctxone = \ctxseq{\varone}{\tpone}$.
  We proof the \emph{canonicity conditions} put by the lemma on $\ctxone$ and $\mtpone$ together with 
  the \emph{variable condition} $\FV(\ctxone) \cap \FNV(\mtpone) = \emptyset$, 
  by induction on the derivation of $\fpInfer{\ctxone}{\termone}{\mtpone}$.
  
  In the only base case we consider the application of rule~\fpfun, where $\termone = \funone \in \FUN \cup \CON$. 
  Using that $\funone \decl \forall \vec{\ivarone}.\mtpone$ implies that
  $\mtpone$ is canonical, the case follows. 
  In the first inductive step, we consider a derivation 
  \[
    \Infer[\fpappvar]
    {\fpInfer{\ctxone \uplus \{\varone \oftype \tpone\}}{\termtwo \api \varone}{\mtpone}}
    {\fpInfer{\ctxone}{\termtwo}{\tpone \sarr \mtpone}} \tpkt
  \]
  By induction hypothesis the type $\tpone \sarr \mtpone$ is canonical. 
  As thus both $\tpone$ and $\mtpone$ are canonical, the first part of the assertion follows.
  Since we get also $\FV(\ctxone) \cap \FNV(\mtpone) = \emptyset$ as a consequence of the induction hypothesis, 
  we see
  \begin{align*}
    \mparbox{1cm}{\FV(\ctxone \uplus \{\varone \oftype \tpone\}) \cap \FNV(\mtpone)} & \\
    & = (\FV(\ctxone) \cup \FV(\tpone)) \cap \FNV(\mtpone)\\
    & = (\FV(\ctxone) \cup \FPV(\tpone)) \cap \FNV(\mtpone)&& \text{(since $\tpone$ is canonical)}\\
    & = \FPV(\tpone) \cap \FNV(\mtpone) && \text{(by induction hypothesis)}\\
    & = \emptyset && \text{(since $\tpone \sarr \mtpone$ is canonical).}
  \end{align*}
  Thus the variable condition holds as desired.
  Let us now consider the final inductive step 
  \[
    \Infer[\fpappnvar]
    {\fpInfer{\ctxone_1 \uplus \ctxone_2}{\termone \api \termtwo}{\mtpone\subst{\ivarone}{\itermone}}}
    {
      \begin{array}{c}
        (\FV(\ctxone_1) \cup \FV(\mtpone)) \cap (\FV(\ctxone_2) \cup \FV(\base[\itermone])) = \emptyset\\
        \fpInfer{\ctxone_1}{\termone}{\base[\ivarone] \sarr \mtpone} \qquad \fpInfer{\ctxone_2}{\termtwo}{\base[\itermone]}  \qquad \termtwo \not\in\VAR
      \end{array}
    }
  \]
  By induction hypothesis on $\fpInfer{\ctxone_1}{\termone}{\base[\ivarone] \sarr \mtpone}$, 
  the type $\mtpone$ is canonical and $\ivarone \not\in \FNV(\mtpone)$. 
  The canonicity conditions then follow directly from Lemma~\ref{l:declaration:canonical} and induction hypothesis. 
  As we also have $\FNV(\mtpone) = \FNV(\mtpone\subst{\ivarone}{\itermone})$,  
  the variable condition follows by induction hypothesis on $\fpInfer{\ctxone_1}{\termone}{\base[\ivarone] \sarr \mtpone}$
  and the side conditions put on the rule.
\end{proof}}
%
The footprint relation can be understood as a function that, given a left-hand side $\funone \api \aseq[k]{\patone}$, 
results in a typing context $\ctxone$ and monotype $\mtpone$.
This function is total, for two reasons. First of all, the above lemma 
confirms that the term $\termone$ in rule \fpappnvar\ is given indeed a canonical type of the stated form.
Secondly, the disjointness condition required by this rule can always be satisfied 
via $\alpha$-conversion. It is thus justified 
to define $\footprint(\funone \api \aseq[k]{\patone}) \defsym (\ctxone,\mtpone)$ for some (particular) context $\ctxone$ and type $\mtpone$ 
that satisfies $\fpInfer{\ctxone}{\funone \api \aseq[k]{\patone}}{\mtpone}$. 

\begin{definition}
  Let $\progone$ be a program, such that every function and constructor has a declared sized type. 
  We call a rule $l = r$ from $\progone$ \emph{well-typed under the interpretation $\iinter$} if
  \[
    \fpInfer{\ctxone}{l}{\mtpone} \IImp \typedsd{\ctxone}{r}{\mtptwo} \textit{ for some monotype $\mtptwo$ with $\mtptwo \subtypeof \mtpone$,}
  \]
  holds for all contexts $\ctxone$ and types $\mtpone$.
  The program $\progone$ is \emph{well-typed under the interpretation $\iinter$} if 
  all its equations are.
\end{definition}

\subsection{Subject Reduction}
It is more convenient to deal with subject reduction when subtyping is \emph{not} confined 
to function application. We thus define the typability relation $\typed{\ctxone}{\termone}{\mtpone}$. 
It is defined in terms of all the rules depicted in \Cref{fig:typecheck:type}, together  with the following subtyping rule.
\newcommand{\checksubtype}{\rl{SubType}}
\[
  \Infer[\checksubtype]
  {\typed{\ctxone}{\termone}{\mtpone}}
  { \typed{\ctxone}{\termone}{\mtptwo} & \mtptwo \subtypeof \mtpone} 
\]
\longv{
\begin{lemma}\label{l:skeleton}
  If $\typed{\ctxone}{\termone}{\mtpone}$ then the simple type of $\termone$ corresponds to the skeleton of $\mtpone$.
\end{lemma}
\begin{proof}
  Note that typing contexts and type declarations assign sized types
  with suitable skeleton to variables and function symbols.  From this
  observation the lemma follows by induction on the derivation of
  $\typed{\ctxone}{\termone}{\mtpone}$.
\end{proof}
Observe that typing is closed under index substitutions in the following sense. 
\begin{lemma}\label{l:typing:substclosed}
  If $\typed{\ctxone}{\termone}{\mtpone}$ then
  $\typed{\ctxone\isubstone}{\termone}{\mtpone\isubstone}$ for any
  index substitution $\isubstone$.
\end{lemma}
\begin{proof}
  The lemma follows by induction on the derivation of $\typed{\ctxone}{\termone}{\mtpone}$. 
  In the case where $\termone$ is a function symbol we use that sized type declarations are closed.
\end{proof}
}%
%
As a first step towards subject reduction, we clarify that the footprint correctly 
accounts for pattern matching. 
Consider an equation $l = r \in \progone$ from a well-typed program $\progone$, 
where $\fpInfer{\ctxone}{l}{\mtptwo}$. 
If the left-hand side matches a term $\termone$ of type $\mtpone$, i.e. $\termone = l\substone$, 
then the type $\mtpone$ is an instance of $\mtptwo$, or a supertype thereof.
Moreover, the images of $\substone$ can all be typed as instances of the corresponding types in 
the typing context $\ctxone$. 
More precise:
\begin{lemma}[Footprint Lemma]\label{l:substitution:lhs}
  Let $\termone = \funone \api \patone_1 \api \cdots \api \patone_n$
  be a linear term over variables $\seq[m]{\varone}$, and let
  $\substone=\substseq[m]{\varone}{\termtwo}$ be a substitution.
  If $\gtyped{\termone\substone}{\mtpone}$ then there exist a context 
  $\ctxone = \ctxseq[m]{\varone}{\tpone}$ and a type $\mtptwo$ such that
  $\fpInfer{\ctxone}{\termone}{\mtptwo}$ holds.
  Moreover, for some index substitution $\isubstone$
  we have $\mtptwo\isubstone \subtypeof \mtpone$ and $\gtyped{\termtwo_n}{\mtpone_n\isubstone}$, 
  where $\tpone_n = \forall \vec{\ivarone}.\mtpone_n$ ($1 \leqslant n \leqslant m$).
\end{lemma}
\longv{%
\begin{proof}
  Suppose $\gtyped{\termone\substone}{\mtpone}$. 
  We prove the lemma by structural induction on $\termone$. 
  By the shape of $\termone$, it suffices to consider the base case $\termone \in \FUN \cup \CON$, 
  and the inductive cases $\termone = \termone_1 \api \varone$ and $\termone = \termone_1 \api \termone_2$ 
  where $\termone_2$ is of base type. 

  In the base case, $\termone\substone = \termone \in \FUN \cup \CON$ where $\termone \decl \forall\vec{\ivarone}.\mtptwo$ for 
  some type $\mtptwo$. It is then not difficult to see that the assumption $\gtyped{\termone}{\mtpone}$ 
  yields an index substitution $\isubstone$ with domain $\vec{\ivarone}$ with $\mtptwo\isubstone \subtypeof \mtpone$.
  As in the considered case $m=0$, the lemma follows.

  In the first inductive step, we consider the case $\termone = \termone_1 \api \varone$. 
  Consider a derivation of $\gtyped{\termone\substone}{\mtpone}$. Then, possibly pushing 
  applications of rule $\checksubtype$ inwards, wlog. we can assume that this derivation 
  ends in an application of rule $\checkappsd$, and hence 
  $\gtyped{\termone_1\substone}{(\forall \vec{\ivartwo}. \mtpone_\varone) \sarr \mtpone}$
  and $\gtyped{\varone\substone}{\mtpone_\varone}$ holds for some type $\mtpone_\varone$, possibly also weakening the 
  type of $\varone\substone$ with an application of rule \checksubtype.
  By induction hypothesis on $\termone_1$, there exist a context $\ctxone_1 = \ctxseq[m]{\varone}{\tpone}$
  and a type $\mtptwo_1$ together with an index substitution $\isubstone_1$ such that 
  (i)~$\fpInfer{\ctxone_1}{\termone_1}{\mtptwo_1}$,
  (ii)~$\mtptwo_1\isubstone_1 \subtypeof (\forall \vec{\ivartwo}. \mtpone_\varone) \sarr \mtpone$, and
  (iii)~$\gtyped{\varone_n\substone}{\mtpone_n\isubstone_1}$ for the type $\mtpone_n$ 
  with $\tpone_n = \forall \vec{\ivarthree}. \mtpone_n$ ($1 \leqslant n \leqslant m$). 
  From~(ii) and definition of the sub-typing relation, 
  we see that $\mtptwo_1 = (\forall \vec{\ivarthree}. \mtptwo_\varone) \sarr \mtptwo$ 
  for some type $\forall \vec{\ivarthree}. \mtptwo_\varone$ and type $\mtptwo$,
  were (iv)~$\mtptwo\isubstone_1 \subtypeof \mtpone$
  and moreover $\forall \vec{\ivarthree}. \mtptwo_\varone\isubstone_1$
  instantiates to a supertype of $\mtpone_\varone$. 
  More precise, there exists an index substitution $\isubstone_\varone$ with domain $\vec{\ivarthree}$, 
  such that (v)~$\mtpone_\varone \subtypeof (\mtptwo_\varone\isubstone_1)\isubstone_\varone$.

  We claim that the lemma is satisfied by taking the type $\mtptwo$ together with the 
  context $\ctxone$ and index substitution $\isubstone$ defined as follows:
  \begin{align*}
    \isubstone(\ivarone) & \defsym
    \begin{cases}
      \isubstone_\varone(\isubstone_1(\ivarone)) & \text{if $\ivarone \in \vec{\ivartwo}$,}\\
      \isubstone_1(\ivarone) & \text{if $\ivarone \not \in \vec{\ivartwo}$;}
    \end{cases}
    & \ctxone & \defsym \ctxone_1 \uplus \{ \varone \oftype \forall \vec{\ivarthree}. \mtptwo_\varone \}\tpkt
  \end{align*}
  Clearly, by~(i) and rule $\fpappvar$ we have $\fpInfer{\ctxone}{\termone_1 \api \varone}{\mtptwo}$. 
  Concerning the remaining properties, 
  first observe that without loss of generality, the variables $\vec{\ivartwo}$ are fresh, 
  i.e. do neither occur in $\mtptwo$, nor in the images of $\isubstone_1$ and $\ctxone_1$. 
  Consequently, $\mtptwo\isubstone = \mtptwo\isubstone_1$ and 
  $\mtpone_n\isubstone=\mtpone_n\isubstone_1$ for each type $\mtpone_n$ 
  with $\ctxone_1(\varone_n) = \forall\vec{\ivarone}.\mtpone_n$ ($1 \leq n \leq m$). 
  Together with (iv) the former equality proves $\mtptwo\isubstone \subtypeof \mtpone$, 
  together with (iii) the latter proves $\gtyped{\varone_n}{\mtpone_n\isubstone}$ ($1 \leq i \leq n$). 
  As on the other also hand $\mtptwo_\varone\isubstone = (\mtptwo_\varone\isubstone_1)\isubstone_\varone$, 
  the derivation $D_\varone$ together with~(v) proves $\gtyped{\varone\substone}{\mtptwo_\varone\isubstone}$ by one application of rule~\checksubtype.
  
  In the second and final inductive case, we consider 
  $\termone = \termone_1 \api \termone_2$ where $\termone_2$ is a non-variable pattern of base type. 
  Let $\seq[o]{\varone}$ and $\seq[o+1][m]{\varone}$ denote the variables of $\termone_1$ and $\termone_2$, respectively. 
  Note that by linearity of $\termone$, these variables are pairwise disjoint. 
  An inference of $\gtyped{\termone\substone}{\mtpone}$ then wlog.\ again ends in an application of rule 
  $\checkappsd$, employing Lemma~\ref{l:skeleton} we see that
  $\gtyped{\termone_1\substone}{\base[\itermone] \sarr \mtpone}$ and 
  $\gtyped{\termone_2\substone}{\base[\itermone]}$ holds for some index term $\itermone$.
  As a consequence of the IH on $\termone_1$ and Lemma~\ref{l:footprint:canonical}, we obtain
  a context $\ctxone_1$ over $\seq[o]{\varone}$, type $\base[\ivarone] \sarr \mtptwo$ and index substitution $\isubstone_1$
  satisfying 
  (i)~$\fpInfer{\ctxone_1}{\termone_1}{\base[\ivarone] \sarr \mtptwo}$ 
  for a simple type $\base[\ivarone] \sarr \mtptwo$,
  (ii)~$\base[\itermone] \subtypeof \base[i]\isubstone_1$, 
  (iii)~$\mtptwo\isubstone_1 \subtypeof \mtpone$ and 
  (iv)~$\gtyped{\varone_n\substone}{\mtpone_n\isubstone_1}$ for the 
  type $\mtpone_n$ with $\ctxone_1(\varone_n) = \forall \vec{\ivarone}. \mtpone_n$ 
  ($1 \leqslant n \leqslant o$). 
  As a consequence of the induction hypothesis on $\termone_2$, we obtain 
  a context $\ctxone_2$  over $\seq[o+1][m]{\varone}$, 
  type $\base[\itermtwo]$ index substitution $\isubstone_2$
  satisfying 
  (v)~$\fpInfer{\ctxone_1}{\termone_2}{\base[\itermtwo]}$,
  (vi)~$\base[\itermtwo]\isubstone_2 \subtypeof \base[\itermone]$ and 
  (vii)~$\gtyped{\varone_n\substone}{\mtpone_n\isubstone_1}$ for the 
  type $\mtpone_n$ with $\ctxone_2(\varone_n) = \forall \vec{\ivarone}. \mtpone_n$ 
  ($o+1 \leqslant n \leqslant m$). 
  Wlog.\ we can assume that (viii) free index variables in 
  $\ctxone_1,\mtptwo_1$ and $\ctxone_2,\base[\itermtwo]$, i.e.
  the domains of $\isubstone_1$ and $\isubstone_2$, are disjoint.
  We claim that the lemma is satisfied by taking the type $\mtptwo\subst{\ivarone}{\itermtwo}$
  together with the context $\ctxone$ and index substitution $\isubstone$ defined as follows:
  \begin{align*}
    \isubstone(\ivartwo) & \defsym 
    \begin{cases}
      \isubstone_1(\ivartwo) & \text{if $\ivartwo \in \dom(\isubstone_1)$,} \\
      \isubstone_2(\ivartwo) & \text{if $\ivartwo \in \dom(\isubstone_2)$;}
    \end{cases}
    & \ctxone & \defsym \ctxone_1 \uplus \ctxone_2 \tpkt
  \end{align*}
  Then $\fpInfer{\ctxone}{\termone}{\mtptwo\subst{\ivarone}{\itermtwo}}$ follows directly from~(i),~(iv) and~(viii) by an application of rule~\fpappnvar. 
  
  Observe that as a consequence of (vi) and (ii) we 
  have $\itermtwo\isubstone_2 \subtypeof \itermone$ 
  and $\itermone \subtypeof \isubstone_1(\ivarone)$, respectively. 
  Since $\leqs$ is transitive by Lemma \eref{l:leqs}{refltrans}
  it follows that $\itermtwo\isubstone_2 \leqs \ivarone\isubstone_1$ holds. 
  Note that we can also assume that the index variable
  $\ivarone$ is fresh, in particular, does not occur in images of $\isubstone_1$.
  Then 
  \begin{align*}
    \mtptwo\subst{\ivarone}{\itermtwo}\isubstone
    & = \mtptwo\isubstone_1\subst{\ivarone}{\itermtwo\isubstone_2} 
    && \text{(as $\ivarone$ is fresh, and using~(viii))} \\
    & \subtypeof \mtptwo\isubstone_1\subst{\ivarone}{\isubstone_1(\ivarone)}
    && \text{(using Lemma~\eref{l:subtyping}{mon} and, by (i), $\ivarone \not\in \FNV(\mtptwo)$)}\\
    & = \mtptwo\isubstone_1 \tpkt
  \end{align*}
  From this and~(iii), by transitivity of the subtyping-relation (Lemma~\eref{l:subtyping}{refltrans}), 
  we thus conclude $\mtptwo\subst{\ivarone}{\itermtwo}\isubstone \subtypeof \mtpone$. 
  Concerning the remaining point, we fix a variable $\varone_n$ from $\termone$. 
  We consider first the case that $\varone_n$ occurs in $\termone_1$. Let $\mtpone_n$ be such that 
  $\ctxone(\varone_n) = \ctxone_1(\varone_n) = \forall \vec{\ivarone}. \mtpone_n$. 
  As by (viii) we have $\mtpone_n\isubstone = \mtpone_n\isubstone_1$, we conclude 
  $\gtyped{\varone_n\substone}{\mtpone_n\isubstone}$ as desired directly from~(iii). 
  Finally, the case where $\varone_n$ is from $\termone_2$ is handled symmetrically. 
  This finishes the proof. 
  \qed
\end{proof}}

\noindent The following constitutes the main lemma of this section, the
\emph{substitution lemma}:
\[
\gtyped{\termone_n}{\mtpone_n}~(1 \leq n \leq m) \text{ and }
\typed{\ctxseq[m]{\varone}{\mtpone}}{\termone}{\mtpone}
\quad\Rightarrow\quad \gtyped{\termone\substseq[m]{\varone}{\termone}}{\mtpone} \tpkt
\]
Indeed, we prove a generalisation. 
\begin{lemma}[Generalised Substitution Lemma]\label{l:substitution:rhs}
  Let $\termone$ be a term with free variables $\seq[m]{\varone}$, let $\ctxone$ be a context over $\seq[m]{\varone}$, 
  and let $\isubstone$ be an index substitution.
  If $\typed{\ctxone}{\termone}{\mtpone}$ for some type $\mtpone$ and 
  $\gtyped{\varone_n\substone}{\mtpone_n\isubstone}$ holds
  for the type $\mtpone_n$ with $\ctxone(\varone_n) = \forall \vec{\ivarone}. \mtpone_n$ ($1 \leqslant n \leqslant m$),
  then $\gtyped{\termone\substone}{\mtpone\isubstone}$.
\end{lemma}
\longv{
\begin{proof}
  We prove the following stronger property.
  Suppose $\typed{\ctxone \uplus \ctxtwo}{\termone}{\mtpone}$
  and let $\substone$ be a substitution over the variables defined by $\ctxone$. 
  Furthermore, suppose that $\gtyped{\varone_n\substone}{\mtpone_n\isubstone}$ holds
  for the type $\mtpone_n$ with $\ctxone(\varone_n) = \forall \vec{\ivarone}. \mtpone_n$.
  Then $\typed{\ctxtwo\isubstone}{\termone\substone}{\mtpone\isubstone}$.
  Note that from this claim the lemma follows by taking $\ctxtwo = \emptyctx$. 
  The proof is by induction on the typing derivation of $\typed{\ctxone \uplus \ctxtwo}{\termone}{\mtpone}$.

  In the first base case, we consider 
  \[
    \Infer[\checkfunsd]
    {\typedsd{\ctxone \uplus \ctxtwo}{s}{\mtpone}} 
    {s \in \FUN \cup \CON & s \decl \tpone & \tpone \instantiates \mtpone}        
    \tpkt
  \]
  Hence also $\typed{\ctxtwo}{s}{\mtpone}$ by rule $\checkfunsd$ 
  and the claim follows, as $\termone\substone=s$, from Lemma~\ref{l:typing:substclosed}.
  In the second base case, we consider the typing
  \[
    \Infer[\checkvarsd]
    {\typedsd{\ctxone \uplus \ctxtwo}{\varone}{\mtpone}}
    {(\ctxone \uplus \ctxtwo)(\varone) \instantiates \mtpone} \tpkt
  \]
  We prove $\typed{\ctxtwo\isubstone}{\varone\substone}{\mtpone\isubstone}$ and consider two sub-cases. 
  In the first case, $\varone \in \dom(\substone)$,
  and thus $\mtpone = \mtpone_\varone\isubstone_\varone$ for a 
  type $\mtpone_\varone$ with $\ctxone(\varone) = \forall\vec{\ivarone}. \mtpone_\varone$ and 
  index substitution $\isubstone_\varone$ with domain $\vec{\ivarone}$. 
  Let $\bar{\isubstone} = \isubstone \compose \isubstone_\varone$, with domain $\vec{\ivarone}$. 
  Via suitable $\alpha$-conversion, we can assume that the variables $\vec{\ivarone}$ do neither occur as images of the index substitution $\isubstone$ nor of the context $\ctxtwo$. 
  Thus $(\mtpone_\varone\isubstone)\bar{\isubstone} = (\mtpone_\varone\isubstone_\varone)\isubstone = \mtpone\isubstone$.
  Using the assumption $\gtyped{\varone\substone}{\mtpone_\varone\isubstone}$
  together with Lemma~\ref{l:typing:substclosed}, we conclude 
  $\gtyped{\varone\substone}{(\mtpone_\varone\isubstone)\bar{\isubstone}}$, 
  by the equality on types we obtain $\gtyped{\varone\substone}{\mtpone\isubstone}$, 
  from which the case follows easily. 
  In the second sub-case we consider $\varone \not\in \dom(\substone)$. 
  As then $\varone\substone = \varone$, we have
  $\typed{\ctxtwo}{\varone\substone}{\mtpone}$ by rule~\checkvarsd\ 
  and conclude $\typed{\ctxtwo\isubstone}{\varone\substone}{\mtpone\isubstone}$ 
  by Lemma~\ref{l:typing:substclosed}.

  In the first inductive step, we consider a typing derivation 
  \[
        \Infer[\checkappsd]{\typed{\ctxone\uplus\ctxtwo}{\termone \api \termtwo}{\mtpone}}
                   {\infer*{\typed{\ctxone\uplus\ctxtwo}{\termone}{(\forall \vec{\ivarone}. \mtptwo) \sarr \mtpone}}{D_1} 
                    & \infer*{\typed{\ctxone\uplus\ctxtwo}{\termtwo}{\mtptwo}}{D_2}
                    & \vec{\ivarone} \not\in\FV(\restrictctx{\ctxone\uplus\ctxtwo}{\FV(\termtwo)})}
  \]
  where wlog.\ the index variables $\vec{\ivarone}$ neither occur in the domain nor in the 
  images of $\isubstone$, potentially applying Lemma~\ref{l:typing:substclosed} on the derivation $D_2$.
  The IH on $D_1$ yields thus a type derivation $E_1$ of the 
  judgments $\typed{\ctxtwo\isubstone}{\termone\substone}{(\forall \vec{\ivarone}. \mtptwo\isubstone) \sarr \mtpone\isubstone}$. 
  Furthermore, the IH on $D_2$ yields a derivation $E_2$ of the 
  judgement $\typed{\ctxtwo\isubstone}{\termtwo\substone}{\mtptwo\isubstone}$. 
  \[
        \Infer[\checkappsd]{\typed{\ctxtwo\isubstone}{(\termone \api \termtwo)\substone}{\mtpone\isubstone}}
                   {\infer*{\typed{\ctxtwo\isubstone}{\termone\substone}{(\forall \vec{\ivarone}. \mtptwo\isubstone) \sarr \mtpone\isubstone}}{E_1} 
                    & \infer*{\typed{\ctxtwo\isubstone}{\termtwo\substone}{\mtptwo\isubstone}}{E_2}
                    & \vec{\ivarone} \not\in\FV(\restrictctx{\ctxtwo\isubstone}{\FV(\termtwo)})} \tpkt
  \]

  In the second inductive step, we consider a typing derivation 
    \[
      \Infer[\checkletsd]
      {\typed{(\ctxone\uplus\ctxtwo)}{\letexp{\termone}{\varone_1}{\varone_2}{\termtwo}}{\mtpone}}
      {
        \infer*{\typed{\ctxone\uplus\ctxtwo}{\termone}{\tpair{\mtpone_1}{\mtpone_2}}}{D_1}
        & \infer*{\typed{\ctxone\uplus\ctxtwo, \varone_1 \oftype \mtpone_1, \varone_2 \oftype \mtpone_2}{\termtwo}{\mtpone}}{D_2}
      }
    \]
  The induction hypothesis on $D_1$ and $D_2$ yield derivations $E_1$ and $E_2$ of the 
  judgments $\typed{\ctxtwo\isubstone}{\termone\substone}{\tpair{\mtpone_1\isubstone}{\mtpone_2\isubstone}}$ 
  and $\typed{\ctxtwo\isubstone,\varone_1\oftype\mtpone_1\isubstone,\varone_2\oftype\mtpone_2\isubstone}{\termtwo\substone}{\mtpone\isubstone}$, respectively. 
  Assuming that the variables $\varone_1,\varone_2$ are renamed apart from the variables in the domain of $\substone$, the IH immediately yields
  $\typed{\ctxtwo\isubstone}{(\letexp{\termone}{\varone_1}{\varone_2}{\termtwo})\substone}{\mtpone\isubstone}$ by one application of rule~\checkletsd. 

  Similar, the case where the typing derivation ends in an application of rule~\checkpairsd\ or rule~\checksubtype\ follows 
  directly from IH.\@ Just in the latter case we additionally use Lemma~\eref{l:subtyping}{substclosed}.
\end{proof}
}%
\noindent The combination of these two lemmas is almost all we need to reach our goal. 
\begin{theorem}[Subject Reduction]\label{t:subred}
  Suppose $\progone$ is well-typed under $\iinter$.
  If $\gtyped{\termone}{\mtpone}$ and $\termone \rew[\progone] \termtwo$ then 
  $\gtyped{\termtwo}{\mtpone}$.
\end{theorem}
\longv{%
\begin{proof}
  The proof is by induction on the evaluation context $E$ underlying the step $\termone \rew[\progone] \termtwo$. 
  In the base case $E = \hole$ we consider two cases. 
  In the first case, we consider 
  \[
    \termone = l\substvec{\varone}{\valone} \rew[\progone] r\substvec{\varone}{\valone} = \termtwo 
    \tkom
  \]    
  for $l = r \in \progone$. Wlog. the variables $\vec{\varone}$ are precisely the variables occurring in $l$. 
  Assume $\gtyped{\termone}{\mtpone}$, and thus Lemma~\ref{l:substitution:lhs} yields a context $\ctxone$ 
  and index substitution $\isubstone$ as well as a type $\mtptwo$ 
  such that (i)~$\fpInfer{\ctxone}{\termone}{\mtptwo}$ and 
  (ii)~$\mtptwo\isubstone \subtypeof \mtpone$ hold. 
  Moreover, for each variable $\varone$ occurring in $l$, we have 
  (iii)~$\gtyped{\varone\substone}{\mtpone_\varone\isubstone}$, 
  where $\mtpone_\varone$ is such that 
  $\ctxone(\varone) = \forall \vec{\ivarone}. \mtpone_\varone$. 
  Note that~(i) together with the typability condition on $\progone$ 
  yields $\typed{\ctxone}{r}{\mtptwo}$. 
  Thus, by~(iii) and Lemma~\ref{l:substitution:rhs}, 
  we conclude $\gtyped{\termtwo}{\mtptwo\isubstone}$. 
  Then $\gtyped{\termtwo}{\mtpone}$ follows 
  from~(ii) by one application of rule~\checksubtype.

  In the second case we consider 
  \[
    \termone = {\letexp{\pair{\termone_1}{\termone_2}}{\varone_1}{\varone_2}{\termone_3}\rew[\progone] \termone_3\subst{\varone_1,\vartwo_2}{\termone_1,\termone_2}} = \termtwo
  \]
  Assume $\gtyped{\termone}{\mtpone}$. By distributing the subtyping rule over rule \checkletsd, 
  a derivation of $\gtyped{\termone}{\mtpone}$ has the following form:
  \[
      \Infer[\checkletsd]
      {\gtyped{\letexp{\pair{\termone_1}{\termone_2}}{\varone_1}{\varone_2}{\termone_3}}{\mtpone}}
      {
        \Infer[\checkpairsd]
        {\gtyped{\pair{\termone_1}{\termone_2}}{\tpair{\mtpone_1}{\mtpone_2}}}
        { 
          \infer*{\gtyped{\termone_1}{\mtpone_1}}{D_1}
          & \infer*{\gtyped{\termone_2}{\mtpone_2}}{D_2}
        }
        & \infer*{\typed{\varone_1 \oftype \mtpone_1, \varone_2 \oftype \mtpone_2}{\termone_3}{\mtpone}}{E}
      }
  \]
  From the derivations $D_1,D_2$ and $E$, and using as index substitution $\isubstone$ the identity,
  Lemma~\ref{l:substitution:rhs} yields
  $\gtyped{\termone_3\subst{\varone_1,\varone_2}{\termone_1,\termone_2}}{\mtpone}$ as desired.

  As the remaining cases follow directly from induction hypothesis, we conclude the theorem.
\end{proof}
}
But what does Subject Reduction tells us, besides guaranteeing that types
are preserved along reduction? Actually, a lot: If 
$\gtyped{\termone}{\base[\itermone]}$, we are now sure that the evaluation
of $\termone$, if it terminates, would lead to a value of size at most
$\interpretation[\iinter]{\itermone}$. 
Of course, this requires that we give (first-order) \emph{data-constructors} a suitable sized type. 
To this end, let us call a sized type \emph{additive} if it is of the form
$\forall \vec{\ivarone}.\ \base[\ivarone_1] \starr \cdots \starr \base[\ivarone_k] \starr \base[\isucc(\ivarone_1 + \dots + \ivarone_k)]$.
\begin{corollary}\label{cor:size}
  Suppose $\progone$ is well-typed under the interpretation $\iinter$, 
  where data-constructors are given an additive type. 
  Suppose the first-order function $\fun{main}$ has type 
  $\forall{\vec{\ivarone}}. \base[\ivarone_1] \starr \cdots \starr \base[\ivarone_k] \starr \base[\itermone]$. 
  Then for all inputs $\seq{\dataone}$, if
  $\fun{main} \api \aseq[k]{\dataone}$ reduces to a data value $\dataone$,
  then the size of $\dataone$ is bounded by $s(\size{\dataone_1},\dots,\size{\dataone_k})$, 
  where $s$ is the function 
  $s(\seq[k]{\ivarone}) = \interpret[\iinter]{\itermone}$.
\end{corollary}
As we have done in the preceding examples, the notion of additive sized type could be suited so that 
constants like the list constructor $\cnil$ are attributed with a size of zero.
Thereby, the sized type for lists would reflect the length of lists. 
%
Note that the corollary by itself, does not mean much about the
\emph{complexity} of evaluating $\termone$. 
We will return to this in Section~\ref{sect:TTTCA}.


\section{Sized Types Inference}\label{sect:STI}
The kind of rich type discipline we have just introduced cannot
be enforced by requiring the programmer to annotate programs with
size types, since this would simply be too burdensome. Studying to
which extent types can be inferred, then, is of paramount importance.

We will now describe a type inference procedure that, given a program,
produces a set of first-order constraints that are satisfiable
\emph{iff} the term is size-typable.  At the heart of this procedure
lies the idea that we turn the typing system from \Cref{fig:typecheck}
into a system that, instead of checking, \emph{collects} all constraints
$\itermone \leqc \itermtwo$ put on indices.  These constraints are
then resolved in a second stage. The so obtained solution can then be
used to reconstruct a typing derivation with respect to the system
from Figure~\ref{fig:typecheck}.  As with any higher-ranked
polymorphic type system, the main challenge here lies in picking
suitable types instances from polymorphic types. In our system, this
concerns rules \checkvarsd\ and \checkfunsd.  Systems used in
practice, such as the one of \citet{JVWS:JFP:07}, use a combination of
forward and backward inference to determine suitable instantiated
types. Still, the resulting inference system is incomplete.  In our
sized type system, higher-ranked polymorphism is confined to size
indices.  This, in turn, allows us to divert the choice to the solving
stage, thereby retaining relative completeness.
To this end, we introduce meta variables $\sivarone$ in our index language.
Whereas in the typing system from Figure~\ref{fig:typecheck} index variables $\ivarone$
are instantiated by concrete index terms $\itermone$, our inference 
system uses a fresh meta variable $\sivarone$ as placeholder for $\itermone$. 
A suitable assignments to $\sivarone$ will be determined in the constraint solving stage.
A minor complication arises as we will have to introduce additional constraints $\noccur{\ivarone}{\sivarone}$ on meta variables $\sivarone$ that condition the set of terms $\sivarone$ may represent. 
This is necessary to deal with the side conditions on free variables, exhibited by the subtyping 
relation as well as in typing the rule for application.
All of this is made precise in the following.

\subsection{First- and Second-order Constraint Problems}\label{ssec:SOCPS}
As a first step towards inference, we introduce metavariables to our
index language. Let $\SOVARS$ be a countably infinite
set of \emph{second-order index variables}, which stand for arbitrary
index terms.  Second-order index variables are denoted by
$\sivarone,\sivartwo,\dots$\,.  The set of \emph{second-order index terms} 
is then generated over the set of index variables $\ivarone
\in \IVARS$, the set of second-order index variables $\sivarone \in
\SOVARS$ and index symbols $\ifunone \in \IS$ as follows.
\begin{align*}
  & \text{\textbf{(second-order index terms)}} & \sitermone,\sitermtwo \bnfdef {} \ivarone \mid \sivarone \mid \ifunone(\seq[\arity{\ifunone}]{\sitermone}) \tpkt
\end{align*}
We denote by $\Var{\sitermone} \subset \IVARS$ the set of (usual) index variables, and by $\SOVar{\sitermone} \subset \SOVARS$ the set of second-order index 
variables occurring in $\sitermone$. 

\begin{definition}[Second-order Constraint Problem, Model]
  A \emph{second-order constraint problem} $\constrone$ (\emph{SOCP} for short) 
  is a set of 
  (i)~\emph{inequality constraints} of the form $\sitermone \leqc \sitermtwo$ and 
  (ii)~\emph{occurrence constraints} of the form $\noccur{\ivarone}{\sivarone}$. 
  Let $\sisubstone$ be a substitution from second-order index variables to 
  first-order index terms $\itermone$, i.e. $\SOVar{\itermone} = \emptyset$.
  Furthermore, let $\iinter$ be an interpretation of $\IS$. 
  Then \emph{$(\iinter,\sisubstone)$ is a model of $\constrone$}, 
  in notation $\sosat{\iinter}{\sisubstone}{\constrone}$, if 
  (i)~$\sitermone\sisubstone \leqs[\iinter] \sitermtwo\sisubstone$ holds for all 
  inequalities $\sitermone \leqc \sitermtwo \in \constrone$; and 
  (ii)~$\ivarone \not\in \Var{\sisubstone(\sivarone)}$ for each occurrence constraint 
    $\noccur{\ivarone}{\sivarone}$.
\end{definition}
We say that $\constrone$ is \emph{satisfiable} if it has a model $(\iinter,\sisubstone)$.
The term $\sisubstone(\sivarone)$ is called the \emph{solution} of $\sivarone$.
We call $\constrone$ a \emph{first-order constraint problem} (\emph{FOCP} for short) if 
none of the inequalities $\sitermone \leqc \sitermtwo$ contain a 
second-order variable. Note that satisfiability of a FOCP $\constrone$ depends only 
on the semantic interpretation $\iinter$ of index functions. 
It is thus justified that FOCPs $\constrone$ contain no occurrence constraints.
We then write $\fosat{\iinter}{\constrone}$ if $\iinter$ models $\constrone$. 

SOCPs are very much suited to our inference machinery. In contrast, 
satisfiability of FOCPs is a re-occurring problem in various fields. 
To generate models for SOCPs, we will reduce satisfiability of SOCPs to the one of FOCPs.
This reduction is in essence a form of \emph{skolemization}. 

\paragraph*{Skolemization.}
Skolemization is a technique for eliminating existentially quantified 
variables from a formula. A witness for an existentially quantified variable can be given 
as a function in the universally quantified variables, the \emph{skolem function}. 
We employ a similar idea in our reduction of satisfiability from SOCPs to FOCPs, 
which substitutes second-order variables $\sivarone$ by \emph{skolem term} $\skolemfun{\sivarone}(\vec{\ivarone})$, 
for a unique \emph{skolem function} $\skolemfun{\sivarone}$, and where the sequence of variables $\vec{\ivarone}$ 
over-approximates the index variables of possible solutions to $\sivarone$. 
The over-approximation of index variables is computed by a simple fixed-point construction, 
guided by the observation that a solution of $\sivarone$ contains wlog.\ 
an index variable $\ivarone$ only when (i)~$\ivarone$ is related to $\sivarone$ in an inequality 
of the SOCP $\constrone$ and (ii)~the SOCP does not require $\noccur{\ivarone}{\sivarone}$.
Based on these observations, skolemization is formally defined as follows. 
\begin{definition}\label{d:skolem}
Let $\constrone$ be a SOCP.\@  
\begin{varenumerate}
\item\label{d:skolem:vars}
  For each second-order variable $\sivartwo$ of $\constrone$, 
  we define the sets $\skolemvarsaux{\sivartwo} \subset \IVARS$ of 
  \emph{index variables related to $\sivartwo$ by inequalities} as the least set satisfying,
  for each $(\sitermone \leqc \sitermtwo) \in \constrone$ with $\sivartwo \in \SOVar{\sitermtwo}$, 
  (i)~$\Var{\sitermone} \subseteq \skolemvarsaux{\sivartwo}$; and 
  (ii) $\skolemvarsaux{\sivarone} \subseteq \skolemvarsaux{\sivartwo}$ whenever $\sivarone$ occurs in $\sitermone$.
  The set of \emph{skolem variables} for $\sivartwo$ is then given by
  $\skolemvars{\sivartwo}
    \defsym \skolemvarsaux{\sivartwo} \setminus \{ \ivarone \mid (\noccur{\ivarone}{\sivartwo}) \in \constrone \}$.
\item For each second-order variable $\sivarone$ of $\constrone$, let $\skolemfun{\sivarone}$ be a fresh index symbol, 
  the \emph{skolem function} for $\sivarone$. 
  The arity of $\skolemfun{\sivarone}$ is the cardinality of $\skolemvars{\sivarone}$.
  The \emph{skolem substitution} $\skolemsubst$ is given by 
  $\skolemsubst(\sivarone) \defsym \skolemfun{\sivarone}(\seq[k]{\ivarone}) \text{ where $\skolemvars{\sivarone} = \sseq[k]{\ivarone}$}$.
\item We define the \emph{skolemization} of $\constrone$ by
  $\skolemize(\constrone)
    \defsym \{ \sitermone\skolemsubst \leqc \sitermtwo\skolemsubst \mid \sitermone \leqc \sitermtwo \in \constrone\}$.
\end{varenumerate}
\end{definition}

Note that the skolem substitution $\skolemsubst$ satisfies by definition all occurrence constraints of $\constrone$. 
Thus skolemization is trivially sound:
$\fosat{\iinter}{\skolemize(\constrone)}$ implies $\sosat{\iinter}{\skolemsubst}{\constrone}$.
Concerning completeness, the following lemma provides the central observation. 
Wlog.\ a solution to $\sivarone$ contains only variables of $\skolemvars{\sivarone}$:
\begin{lemma}\label{l:so2focs}
  Let $\constrone$ be a SOCP with model $(\iinter,\sisubstone)$.
  Then there exists a restricted second-order substitution $\sisubstone_r$ such that
  $(\iinter,\sisubstone_r)$ is a model of $\constrone$ and $\sisubstone_r$
  satisfies $\Var{\sisubstone_r(\sivarone)} \subseteq \skolemvars{\sivarone}$
  for each second-order variable $\sivarone$ of $\constrone$.
\end{lemma}
\shortv{
  \begin{proof}
    The restricted substitution $\sisubstone_r$ is obtained from $\sisubstone$
    by substituting in $\sisubstone(\sivarone)$ zero for all non-skolem variables $\ivarone \not\in\skolemvars{\sivarone}$. 
    From the assumption that $(\iinter,\sisubstone)$ is a model of $\constrone$, 
    it can then be shown that $\sitermone\sisubstone_r \leqs \sitermtwo\sisubstone_r$
    holds for each inequality $(\sitermone \leqc \sitermtwo) \in \constrone$, 
    essentially using the inequalities depicted in Lemma~\ref{l:leqs}.
    As the occurrence constraints are also satisfied  under the new model by definition, the lemma follows.
\end{proof}
}\longv{
\begin{proof}\cmt{check proof, because statement slightly changed}
  For a set of index variables $V \subseteq \IVARS$, let us denote by 
  $\subst{V}{\itermone}$ the index substitution with domain $V$ that maps every element from $V$ to $\itermone$, 
  and let $\complmnt{V} = \IVARS \setminus V$ denote the complement of $V$. 
  We define the restricted second-order substitution $\sisubstone_r$ by replacing 
  all non-skolem variables of $\sivarone$ in the solution $\sisubstone(\sivarone)$ by zero, i.e.
  \[
    \sisubstone_r(\sivarone) \defsym \sisubstone(\sivarone) \subst{\complmnt{\skolemvars{\sivarone}}}{0}
    \text{ for all $\sivarone \in \SOVar{\constrone}$.}
  \]
  Then by definition, $\Var{\sisubstone_r(\sivarone)} \subseteq \skolemvars{\sivarone} \cap \Var{\sisubstone(\sivarone)}$ holds for each second-order variable $\sivarone$. 
  This inclusion, together with the assumption $\sosat{\iinter}{\sisubstone}{\constrone}$, yields 
  that $\sisubstone_r$ satisfies still all occurrence constraints of $\constrone$. 
  To conclude $\sosat{\iinter}{\sisubstone_r}{\constrone}$, 
  it thus remains to show that $\sitermone\sisubstone_r \leqs \sitermtwo\sisubstone_r$
  holds for each inequality $(\sitermone \leqc \sitermtwo) \in \constrone$.
  
  Define $V \defsym \bigcap_{\sivartwo \in \SOVar{\sitermtwo}} \skolemvars{\sivartwo}$,
  and observe that $\Var{\sitermone} \subseteq V$
  and $\skolemvars{\sivarone} \subseteq V$ holds for each $\sivarone \in \SOVar{\sitermone}$, 
  by definition of $\skolemvars{\sivarone}$. 
  Consequently 
  \begin{equation}
    \label{eq:so2focs:var}
    \Var{\sitermone\sisubstone_r} 
    = \Var{\sitermone} \cup \bigcup_{\mathclap{\sivarone \in \SOVar{\sitermone}}} \Var{\sisubstone_r(\sivarone)} 
    \subseteq \Var{\sitermone} \cup \bigcup_{\mathclap{\sivarone \in \SOVar{\sitermone}}} \skolemvars{\sivarone}
    \subseteq V
    \tpkt
  \end{equation}
  Notice that 
  \begin{equation}
    \label{eq:so2focs:a}
    \sitermone\sisubstone_r \leqs \sitermone\sisubstone \leqs \sitermtwo\sisubstone \tkom
  \end{equation}
  holds, by Lemma~\eref{l:leqs}{zeroleft} and assumption. 
  Furthermore, we have
  \begin{equation}
    \label{eq:so2focs:subst}
    \sitermthree\subst{\complmnt{V}}{0} \leqs \sitermthree\subst{\complmnt{\skolemvars{\sivartwo}}}{0} 
    \text{ for all second-order index terms $\sitermthree$ and $\sivartwo \in \SOVar{\sitermtwo}$,}
  \end{equation}
  since $\complmnt{\skolemvars{\sivartwo}} \subseteq \complmnt{V}$ by definition, 
  using again Lemma~\eref{l:leqs}{zeroleft}.
  Putting things together, we conclude
  \begin{align*}
    \sitermone\sisubstone_r 
    & \leqs (\sitermtwo\sisubstone)\subst{\complmnt{V}}{0} 
    && \text{(by Lemma~\eref{l:leqs}{zeroright} with~\eqref{eq:so2focs:a} and~\eqref{eq:so2focs:var})}\\
    & \leqs \sitermtwo(\subst{\complmnt{V}}{0} \compose \sisubstone)
    && \text{(combining Lemma~\eref{l:leqs}{zeroleft} and Lemma~\eref{l:leqs}{subst})}\\
    & \leqs \sitermtwo\sisubstone_r 
    && \text{(by Lemma~\eref{l:leqs}{mon}, using~\eqref{eq:so2focs:subst}). }
  \end{align*}
\end{proof}}
\begin{theorem}[Skolemisation --- Soundness and Completeness]\label{t:so2focs}
\envskipline
  \begin{varenumerate}
    \item \textbf{Soundness}: 
      If $\fosat{\iinter}{\skolemize(\constrone)}$ 
      then $\sosat{\iinter}{\skolemsubst}{\constrone}$ holds.
    \item \textbf{Completeness}: 
      If $\sosat{\iinter}{\sisubstone}{\constrone}$ 
      then $\fosat{\skolemiinter}{\skolemize(\constrone)}$ 
      holds for an extension $\skolemiinter$ of $\iinter$ to skolem functions.
  \end{varenumerate}
\end{theorem}
\begin{proof}
  It suffices to consider completeness.
\shortv{%
  Suppose $\sosat{\iinter}{\sisubstone}{\constrone}$ holds, 
  where wlog. $\sisubstone$ satisfies 
  $\Var{\sisubstone(\sivarone)} \subseteq \skolemvars{\sivarone}$ 
  for each second-order variable $\sivarone \in \SOVar{\constrone}$ by Lemma~\ref{l:so2focs}.
  Let us extend the interpretation $\iinter$ to an interpretation $\skolemiinter$ by
  defining
  $\interpretation[\skolemiinter]{\skolemfun{\sivarone}}(\seq[k]{\ivarone}) 
  \defsym \interpret[\iinter][]{\sisubstone(\sivarone)}$,
  where $\skolemvars{\sivarone} = \sseq[k]{\ivarone}$, for all $\sivarone \in \SOVar{\constrone}$.
  By the assumption on $\sisubstone$, $\skolemiinter$ is well-defined.
  From the definition of $\skolemiinter$, it is then not difficult to conclude that also $(\skolemiinter,\skolemsubst)$ 
  is a model of $\constrone$, and consequently, $\iinter$ is a model of $\skolemize(\constrone)$.
}\longv{%
  Suppose $\sosat{\iinter}{\sisubstone}{\constrone}$ holds.
  Then Lemma~\ref{l:so2focs} yields a second-order substitution $\sisubstone_r$ 
  with $\sosat{\iinter}{\sisubstone_r}{\constrone}$ satisfying 
  \begin{equation}
    \label{eq:tso2focs:vars}
    \Var{\sisubstone_r(\sivarone)} 
    \subseteq \skolemvars{\sivarone}
    \tkom
  \end{equation}
  for each second-order variable $\sivarone \in \SOVar{\constrone}$. 
  We define the interpretation $\skolemiinter$ as the least extension of $\iinter$ such that
  \[
    \interpretation[\skolemiinter]{\skolemfun{\sivarone}}(\seq[k]{\ivarone}) 
    \defsym \interpret[\iinter][]{\sisubstone_r(\sivarone)} \tkom
    \text{ where $\skolemvars{\sivarone} = \sseq[k]{\ivarone}$.}
  \]
  Note that $\skolemiinter$ is well-defined by~\eqref{eq:tso2focs:vars}.
  Then for all $\sivarone \in \SOVar{\constrone}$ and assignments $\assignone$, 
  by definition we have
  \begin{equation}
    \label{eq:tso2focs:subst}
    \interpret[\skolemiinter][\assignone]{\skolemsubst(\sivarone)} 
    = \interpret[\iinter][\assignone]{\skolemfun{\sivarone}(\seq[k]{\ivarone})} 
    = \interpretation[\skolemiinter]{\skolemfun{\sivarone}}(\map{\assignone}[k]{\ivarone})
    = \interpret[\iinter][\assignone]{\sisubstone_r(\sivarone)} 
    \tkom
  \end{equation}
  for $\skolemvars{\sivarone} = \sseq[k]{\ivarone}$.
  To conclude, observe that for every constraint
  $(\itermone\skolemsubst \leqc \itermtwo\skolemsubst) \in \skolemize(\constrone)$, where $(\itermone \leqc \itermtwo) \in \constrone$, we have
  \begin{align*}
    \interpret[\skolemiinter]{\itermone\skolemsubst}
    & = \interpret[\iinter]{\itermone\sisubstone_r}
    && \text{(consequence of~\eqref{eq:tso2focs:subst})} \\
    & \leq \interpret[\skolemiinter]{\itermtwo\sisubstone_r}  
    && \text{(as $\sosat{\iinter}{\sisubstone_r}{\constrone}$)} \\
    & = \interpret[\skolemiinter]{\itermtwo\skolemsubst} 
    && \text{(consequence of~\eqref{eq:tso2focs:subst}),}
  \end{align*}
  i.e. $\itermone\skolemsubst \leqs[\skolemiinter] \itermtwo\skolemsubst$ 
  holds. We conclude the theorem.
}
\end{proof}
\subsection{Constraint Generation}\label{ssec:CGEN}
\newcommand{\inferstbase}{\rl{\ensuremath{\subtypeof[{\stbase}]}-I}}
\newcommand{\inferstpair}{\rl{\ensuremath{\subtypeof[\times]}-I}}
\newcommand{\inferstarr}{\rl{\ensuremath{\subtypeof[\starr]}-I}}
\newcommand{\inferstforall}{\rl{\ensuremath{\subtypeof[\forall]}-I}}

\newcommand{\infervar}{\rl{Var-I}}
\newcommand{\inferfun}{\rl{Fun-I}}
\newcommand{\inferlet}{\rl{Let-I}}
\newcommand{\inferapp}{\rl{App-I}}
\newcommand{\inferpair}{\rl{Pair-I}}

\begin{figure}
\begin{subfigure}[b]{\textwidth}
  \centering
  \begin{framed}
      \[
        \begin{array}[b]{c@{\quad}c}
          \Infer[\inferstbase]
          {\subtypedi{\set{\itermone \leqc \itermtwo}}{\base[\itermone]}{\base[\itermtwo]}}
          {} 
          & \Infer[\inferstpair]
            {\subtypedi{\constrone_1,\constrone_2}{\tpair{\mtpone_1}{\mtpone_2}}{\tpair{\mtpone_3}{\mtpone_4}}}
            { \subtypedi{\constrone_1}{\mtpone_1}{\mtpone_3} & \subtypedi{\constrone_2}{\mtpone_2}{\mtpone_4}}
          \\[1mm]
          \Infer[\inferstarr]
          {\subtypedi{\constrone_1,\constrone_2}{\tpone_1 \sarr \mtpone_1}{\tpone_2 \sarr \mtpone_2}}
          {\subtypedi{\constrone_1}{\tpone_2}{\tpone_1} & \subtypedi{\constrone_2}{\mtpone_1}{\mtpone_2}}
          & \Infer[\inferstforall]
            {\subtypedi{\constrone,\noccur{\vec{\ivarone}}{\mtpone_1},\noccur{\vec{\ivarone}}{\mtpone_2}}{\forall\vec{\ivarone}.\mtpone_1}{\forall\vec{\ivartwo}. \mtpone_2}} 
            {\fresh{\vec{\sivarone}}
            & \subtypedi{\constrone}{\mtpone_1}{\mtpone_2\substvec{\ivartwo}{\sivarone}}
            & \vec{\ivarone} \not \in \FV(\forall\vec{\ivartwo}. \mtpone_2)} 
        \end{array}
      \]
  \end{framed}
  \caption{Subtyping rules.}\label{fig:typeinfer:subtype}
\end{subfigure}
\\[5mm]
\begin{subfigure}[b]{\textwidth}
  \centering
  \begin{framed}
    \[
      \begin{array}[b]{c@{\quad}c}
        \Infer[\infervar]
        {\typedi{\emptyset}{\ctxone,\varone \oftype \forall \vec{\ivarone}. \mtpone}{\varone}{\mtpone\substvec{\ivarone}{\sivarone}}}
        {\fresh{\vec{\sivarone}}}
        & \Infer[\inferfun]
          {\typedi{\emptyset}{\ctxone}{s}{\mtpone\substvec{\ivarone}{\sivarone}}}
          {x \in \FUN \cup \CON & s \decl \forall \vec{\ivarone}. \mtpone & \fresh{\vec{\sivarone}}}
        \\ 
        \Infer[\inferlet] 
        {\typedi{\constrone_1,\constrone_2}{\ctxone}{\letexp{\termone}{\varone_1}{\varone_2}{\termtwo}}{\mtpone}}
        { \typedi{\constrone_1}{\ctxone}{\termone}{\tpair{\mtpone_1}{\mtpone_2}}
        & \typedi{\constrone_2}{\ctxone,\varone_1 \oftype \mtpone_1,\varone_2 \oftype \mtpone_2}{\termtwo}{\mtpone}}
        & \Infer[\inferpair]
          {\typedi{\constrone_1,\constrone_2}{\ctxone}{\pair{\termone_1}{\termone_2}}{\tpair{\mtpone_1}{\mtpone_2}}}
          {\typedi{\constrone_1}{\ctxone}{\termone_1}{\mtpone_1} & \typedi{\constrone_2}{\ctxone}{\termone_2}{\mtpone_2}}
      \end{array}
    \]
    \[
      \Infer[\inferapp] 
      {\typedi{\constrone_1,\constrone_2,\constrone_3,\noccur{\vec{\ivarone}}{\mtptwo_1},\noccur{\vec{\ivarone}}{\restrictctx{\ctxone}{\FV(\termtwo)}}}{\ctxone}{\termone \api \termtwo}{\mtpone}}
      {\typedi{\constrone_1}{\ctxone}{\termone}{(\forall \vec{\ivarone}. \mtptwo_1) \sarr \mtpone}
        & \typedi{\constrone_2}{\ctxone}{\termtwo}{\mtptwo_2}
        & \subtypedi{\constrone_3}{\mtptwo_2}{\mtptwo_1}
        & \vec{\ivarone} \not\in \FV(\restrictctx{\ctxone}{\FV(\termtwo)})
      }
    \]
  \end{framed}
  \caption{Typing rules}\label{fig:typeinfer:type}
\end{subfigure}
\\[5mm]
\caption{Type inference rules, generating a second-order constraint solving problem.}\label{fig:typeinfer}
\end{figure}

We now define a function $\obligations$ that maps a program $\progone$ to a SOCP $\constrone$.  
If $(\iinter,\sisubstone)$ is a model of $\constrone$, then $\progone$
will be well-typed under the interpretation $\iinter$.  Throughout the
following, we allow second-order index terms to occur in sized types.
If a second-order variable occurs in a type $\tpone$, we
call $\tpone$ a \emph{template type}.
The function $\obligations$ is itself defined on the
two statements $\subtypedi{\constrone}{\mtpone}{\mtptwo}$ and
$\typedi{\constrone}{\ctxone}{\termone}{\mtpone}$ that are used in
the generation of constraints resulting from the subtyping and the
typing relation, respectively.  
The inference rules are depicted in Figure~\ref{fig:typeinfer}.
These are in one-to-one correspondence with those of Figure~\ref{fig:typecheck}.
The crucial difference is that rule \inferstbase\ simply records a constraint
$\itermone \leqc \itermtwo$, whereas the corresponding rule \checkstbase\ in \Cref{fig:typecheck:subtype}
relies on the semantic comparison $\itermone \leqs \itermtwo$.
Instantiation of polytypes is resolved by substituting second-order variables, in rule \infervar\ and \inferfun. 
For a sequence of index variables $\vec{\ivarone} = \seq[m]{\ivarone}$ and sequence of monotype $\vec{\mtpone} = \seq[n]{\mtpone}$,
we use the notation $\noccur{\vec{\ivarone}}{\vec{\mtpone}}$ to denote 
the collection of occurrence constraints $\noccur{\ivarone_k}{\sivarone}$ 
for all $1 \leq k \leq m$ and $\sivarone \in \SOVar{\mtpone_l}$, $1 \leq l \leq n$. 
Occurrence constraints are employed in rules \inferstforall\ and \inferapp\
to guarantee freshness of the quantified index variables also with respect
to solutions to second-order index variables. 

Notice that the involved rules are again syntax directed. 
Consequently, a derivation of
$\typedi{\constrone}{\ctxone}{\termone}{\mtpone}$ naturally gives
rise to a procedure that, given a context $\ctxone$ and term
$\termone$, yields the SOCP $\constrone$ and template monotype $\mtpone$,
modulo renaming of second-order variables.  By imposing an order on
how second-order variables are picked in the inference of
$\typedi{\constrone}{\ctxone}{\termone}{\mtpone}$, the resulting SOCP
and template type become unique. The function
$\sinfer(\ctxone,\termone) \defsym (\constrone,\mtpone)$ defined this
way is thus well-defined.
In a similar way, we define the function $\subtypeOf(\mtpone,\mtptwo) \defsym \constrone$, 
where $\constrone$ is the SOCP with $\subtypedi{\constrone}{\mtpone}{\mtptwo}$.
\begin{definition}[Constraint Generation]
  For a program $\progone$ we define 
  \[
    \obligations(\progone) = \{ \scheck(\ctxone,r,\mtpone) \mid l = r \in \progone \text{ and } \footprint(l) = (\ctxone,\mtpone) \} 
    \tkom
  \]  
  where $\scheck(\ctxone,\termone,\mtpone) = \constrone_1 \cup \constrone_2$ for 
  $(\constrone_1,\mtptwo) = \sinfer(\ctxone,\termone)$ and $\constrone_2 = \subtypeOf(\mtptwo,\mtpone)$.
\end{definition}
\subsection{Soundness and Relative Completeness}
We will now give a series of soundness and completeness
results that will lead us to the main result about type inference,
namely Corollary~\ref{cor:soundcomplete} below.
In essence, we show that a derivation of 
$\subtypedi{\constrone}{\mtpone}{\mtptwo}$ (and $\typedi{\constrone}{\ctxone}{\termone}{\mtpone}$)
together with a model $\sosat{\iinter}{\sisubstone}{\constrone}$
can be turned into a derivation of $\mtpone\sisubstone \subtypeof[\iinter] \mtptwo\sisubstone$ (and $\typedsd[\iinter]{\ctxone\sisubstone}{\termone}{\mtpone\sisubstone}$), 
and \emph{vice versa}.
\begin{lemma}\label{l:infer:subtype}
  Subtyping inference is sound and complete, more precise:
  \begin{varenumerate}
    \item\label{l:infer:subtype:sound} \textbf{Soundness}: 
      If $\subtypedi{\constrone}{\mtpone}{\mtptwo}$ holds 
      for two template types $\mtpone, \mtptwo$ then
      $\mtpone\sisubstone \subtypeof[\iinter] \mtptwo\sisubstone$ holds
      for every model $(\iinter,\sisubstone)$ of $\constrone$.
    \item\label{l:infer:subtype:complete} \textbf{Completeness}:
      If $\mtpone\sisubstone \subtypeof \mtptwo\sisubstone$ holds
      for two template types $\mtpone$ and $\mtptwo$ and 
      second-order index substitution $\sisubstone$ then
      $\subtypedi{\constrone}{\mtpone}{\mtptwo}$ is derivable for some SOCP $\constrone$.
      Moreover, there exists an extension $\sisubsttwo$ of $\sisubstone$, whose 
      domain coincides with the second-order variables occurring in $\subtypedi{\constrone}{\mtpone}{\mtptwo}$,
      such that $(\iinter,\sisubsttwo)$ is a model of $\constrone$.
  \end{varenumerate}
\end{lemma}
\begin{proof}
\shortv{ 
  Concerning soundness, we consider a derivation of
  $\subtypedi{\constrone}{\mtpone_1}{\mtpone_2}$, and fix a second-order substitution $\sisubstone$ and interpretation $\iinter$ such that 
  $\sosat{\iinter}{\sisubstone}{\constrone}$ holds. 
  Then $\mtpone_1\sisubstone \subtypeof \mtpone_2\sisubstone$ can be proven by induction on the derivation of 
  $\subtypedi{\constrone}{\mtpone_1}{\mtpone_2}$. 
  Concerning completeness we fix a second-order substitution $\sisubstone$
  and construct for any two types $\mtpone_1,\mtpone_2$ with 
  $\mtpone_1\sisubstone \subtypeof \mtpone_2\sisubstone$
  an inference of 
  $\subtypedi{\constrone}{\mtpone_1}{\mtpone_2}$ for some SOCP $\constrone$ 
  together with an extension $\sisubsttwo$ of $\sisubstone$ 
  that satisfies $\sosat{\iinter}{\sisubsttwo}{\constrone}$. 
  The construction is done by induction on the proof of 
  $\mtpone_1\sisubstone \subtypeof \mtpone_2\sisubstone$.
  The substitution $\sisubsttwo$ extends $\sisubstone$ precisely on those fresh
  variables introduced by rule $\inferstforall$ in the constructed 
  proof of $\subtypedi{\constrone}{\mtpone_1}{\mtpone_2}$.
}\longv{
  We prove soundness first. To this end, consider a derivation of
  $\subtypedi{\constrone}{\mtpone_1}{\mtpone_2}$, and fix a second-order substitution $\sisubstone$ and interpretation $\iinter$ such that 
  $\sosat{\iinter}{\sisubstone}{\constrone}$ holds. 
  We prove $\mtpone_1\sisubstone \subtypeof \mtpone_2\sisubstone$ by induction on the derivation of 
  $\subtypedi{\constrone}{\mtpone_1}{\mtpone_2}$. 
  In the base case, we consider the derivation 
  \[
    \Infer[\inferstbase]
    {\subtypedi{\set{\itermone \leqc \itermtwo}}{\base[\itermone]}{\base[\itermtwo]}}
    {} \tpkt
  \]
  As by assumption $\sosat{\iinter}{\sisubstone}{\itermone \leqc \itermtwo}$, i.e., $\itermone\sisubstone \leqs \itermtwo\sisubstone$ holds, 
  we conclude $\base[\itermone]\sisubstone = \base[\itermone\sisubstone] \subtypeof \base[\itermtwo\sisubstone] = \base[\itermtwo]\sisubstone$
  by one application of rule~\checkstbase.
  
  In the first inductive case, we consider 
  \[
    \Infer[\inferstpair]
    {\subtypedi{\constrone_1 , \constrone_2}{\tpair{\mtpone_1}{\mtpone_2}}{\tpair{\mtpone_3}{\mtpone_4}}}
    {
      \subtypedi{\constrone_1}{\mtpone_1}{\mtpone_3} 
      & \subtypedi{\constrone_2}{\mtpone_2}{\mtpone_4}
    } \tkom
  \]
  where by assumption $\sisubstone$ and $\iinter$ satisfy $\sosat{\iinter}{\sisubstone}{\constrone_i}$ ($i = 0,1$). 
  The IHs thus state $\mtpone_1\sisubstone \subtypeof \mtpone_3\sisubstone$ and 
  $\mtpone_2\sisubstone \subtypeof \mtpone_4\sisubstone$, from which we conclude 
  \[
    (\tpair{\mtpone_1}{\mtpone_2})\sisubstone = \tpair{\mtpone_1\sisubstone}{\mtpone_2\sisubstone} \subtypeof \tpair{\mtpone_3\sisubstone}{\mtpone_4\sisubstone} = (\tpair{\mtpone_3}{\mtpone_4})\sisubstone
    \tkom
  \]
  by rule~\checkstpair.
  Similar, the case $\subtypedi{\constrone}{\tpone_1 \sarr \mtpone_1}{\tpone_2 \sarr \mtpone_2}$ 
  is handled. 

  In the final case, we consider 
  \[
    \Infer[\inferstforall]
    {\subtypedi{\constrone,\noccur{\vec{\ivarone}}{\mtpone_1},\noccur{\vec{\ivarone}}{\mtpone_2}}{\forall\vec{\ivarone}.\mtpone_1}{\forall\vec{\ivartwo}. \mtpone_2}} 
    {\fresh{\vec{\sivarone}}
      & \subtypedi{\constrone}{\mtpone_1}{\mtpone_2\substvec{\ivartwo}{\sivarone}}
      & \vec{\ivarone} \not \in \FV(\forall\vec{\ivartwo}. \mtpone_2)} 
  \]
  where (i)~$\sosat{\iinter}{\sisubstone}{\constrone}$ and furthermore 
  (ii)~$\vec{\ivarone} \not\in \FV(\sisubstone(\sivartwo))$ for each second-order variable $\sivartwo \in \SOVar{\mtpone_1} \cup \SOVar{\mtpone_2}$. 
  Wlog.\ we assume that the index variables $\vec{\ivartwo}$ are 
  renamed apart from the free variables occurring in images of $\sisubstone$. 
  Thus, 
  \[
    \forall \vec{\ivartwo}. \mtpone_2\sisubstone 
    \instantiates 
    (\mtpone_2\sisubstone)\subst{\vec{\ivartwo}}{\sisubstone(\vec{\sivarone})}
    = (\mtpone_2\substvec{\ivartwo}{\sivarone}) \sisubstone
    \tpkt
  \]
  Note that (ii)~together with the pre-condition $\vec{\ivarone} \not \in \FV(\forall\vec{\ivartwo}. \mtpone_2)$
  implies 
  $\vec{\ivarone} \not\in \FV(\forall\vec{\ivartwo}. \mtpone_2\sisubstone)$. 
  Using~(i), 
  $\mtpone_1\sisubstone \subtypeof (\mtpone_2\substvec{\ivartwo}{\sivarone})\sisubstone$ holds by IH.
  As we have shown already that the right-hand side is an instance of $\forall \vec{\ivartwo}. \mtpone_2\sisubstone$,
  we obtain
  \[
    (\forall\vec{\ivarone}. \mtpone_1)\sisubstone 
    = \forall\vec{\ivarone}. \mtpone_1\sisubstone
    \subtypeof \forall\vec{\ivartwo}. \mtpone_2\sisubstone 
    = (\forall\vec{\ivartwo}. \mtpone_2)\sisubstone
    \tkom
  \] 
  by rule~\checkstforall. Here, the first equality is a consequence of (ii), 
  and the second one follows as the index variables $\vec{\ivartwo}$ 
  do not occur in images of $\sisubstone$. 
  This concludes the final case of the soundness proof.

  We now consider completeness. 
  To this end, we fix a second-order substitution $\sisubstone$
  and construct for any two types $\mtpone_1$ and $\mtpone_2$ with 
  $\mtpone_1\sisubstone \subtypeof \mtpone_2\sisubstone$
  an inference of 
  $\subtypedi{\constrone}{\mtpone_1}{\mtpone_2}$ for some SOCP $\constrone$ 
  together with an extension $\sisubsttwo$ of $\sisubstone$ 
  that satisfies $\sosat{\iinter}{\sisubsttwo}{\constrone}$. 
  The construction is by induction on the proof of 
  $\mtpone_1\sisubstone \subtypeof \mtpone_2\sisubstone$.
  The substitution $\sisubsttwo$ extends $\sisubstone$ precisely on those fresh
  variables introduced by rule $\inferstforall$ in the constructed 
  proof of $\subtypedi{\constrone}{\mtpone_1}{\mtpone_2}$.

  In the first base case, we consider a proof
  \[
    \Infer[\checkstbase]
    {\base[\itermone]\sisubstone \subtypeof \base[\itermtwo]\sisubstone}
    {\itermone\sisubstone \leqs \itermtwo\sisubstone} 
  \]
  Then we have $\subtypedi{\set{\itermone \leqc \itermtwo}}{\base[\itermone]}{\base[\itermtwo]}$ by 
  rule $\inferstbase$ 
  and taking $\sisubsttwo \defsym \sisubstone$ proves the case, 
  as we have $\sosat{\iinter}{\sisubstone}{\itermone \leqc \itermtwo}$ by the assumption
  $\base[\itermone]\sisubstone \subtypeof \base[\itermtwo]\sisubstone$.

  In the first inductive case, we consider a proof 
  \[
    \Infer[\checkstpair]
    {\tpair{\mtpone_1\sisubstone}{\mtpone_2\sisubstone} 
      \subtypeof \tpair{\mtpone_3\sisubstone}{\mtpone_4\sisubstone}}
    {
      \infer*{\mtpone_1\sisubstone \subtypeof \mtpone_3\sisubstone}{D_1} 
      & \infer*{\mtpone_2\sisubstone \subtypeof \mtpone_4\sisubstone}{D_2}
    }
  \]
  The induction hypothesis on $D_1$ and $D_2$ yield inferences $E_1$ and $E_2$ 
  of $\subtypedi{\constrone_1}{\mtpone_1}{\mtpone_3}$ and 
  $\subtypedi{\constrone_2}{\mtpone_2}{\mtpone_4}$ for some 
  SOCPs $\constrone_1,\constrone_2$.
  Wlog.\ we suppose that the second-order variables introduced in $E_1$ and $E_2$ 
  by rule $\inferstforall$ are disjoint.
  We thus conclude 
  \[
      \Infer[\inferstpair]
      {\subtypedi{\constrone_1 \cup \constrone_2}{\tpair{\mtpone_1}{\mtpone_2}}{\tpair{\mtpone_3}{\mtpone_4}}}
      {
        \infer*{\subtypedi{\constrone_1}{\mtpone_1}{\mtpone_3}}{E_1}
        & \infer*{\subtypedi{\constrone_2}{\mtpone_2}{\mtpone_4}}{E_2}
      }
  \]
  Let $\sisubsttwo_i$ ($i \in \{0,1\}$) be two extensions of $\sisubstone$ that 
  satisfy 
  $\sosat{\iinter}{\sisubsttwo_i}{\constrone_i}$. 
  These two substitutions exist by IH, and satisfy 
  $\dom(\sisubsttwo_1) \cap \dom(\sisubsttwo_2) = \dom(\sisubstone)$ by assumption. 
  Define $\sisubsttwo \defsym \sisubsttwo_1 \uplus \sisubsttwo_2$.
  Clearly, 
  $\sosat{\iinter}{\sisubsttwo}{\constrone_i}$
  follows by definition from the assumption $\sosat{\iinter}{\sisubsttwo_i}{\constrone_i}$. We thus have
  $\sosat{\iinter}{\sisubsttwo}{\constrone_1 \cup \constrone_2}$ and conclude the case. 

  In exactly the same spirit we treat proofs ending in an application of rule $\checkstarr$.
  Thus, finally, consider a proof of 
  $(\forall\vec{\ivarone}.\mtpone_1)\sisubstone \subtypeof (\forall\vec{\ivartwo}.\mtpone_2)\sisubstone$,
  where $\vec{\ivarone},\vec{\ivartwo}$ are assumed to be renamed apart from variables 
  occurring in the domain and images of $\sisubstone$. 
  Also, we assume that $\vec{\ivartwo}$ are renamed apart from the free variables in $\mtpone_1$, 
  and likewise, that $\vec{\ivarone}$ are renamed apart from the free variables in $\mtpone_2$.
  Thus, 
  $(\forall\vec{\ivarone}.\mtpone_1)\sisubstone = \forall\vec{\ivarone}.\mtpone_1\sisubstone$
  and likewise 
  $(\forall\vec{\ivartwo}.\mtpone_2)\sisubstone = \forall\vec{\ivartwo}.\mtpone_2\sisubstone$.
  A proof of 
  $(\forall\vec{\ivarone}.\mtpone_1)\sisubstone \subtypeof (\forall\vec{\ivartwo}.\mtpone_2)\sisubstone$ 
  has then the form
  \[
      \Infer[\checkstforall]
      {\forall\vec{\ivarone}.\mtpone_1\sisubstone \subtypeof \forall\vec{\ivartwo}.\mtpone_2\sisubstone} 
      {
        \forall\vec{\ivartwo}.\mtpone_2\sisubstone \instantiates (\mtpone_2\sisubstone)\sisubstone_{\vec{\ivartwo}}
        & \infer*{\mtpone_1\sisubstone \subtypeof (\mtpone_2\sisubstone_{\vec{\ivartwo}})\sisubstone}{D}
        & \vec{\ivarone} \not\in\FV(\forall\vec{\ivartwo}.\mtpone_2\sisubstone)
      } 
  \]
  for some index substitution $\sisubstone_{\vec{\ivartwo}}$ defined precisely on $\vec{\ivartwo}$.
  For each $\ivartwo \in \vec{\ivartwo}$, 
  let $\sivarone_\ivartwo$ be a fresh second-order variable, 
  and let $\vec{\sivarone} \defsym (\sivarone_\ivartwo)_{\ivartwo \in \vec{\ivartwo}}$.
  We define $\sisubstone_{\vec{\sivarone}}$ as the extension of $\sisubstone$ by 
  $\sisubstone_{\vec{\sivarone}}(\sivarone_\ivartwo) \defsym \sisubstone_{\vec{\ivartwo}}(\ivartwo)$ for each 
  $\ivartwo$ in $\vec{\ivartwo}$. 
  Notice that using $\vec{\ivartwo} \not\in \FV(\mtpone_1)$, 
  we have
  $\mtpone_1\sisubstone = \mtpone_1\sisubstone_{\vec{\sivarone}}$ 
  and using $\vec{\ivartwo} \not\in \FV(\sisubstone(\ivarone))$ for each $\ivarone \in \dom(\sisubstone)$ we have
  $(\mtpone_2\sisubstone_{\vec{\ivartwo}})\sisubstone = \mtpone_2\substvec{\ivartwo}{\sivarone}\sisubstone_{\vec{\sivarone}}$. 
  Thus, the derivation $D$ proves 
  $\mtpone_1\sisubstone_{\vec{\sivarone}} \subtypeof \mtpone_2\sisubstone_{\vec{\sivarone}}$
  and IH yields an inference $E$ of 
  $\subtypedi{\constrone}{\mtpone_1}{\mtpone_2\substvec{\ivartwo}{\sivarone}}$, for some 
  SOCP $\constrone$. Also, IH yields an extension 
  $\sisubsttwo$ of $\sisubstone_{\vec{\sivarone}}$, and thus an extension of $\sisubstone$,
  so that $\sosat{\iinter}{\sisubsttwo}{\constrone}$ holds. 
  Using the inference $E$, we conclude 
  \[
    \Infer[\inferstforall]
    {\subtypedi{\constrone,\noccur{\vec{\ivarone}}{\mtpone_1},\noccur{\vec{\ivarone}}{\mtpone_2}}{\forall\vec{\ivarone}.\mtpone_1}{\forall\vec{\ivartwo}. \mtpone_2}} 
    {
      & \fresh{\vec{\sivarone}}
      & \infer*{\subtypedi{\constrone}{\mtpone_1}{\mtpone_2\substvec{\ivartwo}{\sivarone}}}{E}
      & \vec{\ivarone} \not \in \FV(\forall\vec{\ivartwo}. \mtpone_2)
    }
  \]
  Here, $\vec{\ivarone} \not \in \FV(\forall\vec{\ivartwo}. \mtpone_2)$ holds 
  by the assumption that the variables $\vec{\ivartwo}$ have been renamed apart from $\FV(\mtpone_2)$. 
  Finally, note that the second-order variables occurring in $\mtpone_1$ and $\mtpone_2$ 
  are all defined by $\sisubstone$, conclusively $\sisubsttwo(\sivarone) = \sisubstone(\sivarone)$ 
  for every $\sivarone \in \SOVar{\mtpone_1} \cup \SOVar{\mtpone_2}$. 
  As by assumption $\vec{\ivarone}$ does not occur in images of $\sisubstone$, 
  it is not difficult to see that the model $(\iinter,\sisubsttwo$) satisfy the constraints 
  $\noccur{\vec{\ivarone}}{\mtpone_1},\noccur{\vec{\ivarone}}{\mtpone_2}$.
  Since by IH directly $\sosat{\iinter}{\sisubsttwo}{\constrone}$ holds, we conclude this final case.
}
\end{proof}

\begin{lemma}\label{l:infer}
  Type inference is sound and complete in the following sense:
  \begin{varenumerate}
    \item\label{l:infer:sound} \textbf{Soundness}: 
      If $\typedi{\constrone}{\ctxone}{\termone}{\mtpone}$  holds 
      for a template type $\mtpone$ then
      $\typedsd[\iinter]{\ctxone\sisubstone}{\termone}{\mtpone\sisubstone}$
      holds for every model $(\iinter,\sisubstone)$ of $\constrone$.
    \item\label{l:infer:complete} \textbf{Completeness}:
      If $\typedsd[\iinter]{\ctxone}{\termone}{\mtpone}$ 
      holds for a context $\ctxone$ and type $\mtpone$
      then there exists a template type $\mtptwo$ and a 
      second-order index substitution $\sisubstone$, with $\mtptwo\sisubstone = \mtpone$, such that
      $\typedi{\constrone}{\ctxone}{\termone}{\mtptwo}$ is derivable for some SOCP $\constrone$.
      Moreover, $(\iinter,\sisubstone)$ is a model of $\constrone$.
  \end{varenumerate}
\end{lemma}
\begin{proof}
\shortv{
  Concerning soundness, we fix a model $(\iinter,\sisubstone)$ of $\constrone$ and prove then
  $\typedsd[\iinter]{\ctxone\sisubstone}{\termone}{\mtpone\sisubstone}$ by induction on $\typedi{\constrone}{\ctxone}{\termone}{\mtpone}$.
  Concerning completeness, we prove the following stronger statement. 
  Let $\sisubstone$ be a second-order index substitution, 
  let $\ctxone$ be a context over template schemas and let $\mtpone$ be a type.
  If $\typedsd[\iinter]{\ctxone\sisubstone}{\termone}{\mtpone}$ 
  is derivable then there exists an extension $\sisubsttwo$ of $\sisubstone$ together with a template type $\mtptwo$, 
  where $\mtptwo\sisubsttwo=\mtpone$,
  such that $\typedi{\constrone}{\ctxone}{\termone}{\mtptwo}$ holds for some SOCP $\constrone$. 
  Moreover, $(\iinter,\sisubsttwo)$ is a model of $\constrone$. 
  The proof of this statement is then carried out by induction on the derivation of $\typedsd{\ctxone\sisubstone}{\termone}{\mtpone}$.  
  Strengthening of the hypothesis is necessary to deal with let-expressions.
}\longv{
  We consider soundness first. 
  Fix a second-order substitution $\sisubstone$ and interpretation $\iinter$ such that 
  and $\sosat{\iinter}{\sisubstone}{\constrone}$ holds, 
  and suppose $\typedi{\constrone}{\ctxone}{\termone}{\mtpone}$. 
  We prove 
  $\typedsd{\ctxone\sisubstone}{\termone}{\mtpone\sisubstone}$ by induction on the 
  derivation $\typedi{\constrone}{\ctxone}{\termone}{\mtpone}$.

  In the first base case, we consider a derivation 
  \[
    \Infer[\infervar]
    {\typedi{\emptyset}{\ctxone,\varone \oftype \forall \vec{\ivarone}. \mtpone}{\varone}{\mtpone\substvec{\ivarone}{\sivarone}}}
    {\fresh{\vec{\sivarone}}} 
    \tkom
  \] 
  where, without loss of generality, the variables $\vec{\ivarone}$ do not occur in images of $\sisubstone$.
  Thus
  \[
    (\forall \vec{\ivartwo}. \mtpone)\sisubstone
    = \forall \vec{\ivartwo}. (\mtpone\sisubstone) 
    \instantiates 
    (\mtpone\sisubstone)\subst{\vec{\ivartwo}}{\sisubstone(\vec{\sivarone})}
    = (\mtpone\substvec{\ivartwo}{\sivarone}) \sisubstone \tpkt
  \]
  We conclude $\typedsd{\ctxone\isubstone,(\forall \vec{\ivartwo}. \mtpone)\sisubstone}{\varone}{(\mtpone\substvec{\ivartwo}{\sivarone})\sisubstone}$ by rule~\checkvarsd.
  Similar, we handle rule $\inferfun$.

  In the first inductive step, we consider a derivation
    \[
      \Infer[\inferapp] 
      {\typedi{\constrone_1,\constrone_2,\constrone_3,\noccur{\vec{\ivarone}}{\mtptwo_1},\noccur{\vec{\ivarone}}{\restrictctx{\ctxone}{\FV(\termtwo)}}}{\ctxone}{\termone \api \termtwo}{\mtpone}}
      {\infer*{\typedi{\constrone_1}{\ctxone}{\termone}{(\forall \vec{\ivarone}. \mtptwo_1) \sarr \mtpone}}{D_1}
        & \infer*{\typedi{\constrone_2}{\ctxone}{\termtwo}{\mtptwo_2}}{D_2}
        & \subtypedi{\constrone_3}{\mtptwo_2}{\mtptwo_1}
        & \vec{\ivarone} \not\in \FV(\restrictctx{\ctxone}{\FV(\termtwo)})
      }
    \]
  where (i)~$\sosat{\iinter}{\sisubstone}{\constrone_i}$ ($i = 1,2,3$) and furthermore 
  (ii)~$\vec{\ivarone} \not\in \FV(\sisubstone(\sivartwo))$ for each second-order variable $\sivartwo \in \SOVar{\mtptwo_1} \cup \SOVar{\restrictctx{\ctxone}{\FV(\termtwo)}}$. 
  By IH on $D_1$ and $D_2$ we obtain 
  $\typedsd{\ctxone\sisubstone}{\termtwo}{(\forall \vec{\ivarone}. \mtptwo_1\sisubstone) \sarr \mtpone\sisubstone}$
  and $\typedsd{\ctxone\sisubstone}{\termone}{\mtptwo_2\sisubstone}$, respectively.
  Likewise, using~(i) and the assumption $\subtypedi{\constrone_3}{\mtptwo_2}{\mtptwo_1}$ to satisfy the pre-conditions of Lemma~\eref{l:infer:subtype}{sound},
  we see $\subtypedi{\constrone_3}{\mtpone_2\sisubstone}{\mtpone_1\substone}$. 
  Together with (ii) we conclude thus $\typedsd{\ctxone\sisubstone}{\termone \api \termtwo}{\mtpone\sisubstone}$ 
  using rule~\checkappsd. 
  
  Finally, the cases where the derivation ends in an application of rule~\inferlet\ or rule~\inferpair\ follows directly from the IHs.
  We conclude the case of soundness. 
  
  We now consider completeness. 
  To handle let-expressions, we prove the following stronger statement:
  Let $\sisubstone$ be a second-order index substitution, 
  let $\ctxone$ be a context over template schemas and let $\mtpone$ be a type.
  If $\typedsd[\iinter]{\ctxone\sisubstone}{\termone}{\mtpone}$ 
  is derivable then there exists an extension $\sisubsttwo$ of $\sisubstone$ together with a template type $\mtptwo$, 
  where $\mtptwo\sisubsttwo=\mtpone$,
  such that $\typedi{\constrone}{\ctxone}{\termone}{\mtptwo}$ holds for some SOCP $\constrone$. 
  Moreover, $(\iinter,\sisubsttwo)$ is a model of $\constrone$. 
  The proof is by induction on the derivation of $\typedsd{\ctxone\sisubstone}{\termone}{\mtpone}$.  
  In the first base case, we consider a derivation of the form
  \[
    \Infer[\checkvarsd]
    {\typedsd{\ctxone\sisubstone,\varone \oftype (\forall\vec{\ivarone}. \mtpone_\varone)\sisubstone}{\varone}{\mtpone}}
    {(\forall\vec{\ivarone}. \mtpone_\varone)\sisubstone \instantiates \mtpone} 
    \tpkt
  \]
  Wlog.\ the variables 
  $\vec{\ivarone}$ are disjoint from those occurring in the domain and images of $\sisubstone$.
  Hence, for some index substitution $\isubstone_{\vec{\ivarone}}$ with domain $\vec{\ivarone}$, 
  we have 
  \[
    (\forall\vec{\ivarone}. \mtpone_\varone)\sisubstone
    = (\forall\vec{\ivarone}. \mtpone_\varone\sisubstone)
    \instantiates 
    (\mtpone_\varone\sisubstone)\isubstone_{\vec{\ivarone}}
    = \mtpone
    \tpkt
  \]
  For each $\ivarone \in \vec{\ivarone}$, 
  let $\sivarone_\ivarone$ be a fresh second-order variable, 
  and let $\vec{\sivarone} \defsym (\sivarone_\ivarone)_{\ivarone \in \vec{\ivarone}}$.
  We define $\sisubsttwo \defsym \sisubstone \uplus \substvec{\sivarone}{\isubstone_{\vec{\ivarone}}(\ivarone)}$. 
  Using that the index variables $\vec{\ivarone}$ do not occur in images of $\sisubstone$, we conclude
  $\mtpone
    = (\mtpone_\varone\sisubstone)\isubstone_{\vec{\ivarone}}
    = (\mtpone_\varone\substvec{\ivarone}{\sivarone})\sisubsttwo$.
  The case then follows by taking $\mtptwo = \mtpone_\varone\substvec{\ivarone}{\sivarone}$ and $\constrone = \emptyset$, 
  by rule $\infervar$.
  Trivially $(\iinter,\sisubsttwo)$ is a model of $\emptyset$.

  Next, consider the first inductive step where the considered typing derivation ends in an application of 
  rule $\checkappsd$, i.e.:
  \[
        \Infer[\checkappsd]
        {\typedsd{\ctxone}{\termone \api \termtwo}{\mtpone}}
        {
          \infer*{\typedsd{\ctxone}{\termone}{(\forall \vec{\ivarone}. \mtpone_1) \sarr \mtpone}}{D_1} 
          & \infer*{\typedsd{\ctxone}{\termtwo}{\mtpone_2}}{D_2}
          & \infer*{\mtpone_2 \subtypeof \mtpone_1}{D_3}
          & \vec{\ivarone} \not\in\FV(\restrictctx{\ctxone}{\FV(\termtwo)})}
  \]
  The IH on $D_1$ yields an inference $E_1$ of 
  $\typedi{\constrone_1}{\ctxone}{\termone}{(\forall \vec{\ivarone}. \mtptwo_1) \sarr \mtptwo}$
  together with an extension $\sisubsttwo_1$ of $\sisubstone$ such that 
  \[
    ((\forall \vec{\ivarone}. \mtptwo_1) \sarr \mtptwo)\sisubsttwo_1
    = (\forall \vec{\ivarone}. \mtptwo_1\sisubsttwo_1) \sarr \mtptwo\sisubsttwo_1
    = (\forall \vec{\ivarone}. \mtpone_1) \sarr \mtpone
    \tkom
  \]
  and hence, $\mtptwo_1\sisubsttwo_1 = \mtpone_1$ and $\mtptwo\sisubsttwo_1 = \mtpone$.
  Here, we suppose that the quantified index variables $\vec{\ivarone}$ have been 
  renamed apart from the variables occurring in images of $\sisubsttwo_1$. 
  Note that since $\sisubsttwo_1$ extends $\sisubstone$, we have $\ctxone\sisubstone = \ctxone\sisubsttwo_1$, 
  and thus the IH on 
  $D_2$ yields an inference $E_2$ of 
  $\typedi{\constrone_2}{\ctxone}{\termtwo}{\mtptwo_2}$
  together with an extension of $\sisubsttwo_2$ of $\sisubsttwo_1$ such that 
  $\mtpone_2 = \mtptwo_2\sisubsttwo_2$. As $\sisubsttwo_2$ extends $\sisubsttwo_1$, 
  we have $\mtpone_1 = \mtptwo_1\sisubsttwo_1 = \mtptwo_1\sisubsttwo_2$. 
  Thus, $\mtptwo_1\sisubsttwo_2 = \mtpone_1 \subtypeof \mtpone_2 = \mtptwo_2\sisubsttwo_2$ 
  and Lemma~\eref{l:infer:subtype}{complete} yields a derivation $E_3$
  $\subtypedi{\constrone_3}{\mtpone_1}{\mtpone_2}$ for some SOCP $\constrone_3$, 
  together with an extension $\sisubsttwo_3$ that satisfies 
  $\sosat{\iinter}{\sisubsttwo_3}{\constrone_3}$. 
  Note that the IHs on $D_1$ and $D_2$ give 
  $\sosat{\iinter}{\sisubsttwo_1}{\constrone_1}$, as well as 
  $\sosat{\iinter}{\sisubsttwo_2}{\constrone_2}$, respectively. 
  Since $\sisubsttwo_3$ coincides on second-order index variables occurring in $\constrone_i$
  with $\sisubsttwo_i$ ($i \in \{1,2 \}$) we conclude 
  $\sosat{\iinter}{\sisubsttwo_3}{\constrone_1,\constrone_2,\constrone_3}$. 
  Assuming the introduced second-order variables have been renamed apart in the derivations $E_1$, $E_2$ and $E_3$
  we conclude 
  \[
    \Infer[\inferapp] 
    {\typedi{\constrone_1,\constrone_2,\constrone_3,\noccur{\vec{\ivarone}}{\mtptwo_1},\noccur{\vec{\ivarone}}{\restrictctx{\ctxone}{\FV(\termtwo)}}}{\ctxone}{\termone \api \termtwo}{\mtpone}}
    {
      \infer*{\typedi{\constrone_1}{\ctxone}{\termone}{(\forall \vec{\ivarone}. \mtptwo_1) \sarr \mtptwo}}{E_1}
        & \infer*{\typedi{\constrone_2}{\ctxone}{\termtwo}{\mtptwo_2}}{E_2}
        & \infer*{\subtypedi{\constrone_3}{\mtptwo_2}{\mtptwo_1}}{E_3}
        & \vec{\ivarone} \not\in \FV(\restrictctx{\ctxone}{\FV(\termtwo)})
      }
    \]
    Here, $\vec{\ivarone} \not\in \FV(\restrictctx{\ctxone}{\FV(\termtwo)})$ follows from 
    the assumption $\vec{\ivarone} \not\in \FV(\restrictctx{\ctxone\sisubstone}{\FV(\termtwo)})$. 
    Observe that $\sisubsttwo_1$ is defined on all second-order variables in $\ctxone$ and $\mtptwo_1$.
    Recall that the index variables $\vec{\ivarone}$ do not occur in images of $\sisubsttwo_1$. 
    Since $\sisubsttwo_2$ extends $\sisubsttwo_1$, we thus have in particular
    $\vec{\ivarone} \not\in \FV(\sisubsttwo_2(\sivarone))$ for all 
    $\sivarone \in \SOVar{\mtptwo_1} \cup \SOVar{\restrictctx{\ctxone\sisubstone}{\FV(\termtwo)}}$. 
    Conclusively, $(\iinter,\sisubsttwo)$ satisfies the newly generated occurrence constraints
    $\noccur{\vec{\ivarone}}{\mtptwo_1},\noccur{\vec{\ivarone}}{\restrictctx{\ctxone}{\FV(\termtwo)}}$.
    We conclude this case. 

  In the next inductive step, we consider a derivation
  \[
    \Infer[\checkletsd]
    {\typedsd{\ctxone\sisubstone}{\letexp{\termone}{\varone_1}{\varone_2}{\termtwo}}{\mtpone}}
    {
      \infer*{\typedsd{\ctxone\sisubstone}{\termone}{\tpair{\mtpone_1}{\mtpone_2}}}{D_1}
      & \infer*{\typedsd{\ctxone\sisubstone, \varone_1 \oftype \mtpone_1, \varone_2 \oftype \mtpone_2}{\termtwo}{\mtpone}}{D_2}
    }
  \]
  The IH on $D_1$ yields a derivation $E_1$ of the statement 
  $\typedi{\constrone_1}{\ctxone_1}{\termone}{\tpair{\mtptwo_1}{\mtptwo_2}}$ 
  for some template types $\mtptwo_1$ and $\mtptwo_2$ and SOCP $\constrone_1$, 
  together with a second-order index substitution $\sisubsttwo_1$ that extends $\sisubstone$ and satisfies 
  \[
    (\tpair{\mtptwo_1}{\mtptwo_2})\sisubsttwo_1 = \tpair{\mtptwo_1\sisubsttwo_1}{\mtptwo_2\sisubsttwo_1} = \tpair{\mtpone_1}{\mtpone_2} \tpkt
  \] 
  Thus, $\mtpone_1=\mtptwo_1\sisubsttwo_1$ and $\mtpone_2=\mtptwo_2\sisubsttwo_1$. 
  As we also have $\ctxone\sisubstone = \ctxone\sisubsttwo_1$, the IH on $D_2$ yields 
  a derivation $E_2$ of 
  the statement $\typedi{\constrone_2}{\ctxone, \varone_1 \oftype \mtptwo_1, \varone_2 \oftype \mtptwo_2}{\termtwo}{\mtptwo}$ 
  for some template type $\mtptwo$ and SOCP $\constrone_2$,
  together with an extension $\sisubsttwo_2$ of $\sisubsttwo_1$ 
  satisfying $\mtptwo\sisubsttwo_2 = \mtpone$. Thus
  \[
      \Infer[\inferlet] 
      {\typedi{\constrone_1,\constrone_2}{\ctxone}{\letexp{\termone}{\varone_1}{\varone_2}{\termtwo}}{\mtptwo}}
      {
        \infer*{\typedi{\constrone_1}{\ctxone}{\termone}{\tpair{\mtptwo_1}{\mtptwo_2}}}{E_1}
        & \infer*{\typedi{\constrone_2}{\ctxone, \varone_1 \oftype \mtptwo_1, \varone_2 \oftype \mtptwo_2}{\termtwo}{\mtptwo}}{E_2}
      }
  \]
  and it is not difficult to obtain
  $\sosat{\iinter}{\sisubsttwo_2}{\constrone_1,\constrone_2}$ from the IHs on $D_1$ and $D_2$, respectively. 
  We conclude this case.
  Finally, the remaining case where a derivation ends in an application of $\checkpairsd$ follows directly from IH.\@
  We conclude completeness.
}
\end{proof}

\begin{theorem}[Inference --- Soundness and Relative Completeness]
  Let $\progone$ be a program and let $\constrone = \obligations(\progone)$.
  \begin{varenumerate}
    \item \textbf{Soundness}: 
      If $(\iinter,\sisubstone)$ is a model of $\constrone$, then 
      $\progone$ is well-typed under the interpretation $\iinter$.
    \item \textbf{Completeness}: 
      If $\progone$ is well-typed under the interpretation $\iinter$, 
      then there exists a second-order index substitution $\sisubstone$ 
      such that $(\iinter,\sisubstone)$ is a model of $\constrone$.
  \end{varenumerate}
\end{theorem}
\begin{proof}
  Concerning soundness, let $(\iinter,\sisubstone)$ be a model of $\constrone$.
  Fix a rule $l = r$ of $\progone$, and let $(\ctxone,\mtpone) = \footprint(l)$. 
  Notice that $(\iinter,\sisubstone)$ is in particular a model of the constraint
  $\constrone_1 \cup \constrone_2 = \scheck(\ctxone,r,\mtpone) \subseteq \constrone$, 
  where $\typedi{\constrone_1}{\ctxone}{r}{\mtptwo}$ and
  $\subtypedi{\constrone_2}{\mtptwo}{\mtpone}$ for some type $\mtptwo$. 
  Using that the footprint of $l$ does not contain second-order index variables, 
  Lemma~\eref{l:infer}{sound} and Lemma~\eref{l:infer:subtype}{sound} then prove 
  $\typedsd[\iinter]{\ctxone}{\termone}{\mtptwo\sisubstone}$ 
  and $\mtptwo\sisubstone \subtypeof[\iinter] \mtpone$, respectively.
  Conclusively, the rule $l = r$ is well-typed and the claim follows.
  Completeness is proven dually, using 
  Lemma~\eref{l:infer}{complete} and Lemma~\eref{l:infer:subtype}{complete}.
\end{proof}
This, in conjunction with Theorem~\ref{t:so2focs}, then yields:
\begin{corollary}\label{cor:soundcomplete}
  Let $\progone$ be a program and let $\constrone = \obligations(\progone)$.
  \begin{varenumerate}
    \item \textbf{Soundness}: 
      If $\iinter$ is a model of $\skolemize(\constrone)$, then 
      $\progone$ is well-typed under the interpretation $\iinter$.
    \item \textbf{Completeness}: 
      If $\progone$ is well-typed under the interpretation $\iinter$, 
      then $\skolemiinter$ is a model of $\skolemize(\constrone)$, 
      for some extension $\skolemiinter$ of $\iinter$.
  \end{varenumerate}
\end{corollary}


\section{Ticking Transformation and Time Complexity Analysis}\label{sect:TTTCA}
Our size type system is a sound methodology for keeping track of
the size of intermediate results a program needs when evaluated.
Knowing all this, however, is not sufficient for complexity analysis.
In a sense, we need to be able to reduce complexity analysis to
size analysis.

We now introduce the \emph{ticking transformation} mentioned in the
Introduction.  Conceptually, this transformation takes a program
$\progone$ and translates it into another program $\tprogone$ which
behaves like $\progone$, but additionally computes also the runtime on
the given input.
The latter is achieved by threading through the computation a counter,
the \emph{clock}, which is advanced whenever an equation of $\progone$
fires.  Technically, we lift all the involved functions into a
\emph{state monad},%
\footnote{We could have achieved a similar effect via a writer
  monad. We prefer however the more general notion of a state monad,
  as this allows us to in principle also encode resources that can be
  reclaimed, e.g., heap space.}  that carries as state the clock. More
precise, a $k$-ary function $\funone \decl \stone_1 \starr \cdots
\starr \stone_{k} \starr \stone$ of $\progone$ will be modeled in
$\tprogone$ by a function $\tfun{\funone}{k} \decl \ttype{\stone_1}
\starr \cdots \starr \ttype{\stone_{k}} \starr \tclock \starr
\tpair{\ttype{\stone}}{\tclock}$, where $\tclock$ is the type of the
\emph{clock}.  Here, $\ttype{\sttwo}$ enriches functional types
$\sttwo$ with clocks accordingly.  The function $\tfun{\funone}{k}$
behaves in essence like $\funone$, but advances the threaded clock
suitably.  The clock-type $\tclock$ encodes the running time in unary
notation using two constructors $\tickzero^{\tclock}$ and
$\ticksucc^{\tclock \starr \tclock}$. The size of the clock thus
corresponds to its value.  Overall, ticking effectively reduces time
complexity analysis to a size analysis of the threaded clock.

\begin{figure}[t]
  \centering
  \begin{minipage}{.40\linewidth}  
    \begin{minipage}[t][2.35cm]{\linewidth}
\begin{lstlisting}[emph={x},emph={[2] f,g,h},style=haskell, style=numbered,%
                   literate={z6}{{$\varclockone_6$}}2
                            {x0}{{$\varone_0$}}2
                            {x1}{{$\varone_1$}}2
                            {x2}{{$\varone_2$}}2
                            {x3}{{$\varone_3$}}2
                            {x4}{{$\varone_4$}}2
                            {x5}{{$\varone_5$}}2
                            {x6}{{$\varone_6$}}2]
f x = let x1 = g in %\phantom{$\tfun{f}{0}$}%
      let x2 = h in 
      let x3 = x2 x in
      let x4 = x1 x3 in x4

\end{lstlisting}      
    \end{minipage}
  \end{minipage}
  \hfill
  \begin{minipage}{.59\linewidth}
    \begin{minipage}[t][2.35cm]{\linewidth}
\begin{lstlisting}[emph={x,z},emph={[2] f,g,h},style=haskell, style=numbered,%
                   literate={f0}{{$\tfun{f}{0}$}}2
                            {f1}{{$\tfun{f}{1}$}}2
                            {g0}{{$\tfun{g}{0}$}}2
                            {h0}{{$\tfun{h}{0}$}}2
                            {T}{{$\ticksucc$}}1
                            {z0}{{$\varclockone_0$}}2
                            {z1}{{$\varclockone_1$}}2
                            {z2}{{$\varclockone_2$}}2
                            {z3}{{$\varclockone_3$}}2
                            {z4}{{$\varclockone_4$}}2
                            {x0}{{$\varone_0$}}2
                            {x1}{{$\varone_1$}}2
                            {x2}{{$\varone_2$}}2
                            {x3}{{$\varone_3$}}2
                            {x4}{{$\varone_4$}}2
                            {x5}{{$\varone_5$}}2
                            {x6}{{$\varone_6$}}2]
f1 x z = let (x1,z1) = g0 z in
         let (x2,z2) = h0 z1 in
         let (x3,z3) = x2 x z2 in
         let (x4,z4) = x1 x3 z3 in (x4,T z4)
f0 z = (f1, z)
\end{lstlisting}      
    \end{minipage}
  \end{minipage}
  \caption{Equation $\funone \api x = \funtwo \api (\funthree \api \varone)$ in let-normalform (left) and ticked let-normalform (right).}
  \label{fig:ticking}
\end{figure}

Ticking of a program can itself be understood as a two phase process. 
In the first phase, the body $r$ of each equation $\funone \api \aseq[k]{\patone} = r$
is transformed into a very specific let-normalform:
\begin{align*}
  & \text{\textbf{(let-normalform)}} 
  & e \bnfdef \varone \mid \LET{s}{\varone}{e} \mid \LET{\varone_2 \api \varone_3}{\varone_1}{e}
    \tkom
\end{align*}
for variables $\varone_i$ and $s \in \FUN \cup \CON$. 
This first step makes the evaluation order explicit, without changing program semantics. 
On this intermediate representation, it is then trivial to thread through a global counter.
Instrumenting the program this way happens in the second stage. 
Each $k$-ary function $\funone$ is extended with an additional clock-parameter, 
and this clock-parameter is passed through the right-hand side of each defining equation. 
The final clock value is then increased by one.
This results in the definition of the instrumented function $\tfun{\funone}{k}$.
Intermediate functions $\tfun{\funone}{i}$ ($0 \leq i < k$) deal with partial application. 
Compare Figure~\ref{fig:ticking} for an example. 

Throughout the following, we fix a \emph{pair-free program} $\progone$, i.e.
$\progone$ neither features pair constructors nor destructors. Pairs are indeed only added 
to our small programming language to conveniently facilitate ticking. The following 
definition introduces the ticking transformation formally.
Most important, $\texp{\termone^\stone}{\varclockone}{K}$ simultaneously
applies the two aforementioned stages to the term $\termone$. The variable $\varclockone$
presents the initial time. 
The transformation is defined in continuation passing style. Unlike a traditional definition, 
the continuation $K$ takes as input not only the result of evaluating $\termone$, but also 
the updated clock. It thus receives two arguments, viz two terms 
of type $\ttype{\stone}$ and $\tclock$, respectively. 

\begin{definition}[Ticking]
  Let $\progone$ be a program over constructors $\CON$ and functions $\FUN$.
  Let $\tclock \not\in\STBASE$ be a fresh base type, the \emph{clock type}.
  \begin{varenumerate}
  \item To each simple type $\stone$, we associate the following \emph{ticked type} $\ttype{\stone}$:
    \begin{align*}
      \ttype{\stbase} & \defsym \stbase 
      & \ttype{\tpair{\stone_1}{\stone_2}} & \defsym \tpair{\ttype{\stone_1}}{\ttype{\stone_2}} 
      & \ttype{\stone_1 \starr \stone_2} & \defsym \ttype{\stone_1} \starr \tclock \starr \tpair{\ttype{\stone_2}}{\tclock}
    \end{align*}
  \item The set $\TCON$ of \emph{ticked constructors} 
    contains a symbol $\tickzero^{\tclock}$, a symbol $\ticksucc^{\tclock \starr \tclock}$, the \emph{tick},
    and for each constructor $\conone^{\stone_1 \starr \cdots \starr \stone_k \starr \stbase}$ 
    a new constructor $\tconone^{\ttype{\stone_1} \starr \cdots \starr \ttype{\stone_k} \starr \stbase}$.
  \item The set $\TFUN$ of \emph{ticked functions}
    contains for each $s^{\stone_1 \starr \cdots \starr \stone_{i} \starr \stone} \in \FUN \cup \CON$ 
    and $0 \leq i \leq \arity{s}$ a new function 
    $\tfun{s}{i}^{\ttype{\stone_1} \starr \cdots \starr \ttype{\stone_{i}} \starr \tclock \starr \tpair{\ttype{\stone}}{\tclock}}$.
  \item For each variable $\varone^\stone$, we assume a dedicated variable $\tvarone^{\ttype{\stone}}$.
  \item We define a translation from (non-ground) values $\valone^\stone$ over $\CON$ to (non-ground) values $\tvalone^{\ttype{\stone}}$ over $\TCON$ as follows. 
    \[
      \tvalone \defsym
      \begin{cases}
        \tvarone & \text{if $\valone = \varone \in \VAR$,} \\
        \tfun{s}{k} \api \tvalone_1 \api \cdots \api \tvalone_k
        & \text{if $\valone = s \api \valone_1 \api \cdots \api \valone_k$, $s \in \FUN \cup \CON$ and $k < \arity{s}$}, \\
        \tconone \api \tvalone_1 \api \cdots \tvalone_{\arity{\conone}} & \text{if $\valone = \conone \api \valone_1 \api \cdots \api \valone_{\arity{\conone}}$.} 
      \end{cases}
    \]
  \item We define a translation from terms over $\FUN \cup \CON$ to terms
    in \emph{ticked let-normalform} over $\TFUN$ as follows. 
    Let $\tickk\ \varone\ \varclockone = \pair{\varone}{\ticksucc \api \varclockone}$. 
    For a term $\termone$ and variable $\varclockone^{\tclock}$ we define
    $\texp{\termone}{\varclockone}{} \defsym \texp{\termone}{\varclockone}{\tickk}$,
    where 
    \begin{align*}
      & \texp{\termone^\stone}{\varclockone_i}{K}  \defsym
      \begin{cases}
        K\ \tickvar{\termone}\ \varclockone_i 
        & \text{if $\termone$ is a variable,} \\
        \letexp{\tfun{\termone}{0}\ \varclockone_i}{\varone^{\ttype{\stone}}}{\varclockone_{i+1}^{\tclock}}{K\ \varone\ \varclockone_{i+1}}
        & \text{if $\termone \in \FUN \cup \CON$,}\\
        \texp{\termone_1^{\sttwo \starr \stone}}{\varclockone_i}{K_1} 
        & \text{if $\termone = \termone_1^{\sttwo \starr \stone} \api \termone_2^{\sttwo}$\tkom}
      \end{cases}
    \end{align*}
    where in the last clause, 
    $K_1\ \varone_1^{\ttype{\sttwo \starr \stone}}\ \varclockone_j^{\tclock} \defsym \texp{\termone_2^{\sttwo}}{\varclockone_j}{(K_2\ \varone_1)}$ 
    and 
    $K_2\ \varone_1^{\ttype{\sttwo \starr \stone}}\ \varone_2^{\ttype{\sttwo}}\ \varclockone_k^{\tclock} 
    \defsym \letexp{\varone_1 \api \varone_2 \api \varclockone_k}{\varone^{\ttype{\stone}}}{\varclockone_l^{\tclock}}{K\ \varone\ \varclockone_l}$.
    All variables introduced by let-expressions are supposed to be fresh.
  \item The \emph{ticked} program $\tprogone$ consists of the following equations:
    \begin{varenumerate}
    \item\label{d:teq} For each equation $\funone \api \aseq[\arity{\funone}]{\patone} = r$ in $\progone$, the \emph{translated equation} 
      \[
        \mparbox[r]{4cm}{\tfun{\funone}{\arity{\funone}} \api \aseq[\arity{\funone}]{\tpatone}\ \varclockone} 
        = \mparbox{4cm}{\texp{r}{\varclockone}{}\tkom}
      \]
    \item\label{d:taux} for all $s \in \FUN \cup \CON$ and $0 \leq i < \arity{s}$, an \emph{auxiliary equation}
      \[
        \mparbox[r]{4cm}{\tfun{s}{i}\api\aseq[i]{\varone}\api\varclockone} 
        = \mparbox{4cm}{\pair{\tfun{s}{i+1} \api \aseq[i]{\varone}}{\varclockone}\tkom} 
      \]
    \item\label{d:tauxc} for all $\conone \in \CON$, an \emph{auxiliary equation}
      \[
        \mparbox[r]{4cm}{\tfun{\conone}{\arity{\conone}}\api\aseq[\arity{\conone}]{\varone}\api\varclockone} 
        = \mparbox{4cm}{\pair{\tconone \api \aseq[\arity{\conone}]{\varone}}{\varclockone}\tpkt} 
      \]
    \end{varenumerate}
    If $\termone \rew[\tprogone] \termtwo$, then we also write $\termone \rewtick \termtwo$ and $\termone \rewaux \termtwo$
    if the step from $\termone$ to $\termtwo$ follows by a translated (case~\ref{d:teq}) and auxiliary equation (cases~\ref{d:taux} and~\ref{d:tauxc}), respectively.
  \end{varenumerate}
\end{definition}

Our main theorem from this section states that whenever $\tprogone$ is well-type under an interpretation $\iinter$, 
thus in particular $\tfun{\fun{main}}{k}$ receives a type 
$\forall \vec{\ivarone}\ivartwo. \base[\ivarone_1] \starr \cdots \starr \base[\ivarone_{k}] \starr \sizeannotate{\tclock}{\ivartwo} \starr \tpair{\base[\itermone]}{\sizeannotate{\tclock}{\itermtwo}}$,
then the running time of $\progone$ on inputs of size $\vec{\ivarone}$ is bounded by $\interpretation[\iinter]{\itermtwo\subst{\ivartwo}{0}}$.
This is proven in two steps. In the first step, we show a precise correspondence between reductions of $\progone$ and $\tprogone$. 
This correspondence in particular includes that the clock carried around by $\tprogone$ faithfully represents the execution time of $\progone$. 
In the second step, we then use the subject reduction theorem to conclude that the index $\itermtwo$ in turn estimates 
the size, and thus value, of the threaded clock. 

\subsection{The Ticking Simulation}

The ticked program $\tprogone$ operates on very specific terms, viz, terms in let-normal form enriched 
with clocks. The notion of \emph{ticked let-normalforms} over-approximates this set. 
This set of terms is generated from $s \in \FUN \cup \CON$ and $k < \arity{s}$ inductively as follows.
\begin{align*}
  & \text{\textbf{(clock terms)}} 
  & \clockone \bnfdef {} 
  & \varclockone^{\tclock} \mid \tickzero \mid \ticksucc \api \clockone
    \tkom\\
  & \text{\textbf{(ticked let-normalform)}} 
  & \tlone,\tltwo \bnfdef {} 
  & (\tvalone,\clockone) 
    \mid \tfun{s}{k} \api \tvalone_1 \api \cdots \api \tvalone_k \api \clockone
    \mid \letexp{\tlone}{\varone}{\varclockone}{\tltwo}
    \tpkt
\end{align*}
Not every term generated from this grammar is legitimate. 
In a term $\letexp{\tlone}{\varone}{\varclockone}{\tltwo}$, we require 
that the two let-bound variables $\varone,\varclockone$ occur exactly once free in $\tltwo$.
Moreover, the clock variable $\varclockone$ occurs precisely in the \emph{head} of $\tltwo$.
Here, the head of a term in ticked let-normalform is given recursively as follows. 
In $\letexp{\tlone}{\varone}{\varclockone}{\tltwo}$, the head position is the one of $\tlone$.
In the two other cases, the terms are itself in head position.
This ensures that the clock is suitably wired, compare Figure~\ref{fig:ticking}.
Throughout the following, we assume that every term in ticked let-normalform satisfies these criteria.
This is justified, as terms in ticked let-normalform are closed under $\tprogone$-reductions, 
a consequence of the particular shape of right-hand sides in $\tprogone$.

As a first step towards the simulation lemma, we define a translation $\tb{\tlone}$ of the term $\tlone$ 
in ticked let-normalform to a pair, viz, a terms of $\progone$ and a clock term. 
We write $\tbe{\tlone}$ and $\tbc{\tlone}$ 
for the first and second component of $\tb{\tlone}$, respectively. The translation is defined by 
recursion on $\tlone$ as follows.
\begin{align*}
  \tb{\tlone} \bnfdef 
  \begin{cases}
    (\valone,\clockone)
    & \text{if $\tlone = (\tvalone,\clockone)$,}\\
    (s \api \valone_1 \api \cdots \api \valone_k,\clockone)
    & \text{if $\tlone = \tfun{s}{k} \api \tvalone_1 \api \cdots \api \tvalone_k \api \clockone$ where $s \in \FUN \cup \CON$,}\\
    \tb{\tlone_2}\{\substby{\varone}{\tbe{\tlone_1}},\substby{\varclockone}{\tbc{\tlone_1}}\}
    & \text{if $\tlone = \letexp{\tlone_1}{\varone}{\varclockone}{\tlone_2}$.}
  \end{cases}
\end{align*}


\begin{lemma}\label{l:simul:aux}
  Let $\tlone$ be a term in ticked let-normalform. The following holds:
  \begin{varenumerate}
    \item $\tlone \rewtick \tltwo$ implies $\tbe{\tlone} \rew[\progone] \tbe{\tltwo}$ and $\tbc{\tltwo} = \ticksucc \api \tbc{\tlone}$; and
    \item $\tlone \rewaux \tltwo$ implies $\tbe{\tlone} = \tbe{\tltwo}$ and $\tbc{\tltwo} = \tbc{\tlone}$; and
    \item\label{l:simul:aux:nf} if $\tbe{\tlone}$ is reducible with respect to $\progone$, then $\tlone$ is reducible with respect to $\tprogone$.
  \end{varenumerate}
\end{lemma}
    

The first two points of Lemma~\ref{l:simul:aux} immediately yield that given a $\tprogone$-reduction, this 
reduction corresponds to a $\progone$-reduction. In particular, the lemma translates a reduction
\[
  \tfun{\fun{main}}{k}\api\aseq[k]{\tickval{\dataone}}\api\tickzero
  \rewtick \cdot \rssaux \tlone_1 
  \rewtick \cdot \rssaux \cdots
  \rewtick \cdot \rssaux \tlone_\ell\tkom
\]
to
\[
  \tbe{\tfun{\fun{main}}{k} \api \aseq[k]{\tickval{\dataone}}}
  = \fun{main}\api \aseq[k]{\dataone}
  \rew[\progone] \tbe{\tlone_1}
  \rew[\progone] \cdots
  \rew[\progone] \tbe{\tlone_\ell}
  \tkom
\]  
where moreover, $\tbc{\tlone_\ell} = \ticksucc^\ell\api \tickzero$. 
In the following, let us abbreviate ${\rewtick} \cdot {\rssaux}$ by $\rewtrel$.
This, however, is not enough to show that $\tprogone$ simulates 
$\progone$. It could very well be that $\tprogone$ gets stuck 
at $\tlone_\ell$, whereas the corresponding term $\tbe{\tlone_\ell}$ is reducible. 
Lemma~\eref{l:simul:aux}{nf} verifies that this is indeed not the case. 
Another, minor, complication that arises 
is that $\tprogone$ is indeed not able to simulate \emph{any} $\progone$-reduction. 
Ticking explicitly encodes a left-to-right reduction, $\tprogone$ can thus only 
simulate left-to-right, call-by-value reductions of $\progone$. However, 
Proposition~\ref{p:nonambiguous} clarifies that left-to-right is as good as any 
reduction order. To summarise:
\begin{theorem}[Simulation Theorem --- Soundness and Completeness]\label{t:simul}
  Let $\progone$ be a program whose $\fun{main}$ function is of arity $k$.
  \begin{varenumerate}
    \item \textbf{Soundness}: If $\tfun{\fun{main}}{k}\api\aseq[k]{\tickval{\dataone}}\api\tickzero \rewtrell{\ell} \tlone$ 
      then $\fun{main}\api\aseq[k]{\dataone} \rsl[\progone]{\ell} \api \termtwo$ where moreover, 
      $\tbe{\tlone} = \termtwo$ and $\tbc{\tlone} = \ticksucc^\ell \api \tickzero$.
    \item \textbf{Completeness}:
      If $\fun{main}\api\aseq[k]{\dataone} \rsl[\progone]{\ell} \termone$ then there exists an alternative reduction
      $\fun{main}\api\aseq[k]{\dataone} \rsl[\progone]{\ell} \termtwo$ such that
      $\tfun{\fun{main}}{k}\api\aseq[k]{\tickval{\dataone}}\api\tickzero \rewtrell{\ell} \tlone$ where moreover, 
      $\tbe{\tlone} = \termtwo$ and $\tbc{\tlone} = \ticksucc^\ell \api \tickzero$.
  \end{varenumerate}
\end{theorem}

\subsection{Time Complexity Analysis}

As corollary of the Simulation Theorem, essentially through Subject Reduction, we finally obtain our main result.
\begin{theorem}
  Suppose $\tprogone$ is well-typed under the interpretation $\iinter$, 
  where data-constructors, including the clock constructor $\ticksucc$, are given an additive type
  and where 
  $\tfun{\fun{main}}{k} \decl \forall \vec{\ivarone}\ivartwo. \base[\ivarone_1] \starr \cdots \starr \base[\ivarone_{k}] \starr \sizeannotate{\tclock}{\ivartwo} \starr \tpair{\base[\itermone]}{\sizeannotate{\tclock}{\itermtwo}}$.
  The runtime complexity of $\progone$ is bounded from above by 
  $rc(\seq[k]{\ivarone}) \defsym \interpretation[\iinter]{\itermtwo\subst{\ivartwo}{0}}$. 
\end{theorem}
In the proof of this theorem, we use actually a strengthening of Corollary~\ref{cor:size}.
When a term $\tlone$ in ticked let-normal form is given a type $\tpair{\base[\itermone]}{\sizeannotate{\tclock}{\itermtwo}}$, then $\itermtwo$ accounts for the size of $\tbc{\tlone}$, 
even if $\tlone$ is not in normal form.

\section{Prototype and Experimental Results}\label{sect:PER}
We have implemented the discussed inference machinery in a
prototype, dubbed \hosa.\footnote{Available from \url{http://cl-informatik.uibk.ac.at/~zini/software/hosa/}.}  
This tool performs a fully automatic sized type inference on the
typed language given in Section~\ref{sect:APST}, extended with
polymorphic types and inductive data type definitions as presented in examples earlier on.
These extension are already present in the canonical system from \citet{HPS:POPL:96}, 
and help not only towards modularity of the analysis, but enable also a more fine-grained
capture of sizes. 

Thus, in our implementation, the language of types is extended with type variables $\tyvarone$, that range over sized types,
and $n$-ary data type constructors $\tycon{D}$. Each such data constructor is associated with 
$m$ distinct constructors $\conone_i \decl \forall \tyvarone_1 \dots \tyvarone_n.\ \stone_1 \starr \cdots \stone_{k_i} \starr \tycon{D}\ \tyvarone_1 \dots \tyvarone_n$.
To accommodate these extensions to the type language, two main changes are necessary to our type system. 
First, the subtyping relation has to be adapted, to account for type variables 
and $n$-ary data type constructors, see \Cref{fig:typecheck:subtype'}. Notice that
in the second rule, the variance of arguments, given by the types of the corresponding constructors, are taken into account. 
Second, the type system has to be extended with the usual rule for instantiation of type variables. 
Also, some auxiliary definitions, noteworthy the one of canonicity, have to be suited in the obvious way.

\begin{figure}[t]
  \centering
  \begin{framed}
    \[
      \begin{array}[b]{c@{\quad}c}
        \Infer[]
        {\tyvarone \subtypeof \tyvarone}
        {} 
        &
        \Infer[]
        {\tyccon{D}{\itermone}\ \tpone_1 \cdots \tpone_n \subtypeof \tyccon{D}{\itermtwo}\ \tpone'_1 \cdots \tpone'_n}
        {\itermone \leqs \itermtwo & \tpone_i \subtypeof \tpone'_i\ (\text{$i$ in positive position}) & \tpone'_i \subtypeof \tpone_i\ (\text{$i$ in negative position}) }
      \end{array}
    \]
  \end{framed}
  \caption{Subtyping rules in the extended system.}\label{fig:typecheck:subtype'}
\end{figure}

%
In the following, we discuss our implementation, and then
consider some examples that highlight the strength and limitations of
our approach.
\subsection{Technical Overview on the Prototype}

Our tool \hosa\ is implemented in \tool{Haskell}.  Overall, the tool required just
a moderate implementation effort.  \hosa\ itself consists of
approximately 2.000 lines of code. Roughly half of this code is
dedicated to sized type inference, the other half is related to
auxiliary tasks such as parsing etc. Along with \hosa, we have written a
constraint solver, called \gubs. \gubs\ is also implemented
in \tool{Haskell} and weighs also in at around 2.000 lines of code.

In the following, we shortly outline the main execution stages of
\hosa. The overall process is also exemplified in \Cref{fig:hosa}
on the function $\fun{prependAll}$,
which prepends a given list to all elements of it second argument, 
itself a list of lists. This function is defined in terms of $\fun{map}$
and $\fun{append}$, see \Cref{fig:hosa:input} for the definition.

\newcommand{\ifun}[1]{\ifunone_{#1}}

\begin{figure}[!ht]
\begin{subfigure}[b]{\textwidth}
\begin{framed}
  \begin{minipage}{\linewidth}
\centering
\begin{lstlisting}[emph={i,j,k,l,ijk,ij,x,xs,f,ys},emph={[2] append,++,map,prependAll},style=haskell, style=numbered]
map :: $\forall \tyvarone.\ (\tylist{\tyvarone} \starr \tylist{\tyvarone}) \starr \tylist{(\tylist{\tyvarone})} \starr \tylist{(\tylist{\tyvarone})}$
map f [] = []
map f (x : xs) = f x : map f xs

append :: $\forall \tyvarone.\ \tylist{\tyvarone} \starr \tylist{\tyvarone} \starr \tylist{\tyvarone}$
append []       ys = ys
append (x : xs) ys = x : append xs ys

prependAll :: $\forall \tyvarone.\ \tylist{\tyvarone} \starr \tylist{(\tylist{\tyvarone})} \starr \tylist{(\tylist{\tyvarone})}$
prependAll xs = map (append xs)%\label{fig:hosa:input:prependAll}%
\end{lstlisting}
\end{minipage}
\end{framed}
\caption{Function $\fun{prependAll}$ and auxiliary definitions. Types have been specialised.}\label{fig:hosa:input}
\end{subfigure}
\\[5mm]
\begin{subfigure}[b]{\textwidth}
\begin{framed}
  \begin{minipage}{\linewidth}
\centering
\begin{lstlisting}[style=haskell,emph={i,j,k,l,ijk,ij},emph={[2] append,map,prependAll,f1,f3,f2,f5,f4,f6}]
map :: $\forall \tyvarone.\ \forall \ivarone\ivartwo\ivarthree.\ (\forall \ivarfour.\ \tylist[\ivarfour]{\tyvarone} \starr \tylist[\ifun{1}(\ivarfour,\ivarone)]{\tyvarone}) \starr \tylist[\ivarthree]{(\tylist[\ivartwo]{\tyvarone})} \starr \tylist[\ifun{3}(\ivarone,\ivartwo,\ivarthree)]{(\tylist[\ifun{2}(\ivarone,\ivartwo,\ivarthree)]{\tyvarone})}$
append :: $\forall \tyvarone.\ \forall \ivarone\ivartwo.\ \tylist[\ivarone]{\tyvarone} \starr \tylist[\ivartwo]{\tyvarone} \starr \tylist[\ifun{4}(\ivarone,\ivartwo)]{\tyvarone}$
prependAll :: $\forall \tyvarone.\ \forall \ivarone\ivartwo\ivarthree.\ \tylist[\ivarone]{\tyvarone} \starr \tylist[\ivarthree]{(\tylist[\ivartwo]{\tyvarone})} \starr \tylist[\ifun{6}(\ivarone,\ivartwo,\ivarthree)]{(\tylist[\ifun{5}(\ivarone,\ivartwo,\ivarthree)]{\tyvarone})}$
\end{lstlisting}
\end{minipage}
\end{framed}
\caption{Template sized types assigned by \hosa\ to the main function $\fun{prependAll}$ and auxiliary functions.}\label{fig:hosa:templates}
\end{subfigure}
\\[5mm]
\begin{subfigure}[b]{\linewidth}
\begin{framed}
  \begin{minipage}{\linewidth}
    \vspace{-1mm}
    \small
    \newcommand{\iso}[2]{\sivarone_{#1}}
    \begin{align*}
      \ifun{1}(\iso{15}{\ivarone},\ivartwo) & \leqc \ifun{1}(\ivarone,\iso{18}{\ivartwo}) & 
      \ifun{4}(\iso{7}{\ivartwo},\iso{10}{\ivarone}) & \leqc \ifun{1}(\ivarone,\iso{9}{\ivartwo}) &
      \iso{22}{} & \leqc \ifun{2}(\ivarone,\ivartwo,0) & 
      \ivarone & \leqc \iso{12}{\ivarone} & 
      \ivarone & \leqc \iso{8}{\ivarone} \\
      \iso{21}{\ivarone,\ivartwo,\ivarthree,\ivarfour} & \leqc \ifun{2}(\ivarone,\ivartwo,\ivarthree + 1) &
      \ifun{2}(\iso{18}{\ivarone}, \iso{20}{\ivartwo},\iso{19}{\ivarthree}) & \leqc \iso{21}{\ivarone,\ivartwo,\ivarthree,\ivarfour} & 
      0 & \leqc \ifun{3}(\ivarone,\ivartwo,0) & 
      \ivarone & \leqc \iso{13}{\ivarone} &
      \ivarone & \leqc \iso{10}{\ivarone} \\
      \iso{21}{\ivarone,\ivartwo,\ivarthree,\ivarfour} + 1 & \leqc \ifun{3}(\ivarone,\ivartwo,\ivarthree + 1) & 
      \ifun{3}(\iso{18}{\ivarone}, \iso{20}{\ivartwo},\iso{19}{\ivarthree}) & \leqc \iso{21}{\ivarone,\ivartwo,\ivarthree,\ivarfour} & 
      \ivartwo & \leqc \ifun{4}(0,\ivartwo) &
      \ivarone & \leqc \iso{17}{\ivarone} &
      \ivarone & \leqc \iso{11}{\ivarone} \\
      \iso{14}{\ivartwo,\ivarone} + 1 & \leqc \ifun{4}(\ivarone + 1,\ivartwo) & 
      \ifun{4}(\iso{15}{\ivartwo}, \ivarone) & \leqc \iso{16}{\ivarone,\ivartwo,\ivarthree,\ivarfour} & 
      \hspace{-5mm}\ifun{2}(\iso{9}{\ivarone},\iso{11}{\ivartwo},\iso{10}{\ivarthree}) & \leqc \ifun{5}(\ivarone,\ivartwo,\ivarthree) & 
      \ivarone & \leqc \iso{19}{\ivarone} \\
      \ifun{3}(\iso{9}{\ivarone},\iso{11}{\ivartwo},\iso{10}{\ivarthree}) & \leqc \ifun{6}(\ivarone,\ivartwo,\ivarthree) &
      \ivarone & \leqc \iso{7}{\ivarone} & 
      \ivarone & \leqc \iso{20}{\ivarone}
    \end{align*}
  \end{minipage}
\end{framed}
\caption{Second-order constraint system generated from \hosa.}\label{fig:hosa:constraints}
\end{subfigure}
\\[5mm]
\begin{subfigure}[b]{\linewidth}
\begin{framed}
\begin{minipage}{\linewidth}
    \vspace{-1mm}
    \small
    \newcommand{\iso}[2]{\ifun{#1}(#2)}
  \begin{align*}
    \ifun{1}(\ivarone,\ivartwo) & \defsym \ivarone + \ivartwo&    
    \ifun{2}(\ivarone,\ivartwo,\ivarthree) & \defsym \ivarone + \ivartwo&    
    \ifun{3}(\ivarone,\ivartwo,\ivarthree) & \defsym \ivarthree&    
    \ifun{4}(\ivarone,\ivartwo) & \defsym \ivarone + \ivartwo&
    \ifun{5}(\ivarone,\ivartwo,\ivarthree) & \defsym \ivarone + \ivartwo\\    
    \ifun{6}(\ivarone,\ivartwo,\ivarthree) & \defsym \ivarthree&
    \iso{7}{\ivarone} & \defsym\ivarone&
    \iso{8}{\ivarone} & \defsym\ivarone&
    \iso{9}{\ivarone} & \defsym\ivarone&
    \iso{10}{\ivarone} & \defsym\ivarone\\
    \iso{11}{\ivarone} & \defsym\ivarone&
    \iso{12}{\ivarone} & \defsym\ivarone&
    \iso{13}{\ivarone} & \defsym\ivarone&
    \iso{14}{\ivarone,\ivartwo} & \defsym\ivarone + \ivartwo &
    \iso{15}{\ivarone} & \defsym\ivarone\\
    \iso{16}{\ivarone,\ivartwo,\ivarthree} & \defsym\ivarone + \ivartwo&
    \iso{17}{\ivarone} & \defsym\ivarone&
    \iso{18}{\ivarone} & \defsym\ivarone&
    \iso{19}{\ivarone} & \defsym\ivarone&
    \iso{20}{\ivarone} & \defsym\ivarone\\
    \iso{21}{\ivarone,\ivartwo,\ivarthree,\ivarfour} & \defsym\ivarthree&
    \iso{22}{} & \defsym 0
  \end{align*}
\end{minipage}
\end{framed}
\caption{Model inferred by \gubs\ on the generated constraints.}\label{fig:hosa:model}
\end{subfigure}
\\[5mm]
\begin{subfigure}[b]{\linewidth}
\begin{framed}
\begin{minipage}{\linewidth}
\begin{lstlisting}[style=haskell,emph={i,j,k,l,ijk,ij},emph={[2] append,map,prependAll,f1,f3,f2,f5,f4,f6}]
map :: $\forall \tyvarone.\ \forall \ivarone\ivartwo\ivarthree.\ (\forall \ivarfour.\ \tylist[\ivarfour]{\tyvarone} \starr \tylist[\ivarfour + \ivarone]{\tyvarone}) \starr \tylist[\ivarthree]{(\tylist[\ivartwo]{\tyvarone})} \starr \tylist[\ivarthree]{(\tylist[\ivarone + \ivartwo]{\tyvarone})}$
append :: $\forall \tyvarone.\ \forall \ivarone\ivartwo.\ \tylist[\ivarone]{\tyvarone} \starr \tylist[\ivartwo]{\tyvarone} \starr \tylist[\ivarone + \ivartwo]{\tyvarone}$
prependAll :: $\forall \tyvarone.\ \forall \ivarone\ivartwo\ivarthree.\ \tylist[\ivarone]{\tyvarone} \starr \tylist[\ivarthree]{(\tylist[\ivartwo]{\tyvarone})} \starr \tylist[\ivarthree]{(\tylist[\ivarone + \ivartwo]{\tyvarone})}$
\end{lstlisting}
\end{minipage}
\end{framed}
\caption{Inferred size type obtained by instantiating the template types with the model computed by \gubs.}\label{fig:hosa:inferred}
\end{subfigure}
\\[5mm]

\caption{Sized type inference carried out by \hosa\ on $\fun{prependAll}$.}
\label{fig:hosa}
\end{figure}

\paragraph{Hindley-Milner Type Inference and Specialisation.}
As a first step, for each function in the given program a most general
polymorphic type is inferred. Should type inference fail, our
prototype will abort the analysis with a corresponding error message.
As shortly discussed in Section~\ref{sect:ERW}, it is not always
possible to decorate the most general type for higher-order
combinators, such as $\fun{foldr}$ or $\fun{map}$, with size
information. Indeed, in the example from
Figure~\ref{fig:doublefilter} on page~\pageref{fig:doublefilter}, we
have specialised the most general type of $\fun{foldr}$. 
Our implementation performs such a specialisation automatically.
Of course, types cannot be specialised arbitrarily. Rather, our
implementation computes for each higher-order combinator the least
general type that is still general enough to cover all calls to the
particular function.  Technically, this is achieved via
anti-unification and preserves well-typedness of the program.
Should specialisation still yield a type that is too general for size annotation, our tool
is also capable of duplicating the combinator, introducing a new
function per call-site. This will then allow size annotations suitable
for the particular call, at the expense of increased program size.
With respect to $\fun{prependAll}$, our implementation specialises 
the type of the supplied function in the declaration of $\fun{map}$
to match the call in Line~\ref{fig:hosa:input:prependAll} in \Cref{fig:hosa:input}.

\paragraph{Ticking.} 
By default, our tool will perform the ticking transformation from
Section~\ref{sect:TTTCA} on the program obtained in the previous step,
thereby enabling runtime analysis, the main motivating application
behind this work. For the sake of simplicity, ticking is not performed in the running example though.
\paragraph{Annotation of Types with Index Terms.}
To each function, an abstract, canonical sized type is then assigned
by annotating the types inferred in the second stage with
index terms. 
In essence, this is done by annotating polymorphic types $\forall \vec{\tyvarone}.\ \stone_1 \starr \cdots \starr \stone_k \starr \stone$
(where $\stone$ is not a functional type) as follows:
(i)~if the argument type $\stone_i$ is a data type then it is annotated with a fresh variable, the arguments are annotated recursively; 
(ii)~if the argument type $\stone_i$ is a functional type we proceed recursively, and close over all index variables occurring in the so obtained sized type, and
(iii)~we annotate the return type $\stone$ by an index term $\ifunone(\vec{\ivarone})$. Here, the
index symbol $\ifunone$ is supposed fresh. The variables $\vec{\ivarone}$ collect
on the one hand the free index variables occurring in argument types $\stone_i$. Moreover, 
for a functional type in argument position this sequence contains $m$ fresh index variables, for some fixed $m\in \N$, 
the \emph{extra variables}. 
Types of constructors $\conone$ are annotated similarly, except that the index $\ifunone(\vec{\ivarone})$ of its return type is fixed to $\sum_{\ivarone \in \vec{\ivarone}} \ivarone + w$
where $w = 0$ if $\conone$ is nullary, and $w = 1$ otherwise. 
The size index of constructors thus accounts for internal nodes in the tree representation of data values. 
This measure is seemingly ad-hoc, but turns out favourable, as 
the number of internal nodes relate to the number of recursive steps for functions defined by structural recursion. 
For instance, 
we have \hask![]!${} \decl \forall \tyvarone.\ \tylist[0]{\tyvarone}$, 
and (\hask!:!)${} \decl \forall \tyvarone.\ \forall \ivarone.\ \tyvarone \starr \tylist[\ivarone]{\tyvarone} \starr \tylist[\ivarone+1]{\tyvarone}$.

With respect to $\fun{prependAll}$ and $m=1$, annotated types are depicted in \Cref{fig:hosa:templates}.
Notice that the annotated type of $\fun{map}$ features the extra variable $\ivarone$.
Extra variables enable the system to deal
with \emph{closures}, i.e.\ functionals that capture part of the environment. 
Such a closure is for instance created with $\fun{append}\ xs$, on 
Line~\ref{fig:hosa:input:prependAll} in \Cref{fig:hosa:input}.
Intuitively, extra variables index the size of the captured environment.
We return to this point in a moment.
For all of the examples that we considered,
taking $m=1$, i.e.\ adding a single extra variable in step (ii)~above,
is sufficient.  It would be desirable to statically determine the
number of necessary extra variables.
This can likely be done with a simple form of data flow analysis, 
which is however beyond the scope of this work.
\paragraph{Constraint Generation.}
\hosa\ performs type inference as discussed in Section~\ref{sect:STI} based on the 
annotated types assigned in the previous step. The extension to the polymorphic type system 
with inductive data types poses no challenge. The extended subtyping rules from \Cref{fig:typecheck:subtype'}
are straight forward to integrate within the machinery discussed in Section~\ref{sect:STI}. 
It is also not difficult to adapt the rules \infervar\ and \inferfun\ from \Cref{fig:typeinfer:type}
so that type variables $\tyvarone$ in polymorphic types are properly instantiated:
suitable skeletons are already known at this stage, to turn them into suitable sized types our implementation decorates these 
with second-order index variables. For instance, 
suppose (\hask!:!)${} \decl \forall \tyvarone.\ \forall \ivarone.\ \tyvarone \starr \tylist[\ivarone]{\tyvarone} \starr \tylist[\ivarone+1]{\tyvarone}$
is used to construct a list of naturals. Then this constructor will be typed as 
(\hask!:!)${} \decl \tynat[\sivarone] \starr \tylist[\ivartwo]{\tynat[\sivarone]} \starr \tylist[\sivartwo+1]{\tynat[\sivarone]}$,
i.e., the type variable $\tyvarone$ has been instantiated with $\tynat[\sivarone]$ and the index variable $\ivarone$ with $\sivartwo$.
This stage will result 
in a SOCP, which is then translated to a FOCP by skolemisation.
On the function $\fun{prependAll}$, this results in 22 constraints, 
see \Cref{fig:hosa:constraints}.

\paragraph{Constraint Solving.}
\hosa\ makes use of the external tool \gubs\ to find a suitable model for the FOCP,
see \Cref{fig:hosa:constraints}. How this is done is explained in a moment. 
Note that the auxiliary functions $\ifun{7}$---$\ifun{22}$ 
were introduced by skolemization and correspond to $\sivarone_{7}$---$\sivarone_{22}$, respectively.

\paragraph{Concretising Annotated Types.}
In this final stage, \hosa\ combines the annotated types with the computed model, 
by unfolding index functions in template types
according to the model. The resulting sized types are decorated by arithmetical expressions only, 
compare \Cref{fig:hosa:inferred}.
This final result of the analysis is presented to the user.


Note that the extra variable $\ivarone$ in the annotated sized type of $\fun{map}$
is still present in its concrete size type. Indeed, the extra variable is crucial when 
we want to type the body of $\fun{prependAll}$. Here, 
we first derive
\begin{align*}
  & \typedsd[]{xs \oftype \tylist[\ivarone]{\tyvarone}}{\fun{map}}
  {(\forall \ivarfour.\ \tylist[\ivarfour]{\tyvarone} \starr \tylist[\ivarfour + \ivarone]{\tyvarone}) \starr \tylist[\ivarthree]{(\tylist[\ivartwo]{\tyvarone})} \starr \tylist[\ivarthree]{(\tylist[\ivarone + \ivartwo]{\tyvarone})}}
\text{, and} \\
  & \typedsd[]{xs \oftype \tylist[\ivarone]{\tyvarone}}{\fun{append}\ xs}
  {\tylist[\ivarfour]{\tyvarone} \starr \tylist[\ivarfour + \ivarone]{\tyvarone}}\tpkt
\end{align*}
Therefore, by rule $\checkappsd$ we get
\[\typedsd[]{xs \oftype \tylist[\ivarone]{\tyvarone}}{\fun{map}\ (\fun{append}\ xs)}{\tylist[\ivarthree]{(\tylist[\ivartwo]{\tyvarone})} \starr \tylist[\ivarthree]{(\tylist[\ivarone + \ivartwo]{\tyvarone})}}\tkom
\]
as demanded by the well-typedness of $\fun{appendAll}$. 
In a way, the extra variable in the declaration is used to keep track of the length of the list $xs$ captured by the term $\fun{append}\ xs$, 
which in turn, is relayed through the typing of $\fun{map}$ to the result type $\tylist[\ivarthree]{(\tylist[\ivarone + \ivartwo]{\tyvarone})}$.

\subsection{Constraint Solving}\label{sect:constraintsolving}
As for many sized type systems, constraint solving is also a central stage in our approach
and appears in the form of model-synthesis for FOCPs.
Strength and precision of the overall analysis is directly related to this stage.
Sized type inference is undecidable, as a consequence of Corollary~\ref{cor:soundcomplete}, 
model-synthesis is in general undecidable too. 

Synthesizing functions that obey certain set of constraints, as expressed for instance through FOCPs, 
is a fundamental task in program analysis. Consequently, the program verification community
introduced various techniques in this realm.
One popular approach relies on \emph{LP solvers}, compare e.g.,\ \citep{PR04}. 
This approach is effective, moreover, yields tight models. However, it is usually 
restricted to the synthesis of linear functions. This is often sufficient for termination analysis, 
where one is foremost interested that recursion parameters decrease. In our context however, 
this rules out the treatment of all programs that admit a non-linear runtime.
Another approach rests on solving (non-deterministic) \emph{recurrence relations}. To this end, 
dedicated tools like \tool{PUBS}~\citep{AAGP:SAS:08} have been developed, 
which are capable of synthesising non-linear functions. 
Recurrence relations are of limited scope in our context however. For instance, 
function composition cannot be directly expressed in this formalism.

To overcome these limitations, we have developed the \emph{\gubs\ upper bound solver} (\gubs\ for short), an
open source tool dedicated to the synthesis of models for FOCPs.
This tool is capable of synthesising models formed from linear and non-linear \emph{max-polynomials} over the naturals. 
\gubs\ itself is heavily inspired by methods developed in the context of rewriting. 
The rewriting community  pioneered the synthesising of polynomial interpretations, see e.g., \citep{FGMSTZ:SAT:07,FGMSTZ:RTA:08}
or the survey of \citet{Pechoux:TCS:13} on \emph{sup-interpretations}, a closely related topic. In this line of works, 
the problem is reduced to the satisfiability in the quantifier-free fragment of the theory of non-linear integer arithmetic. 
Dedicated to the latter, \tool{MiniSmt}~\citep{ZM:LPAR:10} has been developed. Moreover, 
state-of-the-art SMT solvers such as \tool{Z3}~\citep{MB:TACAS:08} can effectively treat 
quantifier free non-linear integer arithmetic nowadays.

The main novel aspect of \gubs\ is the modular approach it rest upon, which allowed us to integrate 
besides the aforementioned reduction various syntactic simplification techniques, and a per-SCC analysis.
In what follows, we provide a short outline of two central methods implemented in \gubs{}.

\paragraph{Synthesis of Models via SMT}

Conceptually, we follow the method presented by \citet{FGMSTZ:RTA:08}. 
In this approach, each $k$-ary symbol $\ifunone$ is associated with a $k$-ary \emph{template max-polynomial}, 
through a \emph{template interpretation} $\ainter$. Here, a template max-polynomial
is an expression formed from $k$ variables as well as undetermined coefficient variables $\vec{\coeffvarone}$, the \emph{template coefficients},
and the binary connectives $(+)$, $(\cdot)$ and $\max$, corresponding to addition, multiplication and the maximum function, respectively. 
For instance, a linear template for a binary symbol $\ifunone$ is
\[
  \interpret[\ainter]{\ifunone}(x,y) = \max(\coeffvarone_1 \cdot x + \coeffvarone_2 \cdot y + \coeffvarone_3, \coeffvartwo_1 \cdot x + \coeffvartwo_2 \cdot y + \coeffvartwo_3 ) \tpkt
\]
To find a concrete model for a SOCP $\constrone$ based on the template interpretation $\ainter$, \gubs\ 
is searching for concrete values $\vec{n} \in \N$ for the coefficient variables $\vec{\coeffvarone}$
so that
\[
  \forall {\itermone \leqc \itermtwo \in \constrone}.\ \forall \assignone : \IVARS \to \N.\ \interpret[\ainter]{\itermone} \leq \interpret[\ainter]{\itermtwo} \tkom
\] 
holds. 
Once these have been found, 
an interpretation $\iinter$ with $\fosat{\iinter}{\constrone}$ is obtained by substituting $\vec{n}$ for $\vec{\coeffvarone}$ in $\ainter$.
This search is performed itself in two steps.
First, the maximum operator is eliminated in accordance to the following two rules. Here, $C$ represents an arbitrary context over max-polynomials. 
\begin{align*}
  C[\max(\itermone_1,\itermone_2)] \leqc \itermtwo  & \Longrightarrow C[\itermone_1] \leqc \itermtwo \land C[\itermone_2] \leqc \itermtwo \tkom &
  \itermone \leqc C[\max(\itermtwo_1,\itermtwo_2)] & \Longrightarrow \itermone \leqc C[\itermtwo_1] \lor \itermone \leqc C[\itermtwo_2] \tpkt
\end{align*}
Intuitively, this elimination procedure is sound as we are dealing with weakly monotone expressions only. 
Once all occurrences of $\max$ are eliminated, the resulting formula is reduced to \emph{diophantine constraints} over the coefficient variables $\vec{\coeffvarone}$, 
via the so-called \emph{absolute positiveness check}, see also the work of \citet{FGMSTZ:SAT:07}. The diophantine constraints are then given to an SMT-solver that support quantifier-free non-linear integer arithmetic, 
from its assignment and the initially fixed templates \gubs\ then computes concrete interpretations. 
To get more precise bounds, \gubs\ minimises the obtained model by making use of the incremental features of current SMT-solvers, 
essentially by putting additional constraints on coefficients $\vec{\coeffvarone}$. 

The main limitation of this approach is that the shape of interpretations is fixed to that of templates, noteworthy, the degree of the interpretation is fixed in advance. 
As the complexity of the absolute positiveness check depends not only on the size of the given constraint system but to a significant extent also on the degree of interpretation functions, 
our implementation searches iteratively for interpretations of increasing degrees.
Also notice that our max-elimination procedure is incomplete, for instance, it cannot deal with the constraint $x + y \leqc \max(2x,2y)$, which is reduced to $x + y \leqc 2x \lor x + y \leqc 2y$. 
In contrast, \citet{FGMSTZ:RTA:08} propose a complete procedure to eliminate the maximum operator. However, our experimental assessment concluded that this encoding introduces too many auxiliary variables, 
which turned out as a significant bottleneck.  

\paragraph{Separate SCC Analysis}
Synthesis of models via SMT gets impractical on large constraint systems. To overcome this, \gubs\ divides the given constraint system $\constrone$ into 
its \emph{strongly connected components} (\emph{SCCs} for short) $\constrone_1,\dots,\constrone_n$, topologically sorted bottom-up, and finds a model for each SCC $\constrone_i$ iteratively. 
Here, the underlying \emph{call graph} is formed as follows. The nodes are given by the constraints in $\constrone$. 
Let $\itermone_1 \leqc \itermtwo_1$ to $\itermone_2 \leqc \itermtwo_2$ be two constraints in $\constrone$, 
where wlog. $\itermtwo_1 = C[\ifunone_1(\vec{\itermthree}_1),\dots,\ifunone_n(\vec{\itermthree}_n)]$ for a context $C$ without index symbols.
Then there is an edge from $\itermone_1 \leqc \itermtwo_1$ to $\itermone_2 \leqc \itermtwo_2$
if any of the symbols occurring in $\vec{\itermthree}_1,\dots,\vec{\itermthree}_n,\itermone_1$ occurs in $\itermtwo_2$.
The intuition is that once we have found a model for all the successors $\itermone_2 \leqc \itermtwo_2$ of $\itermone_1 \leqc \itermtwo_1$, we can interpret 
the arguments $\vec{\itermtwo_i}$ and the left-hand side $\itermone_1$ within this model. We can then extend this model by finding a suitable interpretation for 
$\ifunone_1,\dots,\ifunone_n$, thereby obtaining a model that satisfies $\itermone_1 \leqc \itermtwo_1$.

\subsection{Experimental Evaluation}\label{sect:experimental}
\begin{figure}[t]
  \centering
  \begin{framed}
    \begin{minipage}{1.0\linewidth}
\begin{lstlisting}[emph={f,g,z,x,xs,ys},emph={[2] id,comp,walk,reverse},style=haskell, style=numbered,%
                   literate={rev0}{{$\tfun{rev}{0}$}}4
                            {rev1}{{$\tfun{rev}{1}$}}4
                            {rev2}{{$\tfun{rev}{2}$}}4 
                            {reverse0}{{$\tfun{reverse}{0}$}}8
                            {reverse1}{{$\tfun{reverse}{1}$}}8
                            {nil0}{{$\tfun{Nil}{0}$}}4
                            {cons0}{{$\tfun{Cons}{0}$}}5
                            {T}{{$\ticksucc$}}1
                            {z0}{{$\varclockone_0$}}2
                            {z1}{{$\varclockone_1$}}2
                            {z2}{{$\varclockone_2$}}2
                            {z3}{{$\varclockone_3$}}2
                            {z4}{{$\varclockone_4$}}2
                            {z5}{{$\varclockone_5$}}2
                            {z6}{{$\varclockone_6$}}2
                            {x0}{{$\varone_0$}}2
                            {x1}{{$\varone_1$}}2
                            {x2}{{$\varone_2$}}2
                            {x3}{{$\varone_3$}}2
                            {x4}{{$\varone_4$}}2
                            {x5}{{$\varone_5$}}2
                            {x6}{{$\varone_6$}}2
]
rev0 :: $\tclock \starr \tpair{(\tyvarone \starr \tclock \starr \tpair{(\tylist{\tyvarone} \starr \tclock \starr \tpair{\tylist{\tyvarone}}{\tclock})}{\tclock})}{\tclock}$
rev0 z = (rev1,z)

rev1 :: $\tylist{\tyvarone} \starr \tclock \starr \tpair{(\tylist{\tyvarone} \starr \tclock \starr \tpair{\tylist{\tyvarone}}{\tclock})}{\tclock}$
rev1 xs z = (rev2 xs,z)

rev2 :: $\tylist{\tyvarone} \starr \tylist{\tyvarone} \starr \tclock \starr \tpair{\tylist{\tyvarone}}{\tclock}$
rev2 []       ys z = (ys,T z)
rev2 (x : xs) ys z = let (x1,z1) = rev0 z in
                     let (x2,z2) = x1 xs z1 in
                     let (x3,z3) = x2 (x : ys) z2 in (x3,T z3)

reverse0 :: $\tclock \starr \tpair{(\tylist{\tyvarone} \starr \tclock \starr \tpair{\tylist{\tyvarone}}{\tclock})}{\tclock}$
reverse0 z = (reverse1,z)

reverse1 :: $\tylist{\tyvarone} \starr \tclock \starr \tpair{\tylist{\tyvarone}}{\tclock}$
reverse1 xs z = let (x1,z1) = rev0 z in
                let (x2,z2) = x1 xs z1 in
                let (x3,z3) = x2 [] z4 in (x3,T z3)
\end{lstlisting}
    \end{minipage}
  \end{framed}
  \caption{Ticked reverse function.}\label{fig:reverse:ticked}
\end{figure}

We will now look at how \hosa\ deals with some examples, including those mentioned in the paper. 
Here, we also relate the strength and precision of tool to that of \hoca~\citep{ADM:ICFP:15} and \raml~\citep{HDW:POPL:17}. 
To the best of our knowledge, these constitute the only two state-of-the-art, freely available, tools for the automated resource analysis of higher-order programs.
\paragraph*{Tail-Recursive List Reversal.}
Reconsider the version of list reversal presented in Figure~\ref{fig:reverse} on page~\pageref{fig:reverse}.
This is an example that could not be handled by the original sized type system introduced by \citet{HPS:POPL:96}.
In Figure~\ref{fig:reverse:ticked} we show the corresponding ticked program. For brevity, the auxiliary 
definitions derived from the list constructors have been inlined.
Our tool infers
\[
\tfun{\fun{reverse}}{1} \decl \forall \tyvarone.\ \forall \ivarone\ivartwo.\tylist[\ivarone]{\tyvarone} \starr \tclock[\ivartwo] \starr \tpair{\tylist[\ivarone]{\tyvarone}}{\tclock[2 + \ivarone + \ivartwo]}\tpkt
\]
Thus, by setting the starting clock to zero, i.e.\ assuming $\ivartwo=0$, \hosa\ derives the bound $2 + \ivarone$ on the runtime of $\fun{reverse}$.
Taking into account that the auxiliary function $\fun{rev}$ performs $\ivarone + 1$ steps on a given list of length $\ivarone$, it is clear
that the derived runtime bound for $\fun{reverse}$ is tight. 
Similar, the derived bound for the size of the returned list is optimal. 
The optimal linear bound could also be found with \hoca\ and \raml. 

\paragraph*{Reverse with Difference Lists.}

\begin{figure}[t]
  \centering
  \begin{framed}
    \begin{minipage}{1.0\linewidth}
\begin{lstlisting}[emph={f,g,z,x,xs},emph={[2] id,comp,walk,reverse},style=haskell, style=numbered]
id :: $\tyvarone \starr \tyvarone$
id z = z

comp :: $(\tyvartwo \starr \tyvarthree) \starr (\tyvarone \starr \tyvartwo) \starr \tyvarone \starr \tyvarthree$
comp f g z = f (g z)

walk :: $\tylist[]{\tyvarone} \starr (\tylist[]{\tyvarone} \starr \tylist[]{\tyvarone})$
walk []     = id
walk (x:xs) = comp (walk xs) ((:) x)

reverse :: $\tylist[]{\tyvarone} \starr \tylist[]{\tyvarone}$
reverse xs = walk xs []
\end{lstlisting}
    \end{minipage}
  \end{framed}
  \caption{List reversal via difference lists. This is the motivating example from~\citet{ADM:ICFP:15}.}\label{fig:revdl}
\end{figure}

In \Cref{fig:revdl} we depict the motivating example from~\citet{ADM:ICFP:15}. Here, an alternative definition of list reversal based on 
\emph{difference lists}, a data structure for representing lists with a constant concatenation operation, is given.
In a functional setting, difference lists can be represented as functions $d \oftype \tylist{\tyvarone} \starr \tylist{\tyvarone}$, 
with $d$ denoting the list $ys$ such that $d \api xs = \fun{append} \api ys \api xs$. Difference lists are commonly used in functional 
programming in order to avoid the unnecessary runtime overhead in expressions such as $(\fun{append} \api (\fun{append} \api xs\ ys) \api zs)$.
On this example, \hosa\ succeeds with the following declaration
\[
  \tfun{\fun{reverse}}{1} \decl \forall \tyvarone.\ \forall \ivarone\ivartwo.\ \tylist[\ivarone]{\tyvarone} \starr \tclock[\ivartwo]{\tyvarone} \starr \tpair{\tylist[\ivarone]{\tyvarone}}{\tclock[3 + 2\cdot\ivarone + \ivartwo]}
  \tkom
\]
confirming that also this version of $\fun{reverse}$ exhibits a linear runtime complexity. An asymptotic linear bound can be derived by \hoca, 
but not by \raml. The latter can be rectified by using a contrived version $\fun{comp'}\api f \api x \api g \api y = f \api x \api (g \api y)$  of function composition
and suitably adapting the body of $\fun{walk}$. Then, \raml\ can infer the bound $3 + 9 \cdot \ivarone$ on the runtime of $\fun{reverse}$. 

\paragraph*{Product.}
Our tool infers
\[
  \tfun{\fun{product}}{2} \decl \forall \ivarone\ivartwo\ivarthree. \tylist[\ivarone]{\tyvarone} \starr \tylist[\ivartwo]{\tyvartwo} \starr \tclock[\ivarthree] \starr \tpair{(\tylist[\ivarone \cdot \ivartwo]{(\tpair{\tyvarone}{\tyvartwo})})}{\tclock[2+3\cdot\ivarone+2\cdot\ivarone\cdot\ivartwo+\ivarthree]}
  \tkom
\]
where $\tfun{\fun{product}}{2}$ corresponds to the ticked version of the function $\fun{product}$ from Figure~\ref{fig:doublefilter}.
The estimated size of the resulting list is precise, the computed runtime is asymptotically precise. Notice that 
the latter bound takes also the evaluation of the anonymous functions into account. 
An asymptotic precise bound can be inferred with \hoca, but not with \raml.

\paragraph*{Insertion Sort.}
\begin{figure}[t]
  \centering
  \begin{framed}
    \begin{minipage}{1.0\linewidth}
\begin{lstlisting}[emph={f,g,z,x,xs,ord},emph={[2] gt,insert,sort,sortNat},style=haskell, style=numbered, literate={NatType}{{$\tynat$}}{3}]
data NatType = Z | S NatType

gt :: $\tynat \starr \tynat \starr \tybool$
gt Z     y     = False
gt (S x) Z     = True
gt (S x) (S y) = gt x y

insert :: $\forall \tyvarone.\ (\tyvarone \starr \tyvarone \starr \tybool) \starr \tyvarone \starr \tylist{\tyvarone}  \starr \tylist{\tyvarone}$
insert ord x []       = x : []
insert ord x (y : ys) = if ord x y then y : insert ord x ys else x : y : ys


insertionSort :: $\forall \tyvarone.\ (\tyvarone \starr \tyvarone \starr \tybool) \starr \tylist{\tyvarone}  \starr \tylist{\tyvarone}$
insertionSort ord []       = []
insertionSort ord (x : xs) = insert ord x (sort ord xs)

sortNat :: $\tylist{\tyvarone}  \starr \tylist{\tyvarone}$
sortNat = insertionSort gt
\end{lstlisting}
    \end{minipage}
  \end{framed}
  \caption{Insertion-sort on natural numbers.}\label{fig:sort}
\end{figure}

In \Cref{fig:sort} we present a version of insertion sort that is parameterised by the comparison operation. 
We have then specialised this function to a comparison on natural numbers.
\hosa\ derives
\[
  \tfun{sortNat}{2} \decl \forall \tyvarone.\ \forall \ivarone\ivartwo\ivarthree.\ 
  \tylist[\ivarone]{\tynat[\ivartwo]} \starr \tclock[\ivarthree] \starr \tpair{\tylist[\ivarone]{\tynat[\ivartwo]}}{\tclock[2 + \ivarone^2\cdot\ivartwo + 2\cdot\ivarone^2 + \ivarthree]} \tpkt
\]
The computed runtime bound $2 + \ivarone^2\cdot\ivartwo + 2\cdot\ivarone^2$ is precise, taking into account
that $\fun{gt}$ is not a constant operation. It is worthy of note that the precise bound could only be 
inferred since \hosa\ is capable of inferring that $\fun{insert} \api ord \api x \api ys$, 
given $x \oftype \tynat[\ivarone]$ and $ys \oftype \tylist[\ivarthree]{\tynat[\ivartwo]}$ 
produces a list of type $\tylist[\ivarthree + 1]{\tynat[\max(\ivarone,\ivartwo)]}$. This demonstrates that
the limitation imposed by the linearity condition on canonical sized types can be overcome with the $\max$ 
operator.
\hoca\ and \raml\ can both give asymptotic precise bounds on this example. 
Concerning the former tool the bound $\mathsf{O}(\ivarone^3 + \ivartwo^3)$, concerning the latter 
a runtime bound $3  - 4 \cdot \ivarone \cdot \ivartwo + 4\cdot\ivarone^2\cdot\ivartwo + 8 \ivarone + 9 \ivarone^2$, is derived.

\paragraph*{Quicksort.}
We have also implemented a version of quicksort. This implementation uses the standard-combinator $\fun{partition}$,
to partition the given list into elements lesser and greater-equal to the pivot element, respectively. 
Our tool derives
\[
  \fun{partition} \decl \forall \tyvarone.\ \forall \ivarone.\ (\tyvarone \starr \tybool) \starr \tylist[\ivarone]{\tyvarone} \starr \tpair{\tylist[\ivarone]{\tyvarone}}{\tylist[\ivarone]{\tyvarone}} \tpkt
\]
This is indeed the most precise type that can be given to $\fun{partition}$ in our system. However, 
it is not precise enough to prove that quicksort runs in polynomial time. Here, one would need to prove 
that the length of the two resulting lists sum up to the length of the argument list. 
On the other hand, both \raml\ and \hoca\ can prove a quadratic bound on the runtime of quicksort. 


\paragraph*{Functional Queues.}
\begin{figure}[t]
  \centering
  \begin{framed}
    \begin{minipage}{1.0\linewidth}
\begin{lstlisting}[emph={f,g,z,f,f',e,r,x,xs,ord},emph={[2] repair snoc fromList error},style=haskell, style=numbered, literate={NatType}{{$\tynat$}}{3}]
data Queue $\tyvarone$ =  Q [$\tyvarone$] [$\tyvarone$]

repair :: $\forall \tyvarone.\ \tycon{Queue}\ \tyvarone \starr \tycon{Queue}\ \tyvarone$
repair (Q [] r)       = Q (reverse r) []
repair (Q (e : f) r) = Q (e : f) r

push :: $\forall \tyvarone.\ \tyvarone \starr \tycon{Queue}\ \tyvarone \starr \tycon{Queue}\ \tyvarone$
push x (Q f r) = repair (Q f (x : r))

pop :: $\forall \tyvarone.\ \tycon{Queue}\ \tyvarone \starr \tpair{\tyvarone}{\tycon{Queue}\ \tyvarone}$
pop (Q [] r)       = error               -- queue empty
pop (Q (e : f) r) = (e, repair (Q f r))

fromList :: $\forall \tyvarone.\ [\tyvarone] \starr \tycon{Queue}\ \tyvarone$
fromList = foldr push (Q [] [])
\end{lstlisting}
    \end{minipage}
  \end{framed}
  \caption{Functional queues.}\label{fig:queue}
\end{figure}

In \Cref{fig:queue} we give an implementation of queues as defined by \citet{Okasaki:99}.
A value $\con{Q}\ f\ r$ represents the queue with initial segment $f$ and reversed remainder $r$. 
Enqueueing thus simply amounts to consing it to $r$, 
whereas dequeuing an element amounts to removing the head of $f$, whenever $f$ is non-empty.
The latter is ensured by the auxiliary function $\fun{repair}$, which is called whenever the queue is modified.
Notice thus that both adding and removing an element from a queue has a linear worst case complexity, 
due to the call to $\fun{reverse}$ in the definition of $\fun{repair}$. However, this cost 
armortises with the number of pushes. 
Our system derives
\[
  \tfun{fromList}{1} \decl \forall \tyvarone.\ \forall \ivarone\ivartwo.\ 
  \tylist[\ivarone]{\tyvarone} \starr \tclock[\ivartwo] \starr 
  \tpair{\tyccon{Queue}{1+\ivarone}\ \tyvarone}{\tclock[2 + \ivarone + 5\cdot\ivarone^2 + \ivartwo]} \tkom
\]
and thus a quadratic runtime bound on $\fun{fromList}$. 
In contrast, both \hoca\ and \raml\ derive an asymptotic precise linear bound. 
Concerning \raml, this is possible because of the underlying amortised analysis. 
\hoca\ derives the precise bound for a different reason: 
$\fun{fromList}$ is translated into two simple recursive definition, 
that turn a list $[x_1,\dots,x_n]$ directly into $\con{Q}\ [x_1]\ [x_n,\dots,x_2]$, 
thereby in particular completely eliminating the problematic calls to $\fun{reverse}$ via $\fun{repair}$. 

\paragraph*{Prepend All.}
Concerning the function $\fun{prependAll}$ from \Cref{fig:hosa:input}, \hosa\ infers
\[
\tfun{\fun{prependAll}}{2} \decl \forall \tyvarone.\ \forall \ivarone\ivartwo\ivarthree.\ \tylist[\ivarone]{\tyvarone} \starr \tylist[\ivarthree]{(\tylist[\ivartwo]{\tyvarone})} \starr \tclock[\ivarfour] \starr \tpair{\tylist[\ivarthree]{\tylist[1 + \ivarone + \ivartwo]{\tyvarone}}}{\tclock[2 + \ivarone\cdot\ivarthree + 2\cdot\ivarthree + \ivarfour]}\tpkt
\]
The runtime of $\fun{prependAll}$ is thus correctly bounded by $2 + \ivarone\cdot\ivarthree + 2\cdot\ivarthree$. 
Evaluating $\fun{prependAll}$ results in $(1 + \ivartwo)$ calls to $\fun{map}$, counting the base case and $\ivartwo$ recursive calls.
Each recursive call triggers the evaluation to $\fun{append}$, itself performing $1 + \ivarone$ reduction steps. 
Taking into account that $\fun{prependAll}$ has to be unfolded first, we see that  
the inferred bound is indeed optimal. 

Worthy of note, the example can also be handled by \hoca. However, \hoca\ is only able to infer an asymptotic
\emph{quadratic} bound. On the other hand, whereas \raml\ can produce asymptotic precise bound for $\fun{append}$ and $\fun{map}$, 
it fails to analyse $\fun{prependAll}$ itself. \raml\ does not attribute potentials to functions, 
thus, it is assumed that the reduction of closures can be solely measured in terms of the formal parameter, but is independent from the captured environment. 
The compositional nature of the analysis underlying \raml\ comes at a price. 




\section{Conclusions}\label{sect:C}
We have described a new system of sized types whose key features are
an abstract index language, and higher-rank index polymorphism.  This
allows for some more flexibility compared to similar type systems from
the literature. The introduced type system is proved to enjoy a form
of type soundness, and to support a relatively complete type inference
procedure, which has been implemented in our prototype tool
\tool{HoSA}.
One key motivation behind this work is achieving a form of modular
complexity analysis without sacrificing its expressive power.  This is
achieved by the adoption of a type system, which is modular and
composable by definition. This is contrast to other methodologies like
program transformations~\citep{ADM:ICFP:15}.  Noteworthy, modularity
carries to some extent through to constraint solving.  The SCCs in the
generated constraint problem are in correspondence with the SCC of the
call-graph in the input program, and are analysed independently.

Future work definitely includes refinements to our constraint solver
\gubs.\@ It would also be interesting to see how our overall
methodology applies to different resource measures like heap size
etc. Concerning heap size analysis, this is possible by ticking
constructor allocations. It could also be worthwhile to integrate a form of
amortisation in our system. 




\bibliographystyle{plainnat}
\bibliography{references}

\end{document}